\documentclass[12pt,draftcls,peerreviewca,onecolumn]{IEEEtran}
\usepackage{etoolbox}
\makeatletter
\patchcmd{\@makecaption}
  {\scshape}
  {}
  {}
  {}
\makeatletter
\patchcmd{\@makecaption}
  {\\}
  {.\ }
  {}
  {}
\makeatother
\usepackage{amsfonts}
\usepackage{mathrsfs}
\usepackage{amsfonts}
\usepackage{amssymb}
\usepackage{graphicx}
\usepackage{epsfig}
\usepackage{psfrag}
\usepackage{amsmath}
\usepackage{array}
\usepackage{cases}
\usepackage{subfigure}
\usepackage{cite,graphicx,amsmath,amssymb,color}
\usepackage{anysize}
\usepackage{algorithmic}
\usepackage{algorithm}

\newcommand{\dis}{\stackrel{d}{\sim}}
\newcommand{\eqla}{\stackrel{(a)}{=}}
\newcommand{\eqlb}{\stackrel{(b)}{=}}
\newcommand{\eqlc}{\stackrel{(c)}{=}}
\newcommand{\eqld}{\stackrel{(d)}{=}}
\newcommand{\eqle}{\stackrel{(e)}{=}}

\newtheorem{Thm}{Theorem}
\newtheorem{Lem}{Lemma}

\newtheorem{Prob}{Problem}

\IEEEoverridecommandlockouts

\begin{document}

\title{Analysis and Optimization of Caching and Multicasting  in Large-Scale  Cache-Enabled Heterogeneous  Wireless Networks}
%\author{\authorblockN{Ying Cui}\authorblockA{Shanghai Jiao Tong University}
%\and \authorblockN{Dongdong Jiang} \authorblockA{Shanghai Jiao Tong University}
%}
\author{Ying Cui\thanks{Y. Cui and D. Jiang are with the Department of  Electronic Engineering, Shanghai Jiao Tong University, China.},  {\em MIEEE}, \and Dongdong Jiang, {\em StMIEEE}}

\maketitle

%The proposed scheme for the cached-enabled HetNet (with backhual links) improves the performance of the scheme for the cached-enabled single-tier network (without backhaul links) in [arxiv xx].

\begin{abstract}
Heterogeneous wireless networks (HetNets) provide a powerful approach to meet the dramatic mobile traffic  growth, but also impose a significant challenge on backhaul. Caching and multicasting at macro and pico base stations (BSs) are two promising methods to support massive content delivery  and reduce backhaul load in HetNets.   In this paper, we jointly consider caching and multicasting  in a large-scale  cache-enabled HetNet with backhaul constraints. We propose a  hybrid caching design consisting of identical caching in the macro-tier and random caching in the pico-tier, and a corresponding multicasting design. By carefully handling different types of interferers and adopting appropriate approximations, we  derive  tractable expressions for  the successful transmission probability in the general region as well as  the high signal-to-noise ratio (SNR) and user density  region, utilizing tools from stochastic geometry. Then, we consider the successful transmission probability maximization by optimizing the design parameters, which is a very challenging mixed discrete-continuous optimization problem. By using optimization techniques and exploring  the structural properties,  we  obtain a near optimal solution with superior performance and manageable  complexity.  This  solution achieves better performance in the general region than any asymptotically optimal solution, under a mild condition. The analysis and optimization results provide valuable design insights for practical cache-enabled HetNets.
\end{abstract}

\begin{keywords}
Cache, multicast, backhaul, stochastic geometry, optimization, heterogenous wireless network
\end{keywords}

\newpage
\section{Introduction}

The rapid proliferation of smart mobile devices has triggered
an unprecedented growth of the global mobile data traffic.
HetNets have been proposed as an effective way
to meet the dramatic traffic growth by deploying short range small-BSs together with traditional
macro-BSs, to provide better time or frequency reuse\cite{HoadleyWC12}. However, this approach imposes a significant challenge of providing expensive high-speed backhaul links for connecting all the small-BSs to the core network\cite{AndrewsComMag13}.

Caching at small-BSs is a promising  approach to alleviate the backhaul capacity
requirement in HetNets\cite{CMag14Chen,Sarkissian4GLTE11,cachingmimoLiu15}. % and reduce the distance between popular contents and their requesters.
Many existing
works have focused on optimal cache placement at small-BSs, which is of critical importance in cache-enabled HetNets.
For example, in \cite{Shanmugam13}  and \cite{LiTWC15}, the authors  consider the optimal content placement at small-BSs to minimize the expected downloading time for files  in a single macro-cell  with multiple small-cells.  %where the coverage areas of small-cells are overlapping.
File requests which cannot be satisfied locally at a small-BS are served by the macro-BS.
%In \cite{PengGC15}, the authors consider the optimal content placement at BSs to minimize the expected downloading time for files in a (single-tier) multi-cell network.
The optimization problems  in \cite{Shanmugam13} and \cite{LiTWC15} are NP-hard,  and low-complexity solutions are  proposed.  In  \cite{BioglioGC15}, the authors propose a caching design based on file splitting and MDS encoding in a single macro-cell  with multiple small-cells. File requests which cannot be satisfied locally at a small-BS are served by the macro-BS, and backhaul rate analysis and optimization are considered. Note that the focuses of \cite{Shanmugam13,LiTWC15,BioglioGC15} are on performance optimization of caching design.

In \cite{EURASIP15Debbah,LiuYangICC16,Yang16}, the authors consider caching the most popular files at each small-BS in large-scale cache-enabled small-cell networks or HetNets,  with backhaul constraints. The service rates of uncached files  are   limited  by the backhaul capacity. In \cite{QuekTWC16}, the authors propose  a partion-based combined caching design in a large-scale cluster-centric small-cell network, without considering backhaul constraints. In \cite{TamoorComLett16}, the authors consider two caching designs, i.e., caching the most popular files and random caching of a uniform distribution, at small-BSs in a  large-scale cache-enabled HetNet, without backhaul constraints. File requests which cannot be satisfied  at a small-BS are served by macro-BSs.
In \cite{DBLP:journals/corr/Tamoor-ul-Hassan15}, the authors consider random caching of a uniform distribution in a large-scale cache-enabled small-cell network, without backhaul constraints, assuming that content  requests follow  a uniform distribution. Note that the focuses of \cite{EURASIP15Debbah,LiuYangICC16,Yang16,QuekTWC16,TamoorComLett16,DBLP:journals/corr/Tamoor-ul-Hassan15} are on performance analysis of  caching designs.

On the other hand, enabling multicast service at BSs in HetNets is  an efficient way to deliver popular contents to multiple
requesters simultaneously, by effectively utilizing the broadcast nature of the wireless medium\cite{eMBMS}.
In \cite{WCNC14Tassiulas} and \cite{ZhouTWC15}, the authors consider a single macro-cell  with multiple small-cells with backhaul costs. Specifically, in \cite{WCNC14Tassiulas}, the optimization of caching and multicasting, which is NP-hard,  is considered, and a simplified solution with approximation guarantee  is proposed.  In \cite{ZhouTWC15},   the optimization of  dynamic multicast scheduling for a given content placement, which is a dynamic programming problem, is considered,  and a low-complexity optimal numerical solution is obtained.%, by exploiting structural properties.
%Note that the main focuses of  \cite{WCNC14Tassiulas} and \cite{ZhouTWC15} are on performance optimization.

The network models considered in \cite{Shanmugam13,LiTWC15,BioglioGC15,WCNC14Tassiulas,ZhouTWC15}   do not capture  the stochastic natures of  channel fading  and geographic locations of  BSs and users.
 The network models considered in \cite{EURASIP15Debbah,LiuYangICC16,Yang16,QuekTWC16,TamoorComLett16,DBLP:journals/corr/Tamoor-ul-Hassan15} are more realistic and can reflect  the stochastic natures of signal and interference.  However, the simple identical caching design considered in \cite{EURASIP15Debbah,LiuYangICC16,Yang16,TamoorComLett16} does not provide spatial file diversity; the combined caching design in \cite{QuekTWC16}  does not reflect the popularity differences of files in each of the three categories; and the random caching design of a uniform distribution in \cite{TamoorComLett16,DBLP:journals/corr/Tamoor-ul-Hassan15} cannot make use of popularity information.  Hence, the caching designs in \cite{EURASIP15Debbah,LiuYangICC16,Yang16,QuekTWC16,TamoorComLett16,DBLP:journals/corr/Tamoor-ul-Hassan15} may not lead to good network performance.
 On the other hand, \cite{ICC15Giovanidis,DBLP:journals/corr/BharathN15,Altman13,arXivSGCaching15} consider analysis and optimization of caching in large-scale cache-enabled single-tier networks. Specifically,  \cite{ICC15Giovanidis}  considers  random caching at BSs, and analyze and optimize  the hit probability. Reference \cite{DBLP:journals/corr/BharathN15} considers random caching with contents being stored at each BS in an i.i.d. manner,  and analyzes the minimum offloading loss.
 In \cite{Altman13}, the authors  study   the expected costs of obtaining a complete content under random uncoded caching and coded caching strategies, which are designed only for different pieces of a single content.
In \cite{arXivSGCaching15}, the authors consider analysis and optimization of joint caching and multicasting. However, the proposed caching and multicasting designs in \cite{ICC15Giovanidis,DBLP:journals/corr/BharathN15,Altman13,arXivSGCaching15} may not be applicable to HetNets with backhaul constraints.
In summary,
to facilitate  designs of practical cache-enabled HetNets for massive content dissemination, further studies are required  to understand the following key questions.

$\bullet$
How do  physical layer and content-related parameters fundamentally affect  performance of cache-enabled HetNets?

 $\bullet$  How can  caching and multicasting jointly and optimally assist  massive content dissemination in cache-enabled HetNets?

 In this paper, we consider  the analysis and optimization of joint caching and multicasting to improve the efficiency of  massive content dissemination  in a large-scale  cache-enabled HetNet with backhaul constraints.
Our main contributions are summarized below.

$\bullet$ First, we propose a  hybrid caching design with certain design parameters, consisting of identical caching in the macro-tier and random caching in the pico-tier, which can provide spatial file diversity. We propose a corresponding  multicasting  design for efficient content dissemination by  exploiting  broadcast nature of the wireless medium.

$\bullet$ Then, by carefully handling different types of interferers and adopting appropriate approximations, we  derive  tractable expressions for  the successful transmission probability in the general region and the asymptotic region, utilizing tools from stochastic geometry. These expressions reveal the impacts of physical layer and content-related parameters on the successful transmission probability.

$\bullet$ Next, we consider the successful transmission probability maximization by optimizing the design parameters, which is a very challenging mixed discrete-continuous optimization problem. We propose a two-step optimization framework to obtain a near optimal solution with superior performance and manageable   complexity. Specifically, we first characterize the structural properties of the asymptotically optimal solutions. Then, based on these properties, we obtain the near optimal solution, which achieves better performance in the general region than any asymptotically optimal solution, under a mild condition.

$\bullet$ Finally, by numerical simulations, we  show that the near optimal solution achieves a significant gain in successful transmission probability over some baseline schemes.

\section{Network Model}\label{sec:netmodel}

We consider a two-tier HetNet where a macro-cell tier is overlaid with a pico-cell tier,  as shown in Fig.~\ref{fig:system}. The locations of the macro-BSs and the pico-BSs are spatially distributed as two independent homogeneous Poisson point processes (PPPs) $\Phi_{1}$ and $\Phi_{2}$ with densities $\lambda_{1}$ and $\lambda_{2}$, respectively, where $\lambda_1<\lambda_2$. The locations of the users are also distributed as an independent homogeneous PPP $\Phi_{u}$ with density $\lambda_{u}$. We refer to the macro-cell tier and the pico-cell tier as the $1$st tier and  the $2$nd tier, respectively. Consider the downlink scenario. Each BS in the $j$th tier  has one transmit antenna with   transmission power $P_j$ ($j=1,2$), where $P_1>P_2$.  Each user has one receive antenna. All BSs are operating on the same frequency band of total bandwidth  $W$ (Hz). Consider a discrete-time system with time being slotted and study one slot of the network. We consider both large-scale fading and small-scale fading. Due to large-scale fading, a transmitted signal from the  $j$th tier with distance $D$ is attenuated by a factor $\frac{1}{D^{\alpha_{j}}}$, where $\alpha_{j}>2$ is the path loss exponent of the $j$th tier. For small-scale fading, we assume Rayleigh fading channels\cite{Andrews11,WCOM13Andrews}.
%\footnote{Rayleigh fading is a widely used probabilistic channel model in existing literature\cite{Andrews11,WCOM13Andrews}.
%%It is based on the assumption that there are a large number of statistically independent reflected and scattered paths with random amplitudes and random phases uniformly distributed in $[0,2\pi]$, and {\em Central Limit Theorem}\cite{Tse:2005:FWC}.
%}%, i.e., each small-scale channel follows $\mathcal {CN}(0, 1)$.

Let $\mathcal N\triangleq \{1,2,\cdots, N\}$ denote the set of $N$ files (e.g., data objects or chucks of data objects) in the HetNet.
For ease of illustration, assume that all  files  have the same size.\footnote{Files  of different sizes can be divided into chunks  of the same  length. Thus, the results in this paper can be  extended to the case of different file sizes.}
Each file is of certain popularity, which is assumed to be identical among all users.  Each user randomly  requests one file, which is file $n\in \mathcal N$ with probability $a_n\in (0,1)$, where $\sum_{n\in \mathcal N}a_n=1$.  Thus, the file popularity distribution is given by $\mathbf a\triangleq (a_n)_{n\in \mathcal N }$, which is assumed to be known  apriori.  In addition, without loss of generality (w.l.o.g.),  assume  $a_{1}> a_{2}>\ldots> a_{N}$.%, i.e., the popularity rank of file $n$ is $n$.

 The  HetNet  consists of  cache-enabled macro-BSs and pico-BSs. Each BS in the $j$th tier is equipped with a cache of size $K_j^c< N$ to store different files. Assume $K_1^c+K_2^c\leq N$.
 Each macro-BS is connected to the core network via a wireline backhaul link of transmission capacity $K_1^b< N$ (files/slot), i.e., each macro-BS can retrieve at most $K_1^b$ different files  from the core network in each slot.\footnote{Note that  storing or retrieving more than one copies of the same file at one BS is redundant and  will waste storage or backhaul resources.}
  %There are no backhaul links connecting pico-BSs to the core network.
Note that $K_1^c$, $K_2^c$ and $K_1^b$ reflect the storage and backhaul  resources in the cache-enabled HetNet.

\begin{figure}[t]
\begin{center}
 \includegraphics[width=8cm]{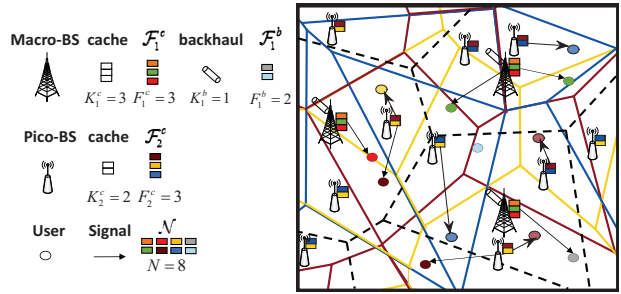}
  \end{center}
     \caption{\small{Network model.
%     There are eight files in the HetNet, represented by eight different colors, respectively. The color of each user indicates the file requested by the user.The colors of the cache at each BS indicate the  files stored at the BS.
The 1st tier corresponds to a Voronoi tessellation (cf. black dashed line segments), determined  by the locations of all the macro-BSs.   Each file $n\in \mathcal F_2^c$ corresponds to a Voronoi tessellation (cf. solid line segments in the same color as the file),  determined  by the locations of all the pico-BSs storing this file.}
 }
\label{fig:system}
\end{figure}

\section{Joint Caching and Multicasting}\label{sec:description}

We are interested in the case where the storage and backhaul resources are limited, and may not be able to satisfy all file requests. In this section, we propose a joint caching and multicasting design with certain design parameters,  which can provide high spatial file diversity and ensure efficient content dissemination.
%In later sections, we shall analyze the network performance under this joint design for any  given design parameters, and optimize the design parameters to maximize the network performance.

\subsection{Hybrid Caching}

To provide high spatial file diversity, we  propose a {\em hybrid caching design} consisting of  identical caching in the 1st tier and random caching in the 2nd tier, as illustrated  in Fig.~\ref{fig:system}.
%\footnote{Later, we can see that our caching design can provide high spatial file diversity \cite{arXivSGCaching15} (in the 2nd tier) and ensure each user to  obtain its desired file from a nearby BS with high probability.}
Let $\mathcal F_j^c\subseteq \mathcal N$ denote
 the set of $F_j^c\triangleq|\mathcal F_j^c|$ files cached in the $j$th tier. Specifically,
our hybrid caching design satisfies the following requirements: (i) {\em non-overlapping  caching across tiers}: each file is stored in at most one tier; (ii) {\em identical caching in the 1st tier}: each macro-BS stores the same set $\mathcal F_1^c$ of $K_1^c$ (different) files; and (iii) {\em random caching in  the 2nd tier}:  each pico-BS randomly stores $K_2^c$ different files out of all files in $\mathcal F_2^c$, forming a subset of $\mathcal F_2^c$.
Thus, we have the following constraint:
\begin{align}
 \mathcal F_1^c, \mathcal F_2^c\subseteq \mathcal N,\ \mathcal F_1^c\cap\mathcal F_2^c=\emptyset, \ F_1^c= K_1^c, \ F_2^c\geq K_2^c.\label{eqn:cache-constr}
 \end{align}

To further illustrate the random caching in the 2nd tier, we first introduce some notations. We say every $K_2^c$ different files in $\mathcal F_2^c $ form a combination. Thus, there are $I\triangleq \binom{F_2^c}{K_2^c}$ different combinations   in total. Let $\mathcal I\triangleq \{1,2,\cdots, I\}$ denote the set of $I$ combinations. Combination $i\in \mathcal I$ can be characterized  by an $F_2^c$-dimensional vector   $\mathbf x_i\triangleq (x_{i,n})_{n\in \mathcal F_2^c}$, where $x_{i,n}=1$ indicates that file $n\in \mathcal F_2^c$ is included in combination $i$ and $x_{i,n}=0$ otherwise. Note that there are  $K_2^c$ 1's  in  each $\mathbf x_i$. %$\sum_{n\in \mathcal N}x_{i,n}=K$ for all $i\in \mathcal I$.
 Denote $\mathcal N_i\triangleq \{n\in \mathcal F_2^c:x_{i,n}=1\}\subseteq \mathcal F_2^c$ as the set of $K_2^c$   files contained in combination $i$.
Each pico-BS stores one combination  at random, which is combination $i\in \mathcal I$ with probability $p_i$ satisfying:\footnote{In this paper, to understand the natures of joint caching and multicasting in cache-enabled HetNets, we shall first pose the analysis and optimization on the basis of all the file combinations in $\mathcal I$ (for the 2nd tier).  Then, based on the insights obtained, we shall focus on reducing  complexity while maintaining superior performance.}
\begin{align}
&0\leq p_i\leq1,\ i\in \mathcal I,\label{eqn:cache-constr-indiv}\\
&\sum_{i\in \mathcal I}p_i=1.\label{eqn:cache-constr-sum}
\end{align}
Denote $\mathbf p\triangleq (p_i)_{ i\in \mathcal I}$.
To facilitate the analysis in later sections, based on $\mathbf p$, we also define the probability that file $n\in\mathcal F_2^c$ is stored at a pico-BS, i.e.,
\begin{align}
T_n\triangleq \sum_{i\in\mathcal{I}_{n}}p_{i}, \  n\in\mathcal F_2^c,\label{eqn:def-T-n}
\end{align}
where $\mathcal I_n\triangleq \{i\in \mathcal I:x_{i,n}=1\}$ denotes the set of $I_n\triangleq \binom{F_2^c-1}{K_2^c-1} $ combinations  containing file $n\in\mathcal F_2^c$.
Denote $\mathbf T\triangleq (T_n)_{n\in\mathcal F_2^c}$. %Note that $T_n=p_n$ when $K=1$.
Note that $\mathbf p$ and $\mathbf T$ depend on $\mathcal F_2^c$. Thus, in this paper, we use $\mathbf p(\mathcal F_2^c)$ and $\mathbf T(\mathcal F_2^c)$  when  emphasizing this relation.
Therefore, the hybrid caching design in the cache-enabled HetNet is specified by the design parameters $\left(\mathcal F_1^c,\mathcal F_2^c,\mathbf p\right)$.

To efficiently utilize backhaul links and  ensure high spatial file diversity, we only retrieve files not stored in the cache-enabled HetNet via  backhaul links. Let $\mathcal F_1^b\subseteq \mathcal N$ denote the set of  $F_1^b\triangleq |\mathcal F_1^b|$ files which can be  retrieved by each macro-BS from the core network. Thus, we have the following constraint:
\begin{align}
\mathcal F_1^b=\mathcal N\setminus (\mathcal F_1^c\cup \mathcal F_2^c).\label{eqn:backhaul}
\end{align}
Therefore, the  file distribution in the cache-enabled HetNet is fully specified by the hybrid caching design $\left(\mathcal F_1^c,\mathcal F_2^c,\mathbf p\right)$.
%Each macro-BS is connected to the core network via a (wireline) backhaul link with transmission capacity $K_1^b$ (files/slot), i.e., each macro-BS can retrieve at most $K_1^b$ different uncached files (in  $\mathcal F_1^b$)   from the core network in each slot.

\subsection{Multicasting}

In this part, we propose a multicasting design associated with the hybrid caching design $\left(\mathcal F_1^c,\mathcal F_2^c,\mathbf p\right)$.  First, we introduce the user association under the proposed hybrid caching design.   In the cache-enabled HetNet, a user  accesses to a tier based on its desired file. Specifically, each user requesting  file $n\in \mathcal F_1^c\cup \mathcal F_1^b$ is associated with the nearest macro-BS and is referred to as a macro-user. While, each user requesting  file $n\in \mathcal F_2^c$ is associated with the nearest pico-BS storing a combination $i\in \mathcal I_n$ (containing file $n$) and is referred to as a pico-user. The associated BS of each user is called its serving BS, and offers the maximum long-term average receive power for its  desired file\cite{WCOM13Andrews}.
Note that under the proposed hybrid caching design $\left(\mathcal F_1^c,\mathcal F_2^c,\mathbf p\right)$, the serving BS of a macro-user is its nearest macro-BS, while the serving BS of a pico-user (affected by $\mathbf p$) may not be
 its  geographically nearest BS. We refer to this association mechanism as the {\em content-centric association} in the cached-enabled HetNet, which is  different from the traditional {\em connection-based association} \cite{WCOM13Andrews} in HetNets.
%this association jointly considers the physical layer  and content-centric properties.

Now, we introduce  file scheduling in the cache-enabled HetNet. Each   BS will serve all the cached files requested by its associated users.
Each macro-BS only serves   at most $K_1^b$  uncached  files requested by its associated users, due to the backhaul constraint for retrieving uncached files. In particular, if the users of a macro-BS request  smaller than or equal to $K_1^b$ different uncached files, the macro-BS serves all of them; if the users of a macro-BS request greater than $K_1^b$ different uncached files, the macro-BS will randomly select  $K_1^b$ different requested uncached files to serve, out of all the requested uncached  files   according to the uniform distribution.

We consider multicasting\footnote{Note that in this paper, the multicast service happens once every slot, and hence no additional delay is introduced.} in the cache-enabled  HetNet for efficient content dissemination. Suppose a BS schedules to serve requests for $k$ different files. Then, it transmits each of the $k$  files  at rate
%\footnote{Note that in this paper, for the tractability of the analysis, we assume that all the files are delivered at the same rate $\tau$ to obtain first-order design insights. Under this assumption, there is no intricate tradeoff between unicast and multicast. We would like to consider content transmissions at multiple bit rates and multicast grouping  in future work.}
$\tau $ (bit/second) and over $\frac{1}{k}$  of total bandwidth W using FDMA.  All the users which request one of the $k$ files from this BS try to decode the file from the single multicast transmission of the file at the BS.
%Consider a macro-BS which has $k_1^c\in\{0,\cdots, K_1^c\}$ different file requests for the $K_1^c$ stored files in $\mathcal F_1^c$ and $k_1^b\in\{0,\cdots, F_1^b\}$ different uncached file requests in $\mathcal F_1^b$, where $k_1^c+k_1^b>0$. We consider two cases. If $k_1^b\leq K_1^b$,  the macro-BS  transmits  each of the $k_1^c+k_1^b$ files  at rate $\tau $ (bit/second) and over $\frac{1}{k_1^c+k_1^b}$  of the total bandwidth $W$ (using FDMA). If $k_1^b> K_1^b$,  the macro-BS first selects $K_1^b$ file requests out of the $k_1^b$ file requests uniformly at random, and then transmits  each of the $k_1^c+K_1^b$ files  at rate $\tau $ and over $\frac{1}{k_1^c+K_1^b}$  of the total bandwidth $W$. Consider a pico-BS which has $k_2^c\in\{1,\cdots, K_2^c\}$ different file requests for the $K_2^c$ stored files in $\mathcal F_2^c$.  The pico-BS transmits each of the $k_2^c$ files at rate $\tau $  and over $\frac{1}{k_2^c}$  of the total bandwidth $W$.
Note that, by avoiding transmitting  the same file multiple times to multiple users, this content-centric transmission (multicast) can  improve the efficiency of the utilization of the wireless medium and reduce the load of the wireless links, compared to the traditional connection-based transmission (unicast).

From the above illustration, we can see that the proposed multicasting design is also affected by the proposed hybrid caching design $\left(\mathcal F_1^c,\mathcal F_2^c,\mathbf p\right)$. Therefore, the design parameters  $\left(\mathcal F_1^c,\mathcal F_2^c,\mathbf p\right)$ affect the performance of the proposed joint caching and multicasting design.

\section{Performance Metric}\label{Sec:perf}
%In this section, we first elaborate on the performance metric. Then, we formulation the optimal caching and multicasting problem to maximize the network performance.

%\subsection{Performance Metric}\label{Subsec:perfm}

In this paper, w.l.o.g., we study the performance of the typical user denoted as $u_0$, which is located at the origin. 
%The index of the typical user and its serving BS is denoted as $0$. 
We assume all BSs are active.
%\begin{align}
%y_{0}=D_{0,0}^{-\frac{\alpha}{2}}h_{0,0}x_{0}+\sum_{\ell\in\Phi_{b}\backslash B_{n,0}}  D_{\ell,0}^{-\frac{\alpha}{2}}h_{\ell,0}x_{\ell}+n_{0},
%\end{align}
Suppose $u_0$ requests file $n$. Let $j_0$ denote the index of the tier  to which $u_0$ belongs, and let $\overline j_0$ denote the other tier.  Let $\ell_0\in\Phi_{j_0}$ denote the index of the serving BS of $u_0$.
%$D_{j_0,\ell_0, 0}$ denotes the distance between $u_{0}$ and  BS $\ell_0\in\Phi_{j_0}$, $h_{j_0,\ell_0,0}\dis \mathcal{CN}\left(0,1\right)$ denotes the small-scale channel between  BS $\ell_0$ and $u_{0}$, %$x_{0}$ is the transmit signal from $\ell_0$  to $u_{0}$,
We denote $D_{j,\ell,0}$ and $h_{j,\ell,0}\dis \mathcal{CN}\left(0,1\right)$ as the distance  and the small-scale channel  between BS $\ell\in\Phi_j$ and $u_{0}$, respectively.
%$x_{j,\ell}$ is the transmit signal from BS $\ell$ of the $j$ tier to its scheduled user,  and
%$n_{0}\dis\mathcal{CN}\left(0,N_{0}\right)$ is
We assume the complex additive white Gaussian noise of power $N_0$ at $u_0$.
When $u_{0}$ requests file $n$ and file $n$ is transmitted by BS $\ell_0$, the  signal-to-interference plus noise ratio (SINR) of $u_{0}$ is given by
\begin{align}
{\rm SINR}_{n,0}=\frac{{D_{j_0,\ell_0,0}^{-\alpha_{j_0}}}\left|h_{j_0,\ell_0,0}\right|^{2}}{\sum_{\ell\in\Phi_{j_0}\backslash \ell_0}D_{j_0,\ell,0}^{-\alpha_{j_0}}\left|h_{j_0,\ell,0}\right|^{2}+\sum_{\ell\in\Phi_{\overline j_0}}D_{\overline j_0,\ell,0}^{-\alpha_{\overline j_0}}\left|h_{\overline j_0,\ell,0}\right|^{2}\frac{P_{\overline j_0}}{P_{j_0}}+\frac{N_{0}}{P_{j_0}}}, \ n\in\mathcal N.
\label{eqn:SINR}
\end{align}

%\begin{figure*}[!t]
%\small{\begin{align}
%&q_1(\mathcal F_1^c,\mathcal F_2^c)=\sum_{n\in \mathcal F_1^c}a_{n}\sum_{k^c=1}^{K_1^c}\sum_{k^b=0}^{F_1^b}\Pr \left[K_{1c,n,0}=k^c\right]\Pr\left[K_{1b,n,0}'=k^b \right]{\rm Pr}\left[\frac{W}{k^c+\min\{K_1^b,k^b\}}\log_{2}\left(1+{\rm SINR}_{1,n,0}\right)\geq\tau \right]\nonumber\\
%&+\sum_{n\in \mathcal F_1^b}a_{n}\sum_{k^c=0}^{K_1^c}\sum_{k^b=1}^{F_1^b}\Pr \left[K_{1c,n,0}'=k^c\right]\Pr\left[K_{1b,n,0}=k^b \right]\frac{\min\{K_1^b,k^b\}}{k^b}{\rm Pr}\left[\frac{W}{k^c+\min\{K_1^b,k^b\}}\log_{2}\left(1+{\rm SINR}_{1,n,0}\right)\geq\tau \right]
%\label{eqn:succ-prob-def-1}\\
%&q_2(\mathcal F_2^c,\mathbf p)=\sum_{n\in \mathcal F_2^c}a_{n}\sum_{k^c=1}^{K_2^c}\Pr \left[K_{2,n,0}^c=k^c\right]{\rm Pr}\left[\frac{W}{k^c}\log_{2}\left(1+{\rm SINR}_{2,n,0}\right)\geq\tau  \right]\label{eqn:succ-prob-def-2}
%\end{align}}
%\normalsize \hrulefill
%\end{figure*}

When $u_0$ requests file $n\in\mathcal F_1^c$ ($n\in \mathcal F_1^b$),  let $K_{1,n,0}^c\in \{1,\cdots, K_1^c\}$  ($\overline{K}_{1,n,0}^c\in \{0,\cdots, K_1^c\}$) and $\overline{K}_{1,n,0}^b\in \{0,\cdots, F_1^b\}$ ($K_{1,n,0}^b\in \{1,\cdots, F_1^b\}$) denote the numbers of different cached and uncached files requested by the users associated with BS $\ell_0\in\Phi_{1}$, respectively.
%If $u_0$ requests file $n\in \mathcal F_1^b$,  let $\overline{K}_{1,n,0}^c\in \{0,\cdots, K_1^c\}$  and $K_{1,n,0}^b\in \{1,\cdots, F_1^b\}$ denote the numbers of different cached and uncached files requested by the users associated with BS $\ell_0\in\Phi_1$, respectively.
When $u_0$ requests file $n\in \mathcal F_2^c$,  let $K_{2,n,0}^c\in \{1,\cdots, K_2^c\}$  denote the number of different cached  files requested by the users associated with BS $\ell_0\in\Phi_2$. Note that  $K_{1,n,0}^c,\overline{K}_{1,n,0}^b,\overline{K}_{1,n,0}^c,K_{1,n,0}^b,K_{2,n,0}^c$ are discrete random variables, the probability mass functions (p.m.f.s) of which depend on $\mathbf a$, $\lambda_u$ and the design parameters $\left(\mathcal F_1^c,\mathcal F_2^c, \mathbf p\right)$. In addition, if $n\in\mathcal F_1^c\cup \mathcal F_2^c$, BS $\ell_0$ will transmit file $n$ for sure; if $n\in\mathcal F_1^b$, for given $K_{1,n,0}^b=k^b\geq1$,  BS $\ell_0$ will transmit file $n$ with probability $\frac{\min\{k^b,K_1^b\}}{k^b}$.
%Given file $n$ is transmitted by BS $\ell_0$,  the corresponding channel capacity of $u_0$ is given by
%\begin{align}
%C_{n,0}\triangleq
%\begin{cases}
% \frac{W}{K_{1,n,0}^c+\min\{\overline{K}_{1,n,0}^b,K_1^b\}}\log_{2}\left(1+{\rm SINR}_{n,0}\right), & n\in\mathcal F_1^c\\
%  \frac{W}{\overline{K}_{1,n,0}^c+\min\{K_{1,n,0}^b,K_1^b\}}\log_{2}\left(1+{\rm SINR}_{n,0}\right), & n\in\mathcal F_1^b\\
%   \frac{W}{K_{2,n,0}^c}\log_{2}\left(1+{\rm SINR}_{n,0}\right), & n\in\mathcal F_2^c
%\end{cases}\label{eqn:capacity}
%\end{align}
Given that file $n$ is transmitted, it can be decoded correctly at $u_0$ if the channel capacity between BS $\ell_0$ and $u_0$ is greater than or equal to $\tau$. %$C_{n,0}\geq \tau$.
Requesters are mostly concerned with whether  their desired files can be successfully received. Therefore, in this paper, we consider the successful transmission probability  of a file requested by $u_0$ %, also referred to as the successful transmission probability,
as the network performance metric.
%\footnote{The traditional SINR coverage probability (which does not capture resource sharing) and rate coverage probability (which only captures resource sharing among different users) cannot reflect resource sharing among different files under multicasting\cite{WCOM13Andrews}.}%, and hence are not suitable to measure the performance in our case.}
%Specifically, we investigate the successful transmission probability of the typical user $u_{0}$.
By total probability theorem,  the
successful transmission probability   under the proposed   scheme is given by:
\begin{align}
q(\mathcal F_1^c,\mathcal F_2^c, \mathbf p)=q_1(\mathcal F_1^c,\mathcal F_2^c)+q_2(\mathcal F_2^c, \mathbf p),\label{eqn:succ-prob-def}
\end{align}
where $\mathcal F_1^b$ is given by \eqref{eqn:backhaul}, and  $q_1(\mathcal F_1^c,\mathcal F_2^c)$ and $q_2(\mathcal F_2^c,\mathbf p)$ are given by \eqref{eqn:succ-prob-def-1} and \eqref{eqn:succ-prob-def-2}, respectively. Note that in \eqref{eqn:succ-prob-def-1} and \eqref{eqn:succ-prob-def-2}, each term multiplied by $a_n$ represents the successful transmission probability of file $n$.

\begin{figure*}[!t]
\small{\begin{align}
q_1(\mathcal F_1^c,\mathcal F_2^c)=&\sum_{n\in \mathcal F_1^c}a_{n}{\rm Pr}\left[\frac{W}{K_{1,n,0}^c+\min\{K_1^b,\overline {K}_{1,n,0}^b\}}\log_{2}\left(1+{\rm SINR}_{n,0}\right)\geq\tau \right]\nonumber\\
&+\sum_{n\in \mathcal F_1^b}a_{n}{\rm Pr}\left[\frac{W}{\overline{K}_{1,n,0}^{c}+\min\{K_1^b,K_{1,n,0}^b\}}\log_{2}\left(1+{\rm SINR}_{n,0}\right)\geq\tau, \text{$n$ is selected} \right]
\label{eqn:succ-prob-def-1}\\
q_2(\mathcal F_2^c,\mathbf p)=&\sum_{n\in \mathcal F_2^c}a_{n}{\rm Pr}\left[\frac{W}{K_{2,n,0}^c}\log_{2}\left(1+{\rm SINR}_{n,0}\right)\geq\tau  \right]\label{eqn:succ-prob-def-2}
\end{align}}
\normalsize \hrulefill
\end{figure*}
%\begin{figure*}[!t]
%\small{\begin{align}
%&q_1(\mathcal F_1^c,\mathcal F_2^c)=\sum_{n\in \mathcal F_1^c}a_{n}\sum_{k^c=1}^{K_1^c}\sum_{k^b=0}^{F_1^b}\Pr \left[K_{1,n,0}^c=k^c\right]\Pr\left[\overline {K}_{1,n,0}^b=k^b \right]{\rm Pr}\left[\frac{W}{k^c+\min\{K_1^b,k^b\}}\log_{2}\left(1+{\rm SINR}_{n,0}\right)\geq\tau \right]\nonumber\\
%&+\sum_{n\in \mathcal F_1^b}a_{n}\sum_{k^c=0}^{K_1^c}\sum_{k^b=1}^{F_1^b}\Pr \left[\overline{K}_{1,n,0}^{c}=k^c\right]\Pr\left[K_{1,n,0}^b=k^b \right]\frac{\min\{K_1^b,k^b\}}{k^b}{\rm Pr}\left[\frac{W}{k^c+\min\{K_1^b,k^b\}}\log_{2}\left(1+{\rm SINR}_{n,0}\right)\geq\tau \right]
%\label{eqn:succ-prob-def-1}\\
%&q_2(\mathcal F_2^c,\mathbf p)=\sum_{n\in \mathcal F_2^c}a_{n}\sum_{k^c=1}^{K_2^c}\Pr \left[K_{2,n,0}^c=k^c\right]{\rm Pr}\left[\frac{W}{k^c}\log_{2}\left(1+{\rm SINR}_{n,0}\right)\geq\tau  \right]\label{eqn:succ-prob-def-2}
%\end{align}}
%\normalsize \hrulefill
%\end{figure*}

Later, we shall see that under the proposed caching and multicasting design  for content-oriented services in the cache-enabled HetNet, the successful transmission probability is sufficiently different from the traditional  rate coverage probability studied for connection-oriented services \cite{WCOM13Andrews}. In particular, the successful transmission probability considered in this paper not only depends on  the physical layer parameters, such as the macro and pico BS densities $\lambda_1$ and  $\lambda_2$,  user density $\lambda_u$, path loss exponents $\alpha_1$ and $\alpha_2$, bandwidth $W$, backhaul capacity $K_1^b$  and transmit signal-to-noise ratios (SNRs) $\frac{P_1}{N_0}$ and $\frac{P_2}{N_0}$, but also relies on the content-related parameters, such as the popularity distribution $\mathbf a$, the cache sizes $K_1^c$ and $K_2^c$, and the design parameters $\left(\mathcal F_1^c, \mathcal F_2^c,\mathbf p(\mathcal F_2^c)\right)$. While, the traditional rate coverage probability only depends on the physical layer parameters. In addition, the successful transmission probability  depends on  the physical layer parameters in a different way from the traditional rate coverage probability.  For example, the content-centric association leads to different distributions of the locations of  serving and interfering BSs; the multicasting transmission results in different file load distributions at each BS \cite{WCOM13Andrews}; and the cache-enabled architecture makes   content availability related to   BS densities.

\section{Performance Analysis}

In this section, we study the successful transmission probability $q\left(\mathcal F_1^c,\mathcal F_2^c,\mathbf{p}\right)$ under the proposed caching and multicasting design  for given design parameters $\left(\mathcal F_1^c,\mathcal F_2^c,\mathbf p(\mathcal F_2^c)\right)$. First, we analyze the successful transmission probability in the general region. Then, we analyze the asymptotic transmission probability in the high SNR and user density region.

\subsection{Performance Analysis in General Region}

In this part, we  would like to analyze the successful transmission probability in the general region, using tools from stochastic geometry. In general, file loads $K_{1,n,0}^c$, $\overline{K}_{1,n,0}^b$, $\overline{K}_{1,n,0}^c$, $K_{1,n,0}^b$, $K_{2,n,0}^c$ and SINR ${\rm SINR}_{n,0}$ are correlated in a complex manner, as BSs with larger association regions have higher file load and lower SINR (due to larger user to BS distance) \cite{AndrewsTWCOffloading14}.  For the tractability of the analysis, as in \cite{WCOM13Andrews} and \cite{AndrewsTWCOffloading14}, the dependence is ignored. Therefore, to obtain the successful transmission probability in \eqref{eqn:succ-prob-def}, we analyze the distributions  of $K_{1,n,0}^c$, $\overline{K}_{1,n,0}^b$, $\overline{K}_{1,n,0}^c$, $K_{1,n,0}^b$, $K_{2,n,0}^c$ and the distribution of ${\rm SINR}_{n,0}$, separately.

First, we calculate the p.m.f.s of $K_{1,n,0}^c$ and $\overline{K}_{1,n,0}^b$ for $n\in \mathcal F_1^c$ as well as  the p.m.f.s of  $\overline{K}_{1,n,0}^c$ and $K_{1,n,0}^b$ for  $n\in \mathcal F_1^b$.  In calculating these p.m.f.s,
%$\Pr \left[K_{1,n,0}^c=k^c\right]$, $\Pr \left[K_{1,n,0}^b=k^b\right]$, $\Pr \left[\overline{K}_{1,n,0}^c=\overline{k}^c\right]$ and $\Pr \left[\overline{K}_{1,n,0}^b=\overline{k}^b\right]$,
we need
the probability density function (p.d.f.) of the size of the Voronoi cell of $\ell_{0}$ w.r.t. file $m\in  \mathcal F_1^c \cup \mathcal F_1^b \setminus \{n\}$. Note that  this p.d.f. is equivalent to  the  p.d.f. of the size of the Voronoi cell to which a randomly chosen user belongs. Based on  a tractable approximated form of this p.d.f. in \cite{SGcellsize13},  which is widely used in existing literature\cite{AndrewsTWCOffloading14,WCOM13Andrews}, we obtain the p.m.f.s of $K_{1,n,0}^c$, $\overline{K}_{1,n,0}^b$, $\overline{K}_{1,n,0}^c$ and $K_{1,n,0}^b$.

\begin{Lem} [p.m.f.s of $K_{1,n,0}^c$ $\overline{K}_{1,n,0}^b$, $\overline{K}_{1,n,0}^c$ and $K_{1,n,0}^b$] The p.m.f.s of $K_{1,n,0}^c$ and $\overline{K}_{1,n,0}^b$ for $n\in \mathcal F_1^c$ and the p.m.f.s of $\overline{K}_{1,n,0}^c$ and $K_{1,n,0}^b$ for $n\in\mathcal F_1^b$ are given by
\begin{align}
&\Pr \left[K_{1,n,0}^c=k^c\right]=g(\mathcal F_{1,-n}^c,k^c-1),\quad k^c=1,\cdots, K_1^c,\label{eqn:K-1-c}\\
%&\Pr \left[K_{1,n,0}^c=k^c\right]=\sum_{\mathcal X \in g(\mathcal F_{1,-n}^c,k^c-1)} \prod_{m\in \mathcal X}\left(1-\eta(a_m\lambda_u,\lambda_1)\right)\prod_{m\in {\mathcal F_{1,-n}^c \setminus \mathcal X}}\eta(a_m\lambda_u,\lambda_1),\quad k^c=1,\cdots, K_1^c,\label{eqn:K-1-c}\\
&\Pr \left[\overline{K}_{1,n,0}^b=k^b\right]=g(\mathcal F_{1}^b,k^b),\quad k^b=0,\cdots, F_1^b,\label{eqn:K-1-b-bar}\\
%&\Pr \left[\overline{K}_{1,n,0}^b=k^b\right]=\sum_{\mathcal X \in g(\mathcal F_{1}^b,k^b)} \prod_{m\in \mathcal X}\left(1-\eta(a_m\lambda_u,\lambda_1)\right)\prod_{m\in {\mathcal F_{1}^b \setminus \mathcal X}}\eta(a_m\lambda_u,\lambda_1),\quad k^b=0,\cdots, F_1^b,\label{eqn:K-1-b-bar}\\
&\Pr \left[\overline{K}_{1,n,0}^c=k^c\right]= g(\mathcal F_{1}^c,k^c),\quad k^c=0,\cdots, K_1^c,\label{eqn:K-1-c-bar}\\
%&\Pr \left[\overline{K}_{1,n,0}^c=k^c\right]=\sum_{\mathcal X \in g(\mathcal F_{1}^c,k^c)} \prod_{m\in \mathcal X}\left(1-\eta(a_m\lambda_u,\lambda_1)\right)\prod_{m\in {\mathcal F_{1}^c \setminus \mathcal X}}\eta(a_m\lambda_u,\lambda_1),\quad k^c=0,\cdots, K_1^c,\label{eqn:K-1-c-bar}\\
&\Pr \left[K_{1,n,0}^b=k^b\right]=g(\mathcal F_{1,-n}^b,k^b-1),\quad k^b=1,\cdots, F_1^b,\label{eqn:K-1-b}
%&\Pr \left[K_{1,n,0}^b=k^b\right]=\sum_{\mathcal X \in g(\mathcal F_{1,-n}^b,k^b-1)} \prod_{m\in \mathcal X}\left(1-\eta(a_m\lambda_u,\lambda_1)\right)\prod_{m\in {\mathcal F_{1,-n}^b \setminus \mathcal X}}\eta(a_m\lambda_u,\lambda_1),\quad k^b=1,\cdots, F_1^b,
%&\Pr \left[K_{1,n,0}^b=k^b\right]=\sum_{\mathcal{F}_{1}^{b1} \in \mathcal{SF}_{1}^{b1}(k^b-1)} \prod_{m\in \mathcal F_{1}^{b1}}\left(1-U_m^{-4.5}\right)\prod_{m\in {\mathcal F_{1,-n}^b \setminus \mathcal F_{1}^{b1}}}U_m^{-4.5},\quad k^b=1,\cdots, K_1^b,\nonumber\\
%&\Pr \left[\overline{K}_{1,n,0}^c=\overline{k}^c\right]=\sum_{\mathcal{F}_{1}^{c1} \in \mathcal{SF}_{1}^{c1}(k^c)} \prod_{m\in \mathcal F_{1}^{c1}}\left(1-U_m^{-4.5}\right)\prod_{m\in {\mathcal F_{1}^c \setminus \mathcal F_{1}^{c1}}}U_m^{-4.5},\quad \overline{k}^c=0,\cdots, K_1^c,\nonumber\\
%&\Pr \left[\overline{K}_{1,n,0}^b=\overline{k}^b\right]=\sum_{\mathcal{F}_{1}^{b1} \in \mathcal{SF}_{1}^{b1}(k^b)} \prod_{m\in \mathcal F_{1}^{b1}}\left(1-U_m^{-4.5}\right)\prod_{m\in {\mathcal F_{1}^b \setminus \mathcal F_{1}^{b1}}}U_m^{-4.5},\quad \overline{k}^b=0,\cdots, K_1^b.
%\label{eqn:K-1-b}
\end{align}
%\begin{align}
%&\Pr \left[K_{1c,n,0}'=k_{1c},K_{1b,n,0}=k_{1b}\right]\nonumber\\
%&\approx\sum_{\mathcal{F}_{1c}^1 \in \mathcal{SF}_{1c}^1(k_1^c)} \prod_{m\in \mathcal F_{1c}^1}\left(1-U_m^{-4.5}\right)\prod_{m\in {\mathcal F_{1c}-\mathcal F_{1c}^1}}U_m^{-4.5}\nonumber\\
%&\times\sum_{\mathcal{F}_{1b}^1 \in \mathcal{SF}_{1b}^1(k_1^b-1)} \prod_{m\in \mathcal F_{1b}^1}\left(1-U_m^{-4.5}\right)\prod_{m\in {\mathcal F_{1b,-n}-\mathcal F_{1b}^1}}U_m^{-4.5}.\nonumber\\
%&\hspace{10mm}k_1^c=0,\cdots, K_1^c, \hspace{10mm}k_1^b=1,\cdots, K_1^b\label{eqn:K-pmf-m}
%\end{align}
where
%$g(\mathcal X,k)\triangleq \left\{\mathcal S \subseteq \mathcal X: |\mathcal S|=k\right\}$,
%$\mathcal{SF}_{1}^{c1}(k-1)\triangleq \left\{\mathcal F_{1}^{c1} \subseteq \mathcal F_{1,-n}^c :|\mathcal F_{1}^{c1}|=k-1\right\}$,
%$\mathcal{SF}_{1}^{b1}(k-1)\triangleq \left\{\mathcal F_{1}^{b1} \subseteq \mathcal F_{1,-n}^b :|\mathcal F_{1}^{b1}|=k-1\right\}$,
%$\mathcal{SF}_{1}^{c1}(k)\triangleq \left\{\mathcal F_{1}^{c1} \subseteq \mathcal F_{1}^c :|\mathcal F_{1}^{c1}|=k\right\}$,
%$\mathcal{SF}_{1}^{b1}(k)\triangleq \left\{\mathcal F_{1}^{b1} \subseteq \mathcal F_{1}^b :|\mathcal F_{1}^{b1}|=k\right\}$,
%$\mathcal{SF}_{1c}^1(k)\triangleq \left\{\mathcal F_{1c}^1 \subseteq \mathcal F_{1c} :|\mathcal F_{1c}^1|=k\right\}$,
%$\mathcal{SF}_{1b}^1(k-1)\triangleq \left\{\mathcal F_{1b}^1 \subseteq \mathcal F_{1b,-n} :|\mathcal F_{1b}^1|=k-1\right\}$,
%$\mathcal{SF}_{1b}^1(k)\triangleq \left\{\mathcal F_{1b}^1 \subseteq \mathcal F_{1b} :|\mathcal F_{1b}^1|=k\right\}$.
%and $U_m=1+3.5^{-1}\frac{a_m\lambda_u}{\lambda_1}$,
%and $\eta(x,y)\triangleq\left(1+\frac{x}{3.5y}\right)^{-4.5}$.
$g(\mathcal F ,k)\triangleq \sum_{\mathcal X \in \left\{\mathcal S \subseteq \mathcal F: |\mathcal S|=k\right\} } \prod_{m\in \mathcal X}\left(1-\left(1+\frac{a_m\lambda_u}{3.5\lambda_1}\right)^{-4.5}\right)\prod_{m\in {\mathcal F \setminus \mathcal X}}\left(1+\frac{a_m\lambda_u}{3.5\lambda_1}\right)^{-4.5}$,  $\quad$
$\mathcal F_{1,-n}^{c}\triangleq\mathcal F_{1}^{c} \setminus \{n\}$ and $\mathcal F_{1,-n}^{b}\triangleq\mathcal F_{1}^{b} \setminus \{n\}$.
\label{Lem:pmf-K-m}
\end{Lem}
\begin{proof}
Please refer to Appendix A.
\end{proof}

Next, we obtain the p.m.f. of $K_{2,n,0}^c$ for $n\in \mathcal F_2^c$.  In calculating the   p.m.f. of $K_{2,n,0}^c$, we need
the p.d.f. of the size of the Voronoi cell of $\ell_{0}$ w.r.t. file $m\in   \mathcal N_i \setminus \{n\}$  when  $\ell_{0}$ contains combination $i\in \mathcal I_n$.  However, this p.d.f. is very complex and is still unknown. For the tractability of  the analysis,  as in \cite{arXivSGCaching15}, we approximate this p.d.f. based on a tractable approximated form of the  p.d.f. of the size of the Voronoi cell to which a randomly chosen user belongs\cite{SGcellsize13}, which is widely used in existing literature\cite{AndrewsTWCOffloading14,WCOM13Andrews}. Under this approximation, we obtain the p.m.f. of $K_{2,n,0}^c$.

\begin{Lem} [p.m.f. of $K_{2,n,0}^c$] The p.m.f. of $K_{2,n,0}^c$ for $n\in\mathcal F_2^c$ is given by
\begin{align}
&\Pr \left[K_{2,n,0}^c=k^c\right]\nonumber\\
&=\sum_{i\in \mathcal I_n}\frac{p_i}{T_n}\sum_{\mathcal X\in  \left\{\mathcal S \subseteq \mathcal N_{i,-n}: |\mathcal S|=k^c-1\right\}  }\prod\limits_{m\in \mathcal X}\left(1-\left(1+\frac{a_m\lambda_u}{3.5T_m\lambda_2}\right)^{-4.5}\right)\prod\limits_{m\in {\mathcal N_{i,-n}\setminus \mathcal X}}\left(1+\frac{a_m\lambda_u}{3.5T_m\lambda_2}\right)^{-4.5},\nonumber\\
&\hspace{12cm}k^c=1,\cdots, K_2^c,\label{eqn:K-pmf}
\end{align}
where
%$\mathcal{SN}_i^1(k-1)\triangleq \left\{\mathcal N_i^1 \subseteq \mathcal N_{i,-n} :|\mathcal N_i^1|=k-1\right\}$,
%$g(\cdot,\cdot)$ and $\eta(\cdot,\cdot)$ are defined in Lemma~\ref{Lem:pmf-K-m}, and
$\mathcal N_{i,-n}\triangleq  \mathcal N_i \setminus \{n\}$.
%$W_m(T_m)\triangleq1+3.5^{-1}\frac{a_m\lambda_u}{T_m\lambda_2}$ and
%$G_{n,i}(\mathcal X,\mathbf T_{i,-n})\triangleq \prod\limits_{m\in \mathcal X}\left(1-\eta\left(a_m\lambda_u,T_m\lambda_2\right)\right)\prod\limits_{m\in {\mathcal N_{i,-n}\setminus \mathcal X}}\eta\left(a_m\lambda_u,T_m\lambda_2\right). $
%\begin{align}
%&G_{n,i}(\mathcal X,\mathbf T_{i,-n})\triangleq \prod_{m\in \mathcal X}\left(1-\eta\left(a_m\lambda_u,T_m\lambda_2\right)\right)\prod_{m\in {\mathcal N_{i,-n}\setminus \mathcal X}}\eta\left(a_m\lambda_u,T_m\lambda_2\right). \nonumber
%\end{align}
\label{Lem:pmf-K}
\end{Lem}
\begin{proof}
Please refer to Appendix B.
\end{proof}

%The distributions of the locations of desired transmitters and interferers and the distributions of  file loads $K_{1,n,0}^c,\overline{K}_{1,n,0}^b,\overline{K}_{1,n,0}^c,K_{1,n,0}^b$ in the cache-enabled HetNet are different from those   in traditional connection-based HetNet\cite{WCOM13Andrews}. By carefully handling these distributions and using tools from stochastic geometry, we obtain a tractable expression for $q(\mathcal F_1^c,\mathcal F_2^c, \mathbf p)$, which can be found in \cite{dropboxISIT16}  and is omitted here due to page limitation.

The distributions of the locations of desired transmitters and interferers are more involved than those in the traditional connection-based HetNets. Thus, it is more challenging to analyze the p.d.f. of  ${\rm SINR}_{n,0}$.
When $u_0$ is a macro-user, as in the traditional connection-based HetNets, there are two types of interferers, namely, i) all the other macro-BSs besides its serving macro-BS, and ii) all the pico-BSs.  When $u_0$ is a pico-user, different from  the traditional connection-based HetNets, there are three types of interferers, namely,  i) all the other pico-BSs storing the combinations containing the desired file of $u_0$  besides its serving pico-BS,   ii) all the pico-BSs without  the desired file of $u_0$, and iii) all the macro-BSs. By carefully handling these distributions, we can derive the p.d.f. of  ${\rm SINR}_{n,0}$, for $n\in\mathcal F_1^c \cup \mathcal F_1^b $ and $n\in\mathcal F_2^c $, respectively.

Then, based on Lemma~\ref{Lem:pmf-K-m} and Lemma~\ref{Lem:pmf-K} as well as the p.d.f. of  ${\rm SINR}_{n,0}$, we can derive the successful transmission probability  $q\left(\mathcal F_1^c,\mathcal F_2^c,\mathbf{p}\right)$.

\begin{Thm} [Performance]
The successful transmission probability  $q\left(\mathcal F_1^c,\mathcal F_2^c,\mathbf{p}\right)$ of $u_{0}$ is given by
%\begin{align}\label{eq:CPrate_multifile_noise}
%q_K\left(\mathbf{p}\right)=2\pi\lambda_{b}\sum_{n\in \mathcal N}a_{n} T_n \int_{0}^{\infty}&d\exp\left(-\frac{2\pi}{\alpha} T_n \lambda_{b}\left(2^{\frac{\tau K}{W}}-1\right)^{\frac{2}{\alpha}}d^{2}B^{'}\left(\frac{2}{\alpha},1-\frac{2}{\alpha},2^{-\frac{\tau K}{W}}\right)\right)\notag\\
%&\times \exp\left(-\frac{2\pi}{\alpha}\left(1-T_n\right)\lambda_{b}\left(2^{\frac{\tau K}{W}}-1\right)^{\frac{2}{\alpha}}d^{2}B\left(\frac{2}{\alpha},1-\frac{2}{\alpha}\right)\right)\nonumber\\
%&\times\exp\left(-\pi T_n \lambda_{b}d^{2}\right)\exp\left(-\left(2^{\frac{\tau K}{W}}-1\right)d^{\alpha}\frac{N_0}{P}\right){\rm d}d,
%\end{align}
\begin{align}\label{eq:CPrate_multifile_noise}
q\left(\mathcal F_1^c,\mathcal F_2^c,\mathbf{p}\right)=&\sum_{n\in \mathcal F_1^c}a_{n} \sum_{k^c=1}^{K_1^c}\sum_{k^b=0}^{F_1^b} \Pr [K_{1,n,0}^c=k^c] \Pr[\overline{K}_{1,n,0}^b=k^b]f_{1,{k^c+\min\{K_1^b,k^b\}}}\nonumber\\
&+\sum_{n\in \mathcal F_1^b}a_{n} \sum_{k^c=0}^{K_1^c}\sum_{k^b=1}^{F_1^b} \Pr [\overline{K}_{1,n,0}^c=k^c]\Pr [K_{1,n,0}^b=k^b] \frac{\min\{K_1^b,k^b\}}{k^b}  f_{1,{k^c+\min\{K_1^b,k^b\}}}\nonumber\\
&+\sum_{n\in \mathcal F_2^c}a_{n} \sum_{k^c=1}^{K_2^c} \Pr[K_{2,n,0}=k^c]f_{2,k^c}(T_n),
\end{align}
where the p.m.f.s of $K_{1,n,0}^c$ $\overline{K}_{1,n,0}^b$, $\overline{K}_{1,n,0}^c$, $K_{1,n,0}^b$ and $K_{2,n,0}^c$ are given by Lemma~\ref{Lem:pmf-K-m} and Lemma~\ref{Lem:pmf-K}, $f_{1,k}$ and  $f_{2,k}(T_n)$  are given by \eqref{eqn:f-1-k} and \eqref{eqn:f-2-k}, and  $T_n$  is given by \eqref{eqn:def-T-n}.
%&f_{1,k}\triangleq \nonumber\\ &2\pi\lambda_{1}\int_{0}^{\infty}d\exp\left(-\pi\lambda_{1}d^{2}\right)\exp\left(- \left(2^{\frac{k\tau}{W}}-1\right)d^{\alpha_1}\frac{N_0}{P_1}\right)\nonumber\\
%&\times\exp\left(-\frac{2\pi\lambda_{2}}{\alpha_2}d^{\frac{\alpha_1}{\alpha_2}}\left(\frac{P_2}{P_1}\left(2^{\frac{k\tau}{W}}-1\right)\right)^{\frac{2}{\alpha_2}}B\left(\frac{2}{\alpha_2},1-\frac{2}{\alpha_2}\right)\right)\nonumber\\
%&\times
%\exp\left(-\frac{2\pi\lambda_{1}x}{\alpha_1}d^2\left(2^{\frac{k\tau}{W}}-1\right)^{\frac{2}{\alpha_1}}B^{'}\left(\frac{2}{\alpha_1},1-\frac{2}{\alpha_1},2^{-\frac{k\tau}{W}}\right)\right){\rm d}d.\label{eqn:def-f}
%\end{align}
Here, $B^{'}\left(x,y,z\right)\triangleq \int_{z}^{1}u^{x-1}\left(1-u\right)^{y-1}{\rm d}u$ and $B(x,y)\triangleq\int_{0}^{1}u^{x-1}\left(1-u\right)^{y-1}{\rm d}u$ denote  the complementary incomplete Beta function and  the Beta function, respectively. \label{Thm:generalKmulti}
\end{Thm}
\begin{proof}
Please refer to Appendix C.
\end{proof}

From Theorem~\ref{Thm:generalKmulti}, we can see that in the general region, the physical  layer parameters $\alpha_1$, $\alpha_2$, $W$, $\lambda_1$, $\lambda_2$, $\lambda_u$, $\frac{P_1}{N_0}$, $\frac{P_2}{N_0}$, and the design parameters $\left(\mathcal F_1^c,\mathcal F_2^c,\mathbf{p}\right)$
jointly affect the successful transmission probability $q\left(\mathcal F_1^c,\mathcal F_2^c,\mathbf{p}\right)$.  The impacts of the physical layer parameters   and  the design parameters on $q\left(\mathcal F_1^c,\mathcal F_2^c,\mathbf{p}\right)$  are coupled in a complex manner.
%It is very difficult to obtain design insights from this expression.

\begin{figure*}[!t]
\small{\begin{align}
f_{1,k}\triangleq& 2\pi\lambda_{1}\int_{0}^{\infty} d\exp\left(- \left(2^{\frac{k\tau}{W}}-1\right)d^{\alpha_1}\frac{N_0}{P_1}\right)\exp\left(-\frac{2\pi\lambda_{2}}{\alpha_2}d^{\frac{2\alpha_1}{\alpha_2}}\left(\frac{P_2}{P_1}\left(2^{\frac{k\tau}{W}}-1\right)\right)^{\frac{2}{\alpha_2}}B\left(\frac{2}{\alpha_2},1-\frac{2}{\alpha_2}\right)\right)\nonumber\\
&  \times\exp\left(-\frac{2\pi\lambda_{1}}{\alpha_1}d^2\left(2^{\frac{k\tau}{W}}-1\right)^{\frac{2}{\alpha_1}}B^{'}\left(\frac{2}{\alpha_1},1-\frac{2}{\alpha_1},2^{-\frac{k\tau}{W}}\right)\right)\exp\left(-\pi\lambda_{1}d^{2}\right){\rm d}d.\label{eqn:f-1-k}\\
f_{2,k}(x)\triangleq &2\pi\lambda_{2}x\int_{0}^{\infty} d\exp\left(- \left(2^{\frac{k\tau}{W}}-1\right)d^{\alpha_2}\frac{N_0}{P_2}\right)\exp\left(-\frac{2\pi\lambda_{1}}{\alpha_1}d^{\frac{2\alpha_2}{\alpha_1}}\left(\frac{P_1}{P_2}\left(2^{\frac{k\tau}{W}}-1\right)\right)^{\frac{2}{\alpha_1}}B\left(\frac{2}{\alpha_1},1-\frac{2}{\alpha_1}\right)\right)\nonumber\\
& \times
\exp\left(-\frac{2\pi\lambda_{2}}{\alpha_2}d^2\left(2^{\frac{k\tau}{W}}-1\right)^{\frac{2}{\alpha_2}}\left(xB^{'}\left(\frac{2}{\alpha_2},1-\frac{2}{\alpha_2},2^{-\frac{k\tau}{W}}\right)+\left(1-x\right)B\left(\frac{2}{\alpha_2},1-\frac{2}{\alpha_2}\right)\right)\right)\nonumber\\
&\times \exp\left(-\pi\lambda_{2}xd^{2}\right){\rm d}d.\label{eqn:f-2-k}
\end{align}}
\normalsize \hrulefill
\end{figure*}

\subsection{Performance Analysis in Asymptotic Region}

In this part, to obtain design insights,  we  focus on analyzing  the asymptotic successful transmission probability  in the high SNR and user density region. %, in which the gain of multicast transmissions over unicast transmissions  is very promising in the cache-enabled HetNet.
Note that in the remaining of the paper, when considering the high SNR region, we assume
$P_1=\beta P$ and $P_2=P$ for some $\beta>1$ and $P>0$, and let $\frac{P}{N_0}\to \infty$.
%The high SNR region indicates the region where $\frac{P}{N_0}\to \infty$.
On the other hand, in the high user density region where $\lambda_u\to \infty $, discrete random variables $K_{1,n,0}^c,\overline{K}_{1,n,0}^c\to K_1^c$, $\overline{K}_{1,n,0}^b,K_{1,n,0}^b\to F_1^b$ and $K_{2,n,0}^c\to K_2^c$ in distribution.
Define $ q_{1, \infty}\left(\mathcal F_1^c, \mathcal F_2^c\right)\triangleq \lim_{\frac{P}{N_0}\to \infty , \lambda_u\to \infty}q_{1}\left(\mathcal F_1^c, \mathcal F_2^c\right)$, $  q_{2, \infty}\left(\mathcal F_2^c,\mathbf{T}\right)\triangleq \lim_{\frac{P}{N_0}\to \infty , \lambda_u\to \infty}q_{2}\left(\mathcal F_2^c,\mathbf{p}\right)$, and $ q_{ \infty}\left(\mathcal F_1^c, \mathcal F_2^c,\mathbf{T}\right)\triangleq \lim_{\frac{P}{N_0}\to \infty , \lambda_u\to \infty} q(\mathcal F_1^c,\mathcal F_2^c, \mathbf p) $. Note that when $\lambda_u \to \infty$, $q_2$ and $q$ become functions of $\mathbf T$ instead of $\mathbf p$.
From Theorem~\ref{Thm:generalKmulti}, we have the following lemma.
\begin{Lem}[Asymptotic Performance]
When  $\frac{P}{N_0}\to \infty$ and $\lambda_u\to \infty$, we have $ q_{ \infty}\left(\mathcal F_1^c, \mathcal F_2^c,\mathbf{T}\right)\\=q_{1, \infty}\left(\mathcal F_1^c, \mathcal F_2^c\right)+q_{2, \infty}\left(\mathcal F_2^c,\mathbf{T}\right)$, where
$q_{1, \infty}\left(\mathcal F_1^c, \mathcal F_2^c\right)= f_{1,K_1^c+\min\{K_1^b,F_1^b\},\infty}\bigg(\sum_{n\in \mathcal F_1^b}\frac{\min\{K_1^b,F_1^b\}}{F_1^b}a_n\\+\sum_{n\in \mathcal F_1^c}a_n \bigg)$ and $q_{2, \infty}\left(\mathcal F_2^c,\mathbf{T}\right)=\sum_{n\in \mathcal F_2^c}a_{n}f_{2,K_2^c,\infty}(T_n)$.
%\begin{align}
% q_{1, \infty}\left(\mathcal F_1^c, \mathcal F_2^c\right)=& f_{1,K_1^c+\min\{K_1^b,F_1^b\},\infty}\left(\sum_{n\in \mathcal F_1^c}a_n +\sum_{n\in \mathcal F_1^b}\frac{\min\{K_1^b,F_1^b\}}{F_1^b}a_n\right),\nonumber\\
% q_{2, \infty}\left(\mathcal F_2^c,\mathbf{T}\right)=&\sum_{n\in \mathcal F_2^c}a_{n}f_{2,K_2^c,\infty}(T_n).\nonumber
%\end{align}
Here, $f_{1,k,\infty}$ and  $f_{2,k,\infty}(T_n)$ are given by \eqref{eqn:f-1-k-infty} and \eqref{eqn:f-2-k-infty}, and  $T_n$ is  given by \eqref{eqn:def-T-n}.

\begin{proof}
Please refer to Appendix D.
\end{proof}
%As $\lambda_u\to \infty$, $q(\mathcal F_1^c,\mathcal F_2^c, \mathbf p)\to q_{\lambda_u,\infty}\left(\mathcal F_1^c, \mathcal F_2^c,\mathbf{T}\right)\triangleq  q_{1,\lambda_u,\infty}\left(\mathcal F_1^c, \mathcal F_2^c\right)+q_{2,\lambda_u,\infty}\left(\mathcal F_2^c,\mathbf{T}\right)$, where $q_{1,\lambda_u,\infty}\left(\mathcal F_1^c, \mathcal F_2^c\right)\triangleq\sum_{n\in \mathcal F_1^c}a_n f_{1,K_1}+\frac{K_1^b}{F_1^b}\sum_{n\in \mathcal F_1^b}a_nf_{1,K_1}$ and $q_{2,\lambda_u,\infty}\left(\mathcal F_2^c,\mathbf{T}\right)\triangleq\sum_{n\in \mathcal F_2^c}a_{n}f_{2,K_2^c}(T_n)$, with $f_{1,k}$ and $f_{2,k}(T_n)$  given by \eqref{eqn:f-1-k} and \eqref{eqn:f-2-k}, and  $T_n$   given by \eqref{eqn:def-T-n}.
\label{Lem:asym-perf}
\end{Lem}
%\begin{proof}
%Please refer to Appendix D.
%\end{proof}

\begin{figure*}[!t]
\small{\begin{align}
f_{1,k,\infty}\triangleq& 2\pi\lambda_{1}\int_{0}^{\infty} d\exp\left(-\pi\lambda_{1}d^{2}\right)\exp\left(-\frac{2\pi\lambda_{2}}{\alpha_2}d^{\frac{2\alpha_1}{\alpha_2}}\left(\frac{1}{\beta}\left(2^{\frac{k\tau}{W}}-1\right)\right)^{\frac{2}{\alpha_2}}B\left(\frac{2}{\alpha_2},1-\frac{2}{\alpha_2}\right)\right)\nonumber\\
&  \times\exp\left(-\frac{2\pi\lambda_{1}}{\alpha_1}d^2\left(2^{\frac{k\tau}{W}}-1\right)^{\frac{2}{\alpha_1}}B^{'}\left(\frac{2}{\alpha_1},1-\frac{2}{\alpha_1},2^{-\frac{k\tau}{W}}\right)\right){\rm d}d.\label{eqn:f-1-k-infty}\\
f_{2,k,\infty}(x)\triangleq &2\pi\lambda_{2}x\int_{0}^{\infty} d\exp\left(-\pi\lambda_{2}xd^{2}\right)\exp\left(-\frac{2\pi\lambda_{1}}{\alpha_1}d^{\frac{2\alpha_2}{\alpha_1}}\left(\beta\left(2^{\frac{k\tau}{W}}-1\right)\right)^{\frac{2}{\alpha_1}}B\left(\frac{2}{\alpha_1},1-\frac{2}{\alpha_1}\right)\right)\nonumber\\
& \times
\exp\left(-\frac{2\pi\lambda_{2}}{\alpha_2}d^2\left(2^{\frac{k\tau}{W}}-1\right)^{\frac{2}{\alpha_2}}\left(xB^{'}\left(\frac{2}{\alpha_2},1-\frac{2}{\alpha_2},2^{-\frac{k\tau}{W}}\right)+\left(1-x\right)B\left(\frac{2}{\alpha_2},1-\frac{2}{\alpha_2}\right)\right)\right)
{\rm d}d.\label{eqn:f-2-k-infty}
\end{align}}
\normalsize \hrulefill
\end{figure*}

Note that $f_{1,K_1^c+\min\{K_1^b,F_1^b\},\infty}$ represents the successful transmission probability for file $n\in \mathcal F_1^c\cup\mathcal F_1^b$ (given that this file is transmitted), and $f_{2,K_2^c,\infty}(T_n)$ represents the successful transmission probability for file $n\in \mathcal F_2^c$, in the asymptotic region.
For given $(\mathcal F_1^c,\mathcal F_2^c, \mathbf T)$,   we interpret Lemma~\ref{Lem:asym-perf} below.   When $F_1^b\leq K_1^b$, the successful transmission probability of file $n_1\in\mathcal F_1^c$ is the same as that of file $n_2\in\mathcal F_1^b$. In other words, when backhaul capacity is sufficient, storing a file at a macro-BS   or retrieving the file  via the backhaul link makes no difference in successful transmission probability.
When $F_1^b> K_1^b$, the successful transmission probability of file $n_1\in\mathcal F_1^c$ is greater than that  of file $n_2\in\mathcal F_1^b$. In other words, when backhaul capacity is limited, storing a file at a macro-BS is better than  retrieving the file via the backhaul link. Note  that $f_{2,k,\infty}(x)$ is an increasing function (Please refer to Appendix E for the proof). Thus, for any $n_1,n_2\in\mathcal F_2^c$ satisfying $T_{n_1}>T_{n_2}$,  the successful transmission probability of file $n_1\in\mathcal F_2^c$ is greater than that  of file $n_2\in\mathcal F_2^c$. That is, a file of higher  probability being cached at a pico-BS has higher  successful transmission probability.  Later, in Section~\ref{Sec:new-opt}, we shall see that the structure of $q_{\infty}\left(\mathcal F_1^c, \mathcal F_2^c,\mathbf{T}\right)$ facilitates the optimization of $q\left(\mathcal F_1^c, \mathcal F_2^c,\mathbf{p}\right)$.

Next, we further study the symmetric  case where $\alpha_1=\alpha_2 \triangleq \alpha$ in the high SNR and user density region. From Lemma~\ref{Lem:asym-perf}, we have the following lemma.
\begin{Lem}[Asymptotic Performance When $\alpha_1=\alpha_2$]
When $\alpha_1=\alpha_2 =\alpha$,   $\frac{P}{N_0}\to \infty$ and $\lambda_u\to \infty$,  we have $q_{ \infty}\left(\mathcal F_1^c, \mathcal F_2^c,\mathbf{T}\right)= q_{1, \infty}\left(\mathcal F_1^c, \mathcal F_2^c\right)+q_{2, \infty}\left(\mathcal F_2^c,\mathbf{T}\right)$, where $q_{1, \infty }\left(\mathcal F_1^c, \mathcal F_2^c\right)= \frac{1}{\omega_{K_1^c+\min\{K_1^b, F_1^b\}}}\left(\sum_{n\in \mathcal F_1^c}a_n +\frac{\min\{K_1^b,F_1^b\}}{F_1^b}\sum_{n\in \mathcal F_1^b}a_n\right)$ and  $q_{2, \infty}\left(\mathcal F_2^c,\mathbf{T}\right)= \sum_{n\in \mathcal F_2^c }\frac{a_nT_n}{\theta_{2,K_2^c}+\theta_{1,K_2^c}T_n}$.
%\begin{align}
%q_{1, \infty }\left(\mathcal F_1^c, \mathcal F_2^c\right)= &\frac{1}{\omega_{K_1^c+\min\{K_1^b, F_1^b\}}}\left(\sum_{n\in \mathcal F_1^c}a_n +\frac{\min\{K_1^b,F_1^b\}}{F_1^b}\sum_{n\in \mathcal F_1^b}a_n\right),\nonumber\\
%q_{2, \infty}\left(\mathcal F_2^c,\mathbf{T}\right)=& \sum_{n\in \mathcal F_2^c }\frac{a_nT_n}{\theta_{2,K_2^c}+\theta_{1,K_2^c}T_n}.\nonumber
%\end{align}
Here, $T_n$ is given by \eqref{eqn:def-T-n}, and $\omega_{k}$, $\theta_{1,k}$ and  $\theta_{2,k}$ are given by
\begin{align}
\omega_{k}=&\frac{2}{\alpha}\left(2^{\frac{k\tau}{W}}-1\right)^{\frac{2}{\alpha}}B'\left(\frac{2}{\alpha},1-\frac{2}{\alpha},2^{\frac{-k\tau}{W}}\right)
+\frac{2\lambda_2}{\alpha\lambda_1}\left(\frac{1}{\beta}\left(2^{\frac{k\tau}{W}}-1\right)\right)^{\frac{2}{\alpha}}B\left(\frac{2}{\alpha},1-\frac{2}{\alpha}\right)+1,\label{eqn:c_1_k}\\
\theta_{1,k}=&\frac{2}{\alpha}\left(2^{\frac{k\tau}{W}}-1\right)^{\frac{2}{\alpha}}B'\left(\frac{2}{\alpha},1-\frac{2}{\alpha},2^{\frac{-k\tau}{W}}\right)-\frac{2}{\alpha}\left(2^{\frac{k\tau}{W}}-1\right)^{\frac{2}{\alpha}}B\left(\frac{2}{\alpha},1-\frac{2}{\alpha}\right)+1 ,\label{eqn:c_2_1_k}\\
\theta_{2,k}=&\frac{2}{\alpha}\left(2^{\frac{k\tau}{W}}-1\right)^{\frac{2}{\alpha}}B\left(\frac{2}{\alpha},1-\frac{2}{\alpha}\right)+\frac{2\lambda_1}{\alpha\lambda_2}\left(\beta\left(2^{\frac{k\tau}{W}}-1\right)\right)^{\frac{2}{\alpha}}B\left(\frac{2}{\alpha},1-\frac{2}{\alpha}\right).\label{eqn:c_2_2_k}
\end{align}
\label{Lem:asym-perf-v2}
\end{Lem}
\begin{proof}
Please refer to Appendix D.
\end{proof}

From Lemma~\ref{Lem:asym-perf-v2}, we can see that   in the high SNR and user density region,  when $\alpha_1=\alpha_2 =\alpha$, the impact of the physical layer parameters $\alpha$, $\beta$ and $W$, captured by $\omega_k$, $\theta_{1,K}$ and $\theta_{2,K}$, and the impact of the design parameters $\left(\mathcal F_1^c,\mathcal F_2^c,\mathbf{p}\right)$ on the successful transmission probability $q_{ \infty}\left(\mathcal F_1^c, \mathcal F_2^c,\mathbf{T}\right)$ can be easily separated. Later, in Section~\ref{Sec:new-opt}, we shall see that this separation greatly facilitates the optimization of $q\left(\mathcal F_1^c, \mathcal F_2^c,\mathbf{p}\right)$.

Fig.~\ref{fig:verification-Kmulti}  plots the successful transmission probability  versus the  transmit  SNR $\frac{P}{N_0}$ and the user density $\lambda_u$.
 Fig.~\ref{fig:verification-Kmulti}  verifies Theorem~\ref{Thm:generalKmulti} and Lemma~\ref{Lem:asym-perf} (Lemma~\ref{Lem:asym-perf-v2}), and demonstrates the accuracy of  the approximations adopted.
 Fig.~\ref{fig:verification-Kmulti} also indicates that  $q_{\infty}\left(\mathcal F_1^c, \mathcal F_2^c,\mathbf{T}\right)$  provides a simple and good approximation for   $q\left(\mathcal F_1^c, \mathcal F_2^c,\mathbf{T}\right)$  in  the high transmit SNR region (e.g., $\frac{P}{N_0}\geq 100$ dB) and the high user density region (e.g., $\lambda_u\geq 3\times 10^{-5}$).

%On the other hand, Fig.~\ref{fig:verification-Kmulti} shows that  $q_K(\mathbf p)$  is greater than  $q_K^{uc}(\mathbf p)$; $q_K(\mathbf p)$ and  $q_K^{uc}(\mathbf p)$ both decrease with  $\lambda_u$; and  the gap   $q_K(\mathbf p)-q_K^{uc}(\mathbf p)$ increases with $\lambda_u$. These observations verify the discussions  in Section~\ref{Subsec:perfm} and demonstrate the advantage of the proposed multicasting scheme over the traditional unicasting scheme.

 \begin{figure}[t]
\begin{center}
 \includegraphics[width=7cm]{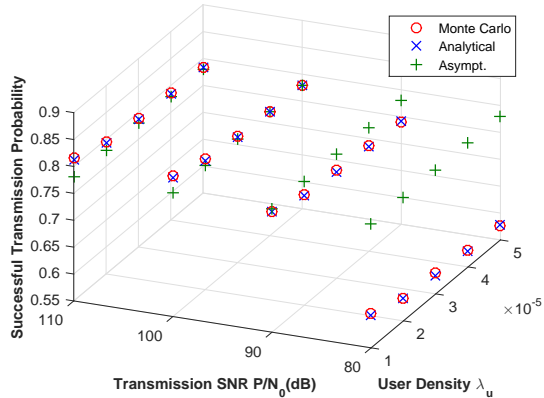}
  \end{center}
    \caption{\small{Successful transmission probability versus transmit SNR $\frac{P}{N_0}$  and user density $\lambda_u$.  $N=10$, $K_1^c=3$, $K_2^c=2$, $K_1^b=1$, $\mathcal F_1^c=\{1,2,3\}$, $\mathcal F_1^b=\{7,8,9,10\}$, $\mathcal F_2^c=\{4,5,6\}$, $\mathbf p = (0.7, 0.2, 0.1)$, $\mathcal N_1=\{4,5\}$, $\mathcal N_2=\{4,6\}$, $\mathcal N_3=\{5,6\}$, $\lambda_{1}=5\times10^{-7}$, $\lambda_{2}=3\times10^{-6}$, $P_1=10^{1.5}P$, $P_2=P$, $\alpha_1=\alpha_2=4$, $W = 20\times 10^6$, $\tau =  2\times10^4$, and $a_n=\frac{n^{-\gamma}}{\sum_{n\in \mathcal N}n^{-\gamma}}$  with $\gamma =1$. In this paper, to simulate the large-scale HetNet, we use a 2-dimensional square of area $15000^2$, which is sufficiently large in our case.
    Note that if the simulation window size is not large enough, the observed interference would be smaller than the true interference due to the edge effect, resulting in larger successful transmission probability than the true value.
    In addition, the Monte Carlo results are obtained by averaging over $10^5$ random realizations.}}
\label{fig:verification-Kmulti}
\end{figure}

%,  $q_{\lambda_u,\infty}\left(\mathcal F_1^c, \mathcal F_2^c,\mathbf{T}\right)$ is a function of $\mathbf T$ instead of $\mathbf p$ and has a much simplified form.

%The caching, scheduling  and multicasting design fundamentally affects the network performance  via the caching design $\mathcal F_1^c$, $\mathcal F_2^c$ and $\mathbf p(\mathcal F_2^c)$.
%We would like to consider the optimal caching, scheduling  and multicasting to maximize the successful transmission probability by carefully optimizing the design $\mathcal F_1^c$, $\mathcal F_2^c$ and $\mathbf p(\mathcal F_2^c)$.  The maximization of $q(\mathcal F_1^c,\mathcal F_2^c, \mathbf p)$ under the constraints in \eqref{eqn:cache-constr}, \eqref{eqn:cache-constr-indiv} and \eqref{eqn:cache-constr-sum}, which is a mixed discrete-continuos optimization with a complex objective function,  is not tractable and brute force solutions do not offer many design insights. Therefore, in the following, we focus on optimizing the asymptotic successful transmission probability in the high user density regime,to get first-order insights into the design of cache-enabled Hetnets. First, we formulate the asymptotic performance optimization problem as follows.

\section{Optimization Problem Formulation}

In this section,  we  formulate the optimal  caching and multicasting design problem  to maximize  the successful transmission probability $q\left(\mathcal F_1^c, \mathcal F_2^c,\mathbf{p}\right)$, which is a mixed discrete-continuous optimization problem. To facilitate the solution of this challenging  optimization problem in the next section,  we also formulate the   asymptotically optimal caching and multicasting design problem  to maximize the asymptotic successful transmission probability $q_{\infty}\left(\mathcal F_1^c, \mathcal F_2^c,\mathbf{T}\right)$ in the high SNR and user density region.

\subsection{Optimization Problem}

The caching and multicasting design fundamentally affects the successful transmission probability  via the design parameters $\left(\mathcal F_1^c,\mathcal F_2^c,\mathbf{p}\right)$.
We would like to maximize $q\left(\mathcal F_1^c, \mathcal F_2^c,\mathbf p\right)$ by carefully optimizing  $\left(\mathcal F_1^c,\mathcal F_2^c,\mathbf{p}\right)$. %Specifically, the optimization problem is formulated as follows.

\begin{Prob} [Performance Optimization]\label{prob:opt}
\begin{align}
q^*\triangleq \max_{\mathcal F_1^c, \mathcal F_2^c, \mathbf{p}} &\quad  q\left(\mathcal F_1^c, \mathcal F_2^c,\mathbf p\right)\nonumber\\
s.t. & \quad \eqref{eqn:cache-constr}, \eqref{eqn:cache-constr-indiv}, \eqref{eqn:cache-constr-sum},\nonumber
\end{align}
where $q\left(\mathcal F_1^c, \mathcal F_2^c,\mathbf p\right)$ is given by \eqref{eq:CPrate_multifile_noise}.
%The optimal solution is written as $\left(\mathcal F_1^{c*}, \mathcal F_2^{c*},\mathbf p^*\right)$.
\end{Prob}

%Note that in this paper, we focus on the successful transmission probability  maximization  to get first-order insights into the design of cache-enabled Hetnets.\footnote{The optimal solution to Problem~\ref{prob:opt} may result in starvation of requesters for files with low caching probabilities. We assume these starving users can be satisfied  through other service mechanisms at additional backhaul or delay costs.}  The  optimization framework in this paper can be easily applied to address  the  QoS requirements in terms of the successful transmission probability of each file, e.g., $q_{K,n}(\mathbf p)\geq Q_{K,n}$ for all $n\in \mathcal N$.
%Later, in Sections~\ref{Subsec:K1-opt} and \ref{Subsec:K-opt}, we shall solve Problem~\ref{prob:opt} for the unit cache  size ($K=1$) and the general cache size ($K>1$), respectively.

Note that Problem~\ref{prob:opt} is a mixed discrete-continuous optimization problem with two main challenges. One is the choice of the  sets of files $\mathcal F_1^c$ and $\mathcal F_2^c$ (discrete variables)  stored in the two tiers, and the other is the choice of the caching distribution $\mathbf p(\mathcal F_2^c)$  (continuous variables) of random caching for the 2nd tier. We thus propose an equivalent alternative formulation of Problem~\ref{prob:opt}  which naturally subdivides Problem~\ref{prob:opt}  according to these two aspects.

\begin{Prob} [Equivalent Optimization]\label{prob:opt-eq}
\begin{align}
q^*=\max_{\mathcal F_1^c, \mathcal F_2^c} &\quad q_{1}(\mathcal F_1^c, \mathcal F_2^c)+q_{2}^*\left(\mathcal F_2^c\right)\label{eqn:prob-caching}\\
s.t. & \quad \eqref{eqn:cache-constr}.\nonumber
\end{align}
%where %The optimal solution is written as $\left(\mathcal F_1^{c*}, \mathcal F_2^{c*}\right)$ and $q_{2}^*(\mathcal F_2^c)$ is given by
\begin{align}
q_{2}^*(\mathcal F_2^c)\triangleq \max_{\mathbf p} &\quad  q_{2}\left(\mathcal F_2^c,\mathbf p\right)\label{eqn:prob-p-F_2^c}\\
s.t. & \quad  \eqref{eqn:cache-constr-indiv}, \eqref{eqn:cache-constr-sum}.\nonumber
\end{align}
%where the optimal solution  is written as $\mathbf p^*(\mathcal F_2^c)$. Let $\mathbf p^*=\mathbf p^*(\mathcal F_2^{c*})$.
\end{Prob}

For given $\mathcal F_2^c$, the optimization problem in \eqref{eqn:prob-p-F_2^c} is in general a non-convex optimization problem with a large number of optimization variables (i.e., $I=\binom{F_2^c}{K_2^c}$ optimization variables), and it is difficult to obtain the global optimal solution and calculate  $q_{2}^*\left(\mathcal F_2^c\right)$.
Even given $q_{2}^*\left(\mathcal F_2^c\right)$, the optimization problem in \eqref{eqn:prob-caching} is a discrete optimization problem over a very large constraint set, and is NP-complete in general.
Therefore,  Problem~\ref{prob:opt-eq} is still  very challenging.

%
%In the following, we propose a two-step optimization framework to obtain an asymptotically optimal solution with manageable complexity and superior performance. Specifically, we first identify a set of asymptotically optimal solutions in the high user density region. Then, we obtain the asymptotically optimal solution achieving the optimal successful transmission probability in the general region among the set of asymptotically optimal solutions, referred to as the best asymptotically optimal solution.

\subsection{Asymptotic Optimization Problem}\label{subsec:asym-opt-prob}

To facilitate the solution of the challenging mixed discrete-continuous  optimization problem, we also formulate the optimization of the asymptotic successful transmission probability $q_{\infty}\left(\mathcal F_1^c, \mathcal F_2^c,\mathbf{T}\right)$ given in Lemma~\ref{Lem:asym-perf}, i.e., which has a much simpler form than $q\left(\mathcal F_1^c, \mathcal F_2^c,\mathbf p\right)$ given in Theorem~\ref{Thm:generalKmulti}. Equivalently, we can consider the asymptotic version of Problem~\ref{prob:opt-eq} in the high SNR and user density region.

%\begin{Prob} [Asymptotic Performance Optimization]\label{prob:opt-asymp}
%\begin{align}
%q_{\infty}^*\triangleq \max_{\mathcal F_1^c, \mathcal F_2^c, \mathbf{p}} &\quad  q_{\infty}\left(\mathcal F_1^c, \mathcal F_2^c,\mathbf{T}\right)\nonumber\\
%s.t. & \quad \eqref{eqn:cache-constr}, \eqref{eqn:cache-constr-indiv}, \eqref{eqn:cache-constr-sum}, \eqref{eqn:def-T-n}, \nonumber
%\end{align}
%where  $q_{\infty}\left(\mathcal F_1^c, \mathcal F_2^c,\mathbf{T}\right)$ is given by  Lemma~\ref{Lem:asym-perf}. The optimal solution is written as $\left(\mathcal F_1^{c*}, \mathcal F_2^{c*},\mathbf p^*\right)$.
%\end{Prob}
%
% Similarly, Problem~\ref{prob:opt-asymp} is also a mixed discrete-continuous optimization.
% By subdividing the discrete and continuous parts and exploring the relationship\footnote{Note that different from $q(\mathcal F_1^c,\mathcal F_2^c, \mathbf p)$,
%$q_{\infty}\left(\mathcal F_1^c, \mathcal F_2^c,\mathbf{T}\right)$   is a function of $\mathbf T$,  instead of $\mathbf p$.} between $\mathbf T$ and $\mathbf p$, we introduce  a new optimization problem  to further simplify Problem~\ref{prob:opt-asymp}.

\begin{Prob} [Asymptotic Optimization]\label{prob:opt-asymp-eq}
\begin{align}
q_{\infty}^*=\max_{\mathcal F_1^c, \mathcal F_2^c} &\quad q_{1,\infty}(\mathcal F_1^c, \mathcal F_2^c)+q_{2,\infty}^*\left(\mathcal F_2^c\right)\label{eqn:prob-asymp-caching}\\
s.t. & \quad \eqref{eqn:cache-constr}.\nonumber
\end{align}
The optimal solution to the optimization in \eqref{eqn:prob-asymp-caching} is written as $\left(\mathcal F_1^{c*}, \mathcal F_2^{c*}\right)$ and $q_{2,\infty}^*(\mathcal F_2^c)$ is given by
\begin{align}
q_{2,\infty}^*(\mathcal F_2^c)\triangleq \max_{\mathbf{p}} &\quad  q_{2,\infty}\left(\mathcal F_2^c,\mathbf{p}\right)\label{eqn:prob-p-asymp-F_2^c-p}\\
s.t. & \quad \eqref{eqn:cache-constr-indiv}, \eqref{eqn:cache-constr-sum}, \eqref{eqn:def-T-n},\nonumber
\end{align}
where the optimal solution to the optimization in \eqref{eqn:prob-p-asymp-F_2^c-p} is written as $\mathbf p^*(\mathcal F_2^c)$. The optimal solution to Problem \ref{prob:opt-asymp-eq} is given by $(\mathcal F_1^{c*},\mathcal F_2^{c*},\mathbf p^*(\mathcal F_2^{c*}))$, which is the asymptotic optimal solution to Problem \ref{prob:opt-eq} (Problem \ref{prob:opt}).
\end{Prob}

Based on Lemma 2 in \cite{arXivSGCaching15}, we know that  the optimization in \eqref{eqn:prob-p-asymp-F_2^c-p} is equivalent to the following optimization for given $\mathcal F_2^c$
\begin{align}
q_{2,\infty}^*(\mathcal F_2^c)\triangleq \max_{\mathbf{T}} &\quad  q_{2,\infty}\left(\mathcal F_2^c,\mathbf{T}\right)\label{eqn:prob-p-asymp-F_2^c}\\
s.t. & \quad \quad 0\leq T_n\leq 1,\ n\in \mathcal F_2^c, \label{eqn:cache-constr-indiv-t}\\
&\quad\sum_{n\in \mathcal F_2^c}T_n=K_2^c, \label{eqn:cache-constr-sum-t}
\end{align}
where the optimal solution  is written as $\mathbf T^*(\mathcal F_2^c)$.
In addition, any $\mathbf p^*(\mathcal F_2^c)$ in convex polyhedron $\mathcal P^*(\mathcal F_2^c)\triangleq \{\mathbf p^*(\mathcal F_2^c): \eqref{eqn:cache-constr-indiv}, \eqref{eqn:cache-constr-sum},  \eqref{eqn:def-T-n*} \}$ is an optimal solution to the optimization in \eqref{eqn:prob-p-asymp-F_2^c-p}, where \eqref{eqn:def-T-n*} is given by:
\begin{align}
\sum_{i\in \mathcal I_n}p_i^*(\mathcal F_2^c)=T_n^*(\mathcal F_2^c), \ n\in \mathcal F_2^c.\label{eqn:def-T-n*}
\end{align}
The vertices of the convex polyhedron $\mathcal P^{*}(\mathcal F_2^c)$ can be obtained based on the simplex method, and any $\mathbf p^*(\mathcal F_2^{c}) \in \mathcal P^{*}(\mathcal F_2^c)$ can be constructed from all the vertices using convex combination. Thus, when optimizing  the asymptotic performance for given $\mathcal F_2^c$, we can  focus on  the optimization in \eqref{eqn:prob-p-asymp-F_2^c} instead of the optimization in \eqref{eqn:prob-p-asymp-F_2^c-p}.

\section{Near Optimal Solution}\label{Sec:new-opt}

In this section, we propose a two-step optimization framework to obtain a near optimal solution with manageable complexity and superior performance in the general region.  We first characterize the structural properties of the asymptotically optimal solutions. Then, based on these properties, we obtain a near optimal solution in the general region.

\subsection{Asymptotically Optimal Solution}

 In this part, we study  the  continuous part  and the discrete part of the asymptotic optimization in Problem~\ref{prob:opt-asymp-eq}, respectively, to obtain design insights into the solution in the general region.

\subsubsection{Continuous  Optimization}  As the structure of $q_{2,\infty}\left(\mathcal F_2^c,\mathbf{T}\right)$ is very complex, it is difficult to obtain the closed-form optimal solution $\mathbf T^*(\mathcal F_2^c)$ to the optimization in \eqref{eqn:prob-p-asymp-F_2^c}. By exploring the structural properties of $q_{2,\infty}\left(\mathcal F_2^c,\mathbf{T}\right)$, we know that files of higher popularity get more storage resources.
\begin{Lem} [Structural Property of Optimization in \eqref{eqn:prob-p-asymp-F_2^c}]
Given any $\mathcal F_2^c\subseteq \mathcal N$ satisfying  $F_2^c\geq K_2^c$ and  $n_1, n_2 \in \mathcal F_2^c$, if $n_1<n_2$, then  $T_{n_1}^*(\mathcal F_2^c)\geq T_{n_2}^*(\mathcal F_2^c)$.\label{Lem:mono-general-asym}
\end{Lem}
\begin{proof} Please refer to Appendix E.
\end{proof}

Now, we focus on obtaining a numerical solution to the optimization in \eqref{eqn:prob-p-asymp-F_2^c}.
For given $\mathcal F_2^c\subseteq \mathcal N$ satisfying  $F_2^c\geq K_2^c$, the optimization in \eqref{eqn:prob-p-asymp-F_2^c} is a continuous optimization of a differentiable  function $q_{2,\infty}\left(\mathcal F_2^c,\mathbf{T}\right)$ over a convex set. In general, it is difficult to  show the convexity of $f_{2,k,\infty}(x)$ in \eqref{eqn:f-2-k-infty}.
 A stationary point  to the optimization in \eqref{eqn:prob-p-asymp-F_2^c} can be obtained using standard gradient projection methods\cite[pp. 223]{Bertsekasbooknonlinear:99}.
Here, we consider the diminishing stepsize\cite[pp. 227]{Bertsekasbooknonlinear:99} satisfying
\begin{align}
\epsilon(t)\to 0\ \text{as}\ t\to \infty,\  \sum_{t=1}^{\infty}\epsilon(t)=\infty, \sum_{t=1}^{\infty}\epsilon(t)^2<\infty, \label{eqn:stepcond}
\end{align}
and propose Algorithm~\ref{alg:local}. In Step 2 of Algorithm~\ref{alg:local}, $\frac {\partial q_{2,\infty}\left(\mathcal F_2^c,\mathbf{T}(t)\right)}{\partial T_n(t)}=a_nf_{2,k,\infty}'(x)$, where $f_{2,k,\infty}'(x)$ is given by \eqref{eqn:f-2-k_d}.
\begin{figure*}[!t]
\small{\begin{align}
f_{2,k,\infty}'(x)\triangleq &2\pi\lambda_{2}x\int_{0}^{\infty} d\exp\left(-\frac{2\pi\lambda_{2}}{\alpha_2}d^2\left(2^{\frac{k\tau}{W}}-1\right)^{\frac{2}{\alpha_2}}\left(xB^{'}\left(\frac{2}{\alpha_2},1-\frac{2}{\alpha_2},2^{-\frac{k\tau}{W}}\right)+\left(1-x\right)B\left(\frac{2}{\alpha_2},1-\frac{2}{\alpha_2}\right)\right)\right)\nonumber\\
& \times\exp\left(-\pi\lambda_{2}xd^{2}\right)\exp\left(-\frac{2\pi\lambda_{1}}{\alpha_1}d^{\frac{2\alpha_2}{\alpha_1}}\left(\frac{P_1}{P_2}\left(2^{\frac{k\tau}{W}}-1\right)\right)^{\frac{2}{\alpha_1}}B\left(\frac{2}{\alpha_1},1-\frac{2}{\alpha_1}\right)\right)\nonumber\\ &\times\left(-\pi\lambda_{2}d^{2}-\frac{2\pi\lambda_{2}}{\alpha_2}d^{2}\left(2^{\frac{k\tau}{W}}-1\right)^{\frac{2}{\alpha_2}}\left(B^{'}\left(\frac{2}{\alpha_2},1-\frac{2}{\alpha_2},2^{-\frac{k\tau}{W}}\right)-B\left(\frac{2}{\alpha_2},1-\frac{2}{\alpha_2}\right)\right)\right)
{\rm d}d\nonumber\\
&+\frac{f_{2,k,\infty}(x)}{x}.\label{eqn:f-2-k_d}
\end{align}}
\normalsize \hrulefill
\end{figure*}
%*********************full version**************************
%*********************simplification**************************
%In Step 2, the expression of $\frac {\partial q_K\left(\mathbf{p}(t)\right)}{\partial p_i(t)}$ can be found in  \cite{arXivSGCaching15}, and   the diminishing stepsize\cite[pp. 227]{Bertsekasbooknonlinear:99} satisfies
%\begin{align}
%\epsilon(t)\to 0\ \text{as}\ t\to \infty,\  \sum_{t=1}^{\infty}\epsilon(t)=\infty, \sum_{t=1}^{\infty}\epsilon(t)^2<\infty.  \label{eqn:stepcond}
%\end{align}
Step 3 is the projection of $\bar T_n(t+1)$ onto the set of the  variables satisfying the constraints in \eqref{eqn:cache-constr-indiv-t} and \eqref{eqn:cache-constr-sum-t}. It is shown in \cite[pp. 229]{Bertsekasbooknonlinear:99} that $\mathbf T(t)$ in Algorithm~\ref{alg:local} converges
to  a stationary point of  the optimization in \eqref{eqn:prob-p-asymp-F_2^c}  as $t\to \infty$.

On the other hand, as illustrated in the discussion of Lemma~\ref{Lem:asym-perf}, $f_{2,k,\infty}(x)$ is actually a cumulative density function (c.d.f.), and is  concave in most of the cases we are interested in. If $f_{2,k,\infty}(x)$ in \eqref{eqn:f-2-k} is concave w.r.t. $x$, the differentiable function $q_{2,\infty}\left(\mathcal F_2^c,\mathbf{T}\right)$ is concave w.r.t. $\mathbf T$, and hence, the optimization in \eqref{eqn:prob-p-asymp-F_2^c} is a convex problem. Then, $\mathbf T(t)$ in Algorithm~\ref{alg:local} converges  to the optimal solution $\mathbf T^*(\mathcal F_2^c)$  to  the optimization in \eqref{eqn:prob-p-asymp-F_2^c}  as $t\to \infty$. In other words, under a mild condition (i.e., $f_{2,k,\infty}(x)$ is convex), we can obtain the optimal solution $\mathbf T^*(\mathcal F_2^c)$  to  the optimization in \eqref{eqn:prob-p-asymp-F_2^c} using Algorithm~\ref{alg:local}.

%As illustrated in the discussion of Lemma~\ref{Lem:asym-perf}, $f_{2,k,\infty}(x)$ is actually a cumulative density function (c.d.f.), and is  convex in most of the cases we are interested in. Thus, in this part, we assume  $f_{2,k,\infty}(x)$ is convex. The optimal solution to the optimization in \eqref{eqn:prob-p-asymp-F_2^c} can be obtained using standard gradient projection methods\cite[pp. 223]{Bertsekasbooknonlinear:99}.
%Here, we consider the diminishing stepsize\cite[pp. 227]{Bertsekasbooknonlinear:99} satisfying
%\begin{align}
%\epsilon(t)\to 0\ \text{as}\ t\to \infty,\  \sum_{t=1}^{\infty}\epsilon(t)=\infty, \sum_{t=1}^{\infty}\epsilon(t)^2<\infty, \label{eqn:stepcond}
%\end{align}
%and propose Algorithm~\ref{alg:local} to solve the optimization in \eqref{eqn:prob-p-asymp-F_2^c}. It is shown in \cite[pp. 229]{Bertsekasbooknonlinear:99} that $\mathbf T(t)$ converges to the optimal solution $\mathbf T^*(\mathcal F_2^c)$ as $t\to \infty$, under a mild condition (i.e., $f_{2,k,\infty}(x)$ is convex).\footnote{If $f_{2,k,\infty}(x)$ in \eqref{eqn:f-2-k} is not convex w.r.t. $x$, $\mathbf T(t)$ converges to a stationary point of  the optimization in \eqref{eqn:prob-p-asymp-F_2^c}\cite[pp. 229]{Bertsekasbooknonlinear:99}.}
%

\begin{algorithm}[t]
\caption{Asymptotically Optimal Solution}
\small{\begin{algorithmic}[1]
\STATE Initialize  $t=1$ and $T_n(1)=\frac{1}{F_2^c}$  for all $n\in \mathcal F_2^c$.
%\LOOP
 \STATE   For all $n\in \mathcal F_2^c$, compute $\bar T_n(t+1)$ according to $\bar T_n(t+1)=T_n(t)+\epsilon(t)\frac {\partial q_{2,\infty}\left(\mathcal F_2^c,\mathbf{T}(t)\right)}{\partial T_n(t)}$,
where  $\{\epsilon(t)\}$ satisfies \eqref{eqn:stepcond}.
 \STATE For all $n\in \mathcal F_2^c$, compute $T_n(t+1)$ according to $T_n(t+1)=\min\left\{\left[\bar T_n(t+1)-\nu^*\right]^+,1\right\}$,
% \begin{align}
%p_i(t+1)=\min\left\{\left[\bar p_i(t+1)-\nu^*\right]^+,1\right\}\label{eqn:proj-K}
%\end{align}
where $\nu^*$ satisfies $\sum_{n\in \mathcal F_2^c}\min\left\{\left[\bar T_n(t+1)-\nu^*\right]^+,1\right\}=K_2^c$.
 \STATE Set $t=t+1$ and go to Step 2.
%\ENDLOOP
\end{algorithmic}}\label{alg:local}
\end{algorithm}
%*********************full version**************************

Next, we consider the symmetric case, i.e., $\alpha_1=\alpha_2 = \alpha$, in the high SNR and user density region. In this case, we can easily verify that $q_{2, \infty}\left(\mathcal F_2^c,\mathbf{T}\right)=\sum_{n\in \mathcal F_2^c }\frac{a_nT_n}{\theta_{2,K_2^c}+\theta_{1,K_2^c}T_n}$  (given in Lemma~\ref{Lem:asym-perf-v2}) is convex and Slater's condition is satisfied, implying that strong duality holds. Using KKT conditions, we can obtain the closed-form solution to the optimization in \eqref{eqn:prob-p-asymp-F_2^c} in this case.

\begin{Lem} [Asymptotically Optimal Solution when $\alpha_1=\alpha_2$]
For given $\mathcal F_2^c$, when $\alpha_1=\alpha_2 =\alpha$,   $\frac{P}{N_0}\to \infty$ and $\lambda_u\to \infty$,  the optimal solution to the optimization in \eqref{eqn:prob-p-asymp-F_2^c} is given by
 \begin{align}
T_n^*(\mathcal F_2^c)=\min\left\{\left[\frac{1}{\theta_{1,K_2^c}}\sqrt{\frac{a_n\theta_{2,K_2^c}}{\nu^*}}-\frac{\theta_{2,K_2^c}}{\theta_{1,K_2^c}}\right]^+,1\right\} , \ n\in \mathcal F_2^c,\label{eqn:opt-K-infty}
\end{align}
where $[x]^+\triangleq \max\{x,0\}$ and $\nu^*$ satisfies
\begin{align}
\sum_{n\in \mathcal F_2^c}\min\left\{\left[\frac{1}{\theta_{1,K_2^c}}\sqrt{\frac{a_n\theta_{2,K_2^c}}{\nu^*}}-\frac{\theta_{2,K_2^c}}{\theta_{1,K_2^c}}\right]^+,1\right\}=K_2^c.\label{eqn:opt-K-infty-v}
\end{align}
Here, $\theta_{1,k}$ and $\theta_{2,k}$ are given by \eqref{eqn:c_2_1_k} and \eqref{eqn:c_2_2_k}, respectively.
\label{Lem:solu-opt-infty}
\end{Lem}
\begin{proof}
Please refer to Appendix F.
\end{proof}

As the water-level in the traditional water-filling power control,   the root $\nu^*$ to the equation in \eqref{eqn:opt-K-infty-v} can be easily solved. Thus, by Lemma~\ref{Lem:solu-opt-infty}, we can efficiently compute   $\mathbf T^*(\mathcal F_2^c)$ when $\alpha_1=\alpha_2$.

Lemma~\ref{Lem:solu-opt-infty} can be interpreted as follows. As illustrated in Fig.~\ref{Fig:optstructure},  $\mathbf T^*(\mathcal F_2^c)$ given by Lemma~\ref{Lem:solu-opt-infty}  has a reverse water-filling structure. The file popularity distribution $\{a_n:n\in\mathcal F_2^c\}$ and the physical layer parameters (captured in  $\theta_{1,K_2^c}$ and $\theta_{2,K_2^c}$)  jointly affect   $\nu^*$.  Given  $\nu^*$, the physical layer parameters (captured in  $\theta_{1,K_2^c}$ and $\theta_{2,K_2^c}$)  affect the  caching probabilities of all the files  in the same way, while the popularity of file $n\in\mathcal F_2^c$  (i.e., $a_n$) only affects the caching probability of file $n$  (i.e., $T^*_n$).
From Lemma~\ref{Lem:solu-opt-infty}, we know that for any $n_1,n_2\in\mathcal F_2^c$ such that $n_1<n_2$,  we have  $T_{n_1}^*>T_{n_2}^*$, as $a_{n_1}> a_{n_2}$. In other words, files in $\mathcal F_2^c$ of higher popularity get more storage resources in the 2nd tier. In addition, there may exist $\bar n\in\mathcal F_2^c$  such that $T_n^*>0$ for all  $n\in\mathcal F_2^c$ and $n<\bar n$, and $T_n^*=0$ for all  $n\in\mathcal F_2^c$ and $n\geq\bar n$. In other words, some files in $\mathcal F_2^c$ of lower popularity may not be stored in the 2nd tier.   For a popularity distribution with a heavy tail,  more  different files  in $\mathcal F_2^c$ can be stored in the 2nd tier.
%These features stem from the successful transmission probability  maximization and are similar to the water-filling structure of the optimal power control in the throughput maximization.
%The optimal solution has the reverse water-filling structure.
%Fig.~\ref{Fig:optstructure} (a) illustrates the optimality structure  in Theorem~\ref{Thm:solu-opt-1-infty}.

 \begin{figure}[t]
\begin{center}
 \includegraphics[width=8cm]{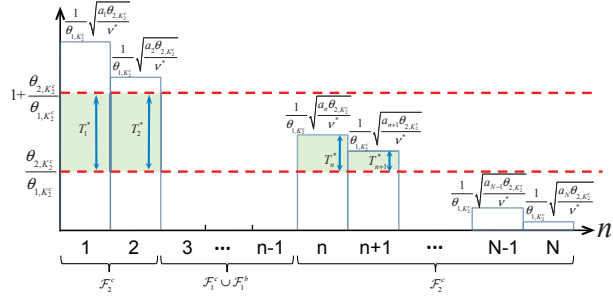}
  \end{center}
     \caption{\small{Illustration of the optimality structure in Lemma~\ref{Lem:solu-opt-infty}. In this example, $\mathcal F_1^c\cup\mathcal F_1^b=\{3,4\cdots,n-1\}$, and $\mathcal F_2^c=\{1,2,n,n+1,\cdots,N\}$ for some $n\in\mathcal N$.}}
\label{Fig:optstructure}
\end{figure}

\subsubsection{Discrete Optimization}
Given  $q_{2,\infty}^*(\mathcal F_2^c)=q_{2,\infty}(\mathcal F_2^c, \mathbf T^*(\mathcal F_2^c))$, the optimization in \eqref{eqn:prob-asymp-caching} is a discrete optimization. It can be shown that the number of possible choices for $(\mathcal F_1^{c}, \mathcal F_2^{c})$  satisfying \eqref{eqn:cache-constr} is given by $\sum_{F_2^c=K_2^c}^{N-K_1^c}\binom{N}{F_2^c}\binom{N-F_2^c}{K_1^c}=\Theta (N^N)$.
%\begin{align}
%\sum_{F_2^c=K_2^c}^{N-K_1^c}\binom{N}{F_2^c}\binom{N-F_2^c}{K_1^c}=\Theta (N^N).\label{eqn:high-order}
%\end{align}
Thus, a brute-force solution to the discrete optimization in \eqref{eqn:prob-asymp-caching}  is not acceptable.
%The key challenge  lies in solving the discrete optimization in \eqref{eqn:prob-asymp-caching} with reasonable complexity.
Now, we explore the  structural properties of the discrete optimization in \eqref{eqn:prob-asymp-caching} to facilitate the design of low-complexity asymptotically optimal  solutions.

\begin{Thm} [Structural Properties of Optimization in \eqref{eqn:prob-asymp-caching}]
There exists an optimal  solution $(\mathcal F_1^{c*},\mathcal F_2^{c*}) $ to the   optimization in \eqref{eqn:prob-asymp-caching} satisfying the following conditions:
(i) $F_2^{c*}\in\{F_{2,lb}^{c*}, F_{2,lb}^{c*}+1,\cdots, N-K_1^c\}$, where $F_{2,lb}^{c*}\triangleq \max\{K_2^c, N-K_1^c-K_1^b\}$;  and
(ii) there exists $n_1^c\in\{1,2,\cdots, F_2^{c*}+1\}$, such that   $\mathcal F_1^{c*}=\left\{n_1^c,n_1^c+1,\cdots,n_1^c+K_1^{c}-1\right\}$ and $\mathcal F_2^{c*}=\mathcal N\setminus  \left(\mathcal F_1^{c*}\cup\mathcal F_1^{b*}\right)$, where $\mathcal F_1^{b*}=\left\{ n_1^c+K_1^c, n_1^c+K_1^c+1,\cdots, n_1^c+K_1^c+N-(K_1^c+F_2^{c*})-1\right\}$.\label{Thm:opt-prop}
\end{Thm}
\begin{proof}
Please refer to Appendix G.
\end{proof}

%Property (i) implies that $F_1^{b*}\leq K_1^b$.

Theorem~\ref{Thm:opt-prop} can be interpreted as follows.
 Property  (ii) indicates that there is an optimal solution $(\mathcal F_1^{c*},\mathcal F_2^{c*}) $ to the optimization in \eqref{eqn:prob-asymp-caching} satisfying  that the files in  $\mathcal F_1^{c*}$, $\mathcal F_1^{b*}$ and $\mathcal F_1^{c*}\cup\mathcal F_1^{b*}$ are consecutive, and the files in $\mathcal F_1^{c*}$ are more popular than those in  $\mathcal F_1^{b*}$. This can be easily understood  from  Fig.~\ref{fig:cachestru}. It can be shown that the number of possible choices  for $(\mathcal F_1^{c}, \mathcal F_2^{c})$ satisfying the properties   in Theorem~\ref{Thm:opt-prop} is given by $\sum_{F_2^c=F_{2,lb}^{c*}}^{N-K_1^c}\sum_{n_1^c=1}^{F_2^c+1}1=\Theta (N^2)$,
%\begin{align}
%\sum_{F_2^c=F_{2,lb}^{c*}}^{N-K_1^c}\sum_{n_1^c=1}^{F_2^c+1}1=\Theta (N^2)\label{eqn:low-order}
%\end{align}
which is much smaller than the number of possible choices just satisfying \eqref{eqn:cache-constr}, i.e., $\Theta (N^N)$.
By restricting  to the choices for $(\mathcal F_1^c,\mathcal F_2^c)$ satisfying the properties in Theorem~\ref{Thm:opt-prop}, we can greatly reduce the complexity for solving the optimization in \eqref{eqn:prob-asymp-caching} without losing any optimality.

%The low-complexity optimal solution can be obtained by  Algorithm~\ref{alg}, which is of complexity $\Theta (N^2)$.

 \begin{figure}[t]
\begin{center}
 \includegraphics[width=12cm]{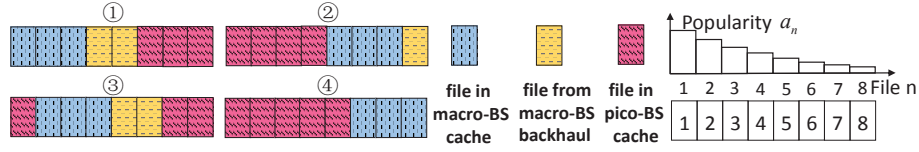}
  \end{center}
     \caption{\small{Illustration of the structural properties in Theorem~\ref{Thm:opt-prop} and Lemma~\ref{Lem:opt-prop}. $K_1^c=3$, $K_1^b=2$ and $K_2^c=2$.}}
\label{fig:cachestru}
\end{figure}

%\begin{algorithm}[htb]
%\caption{\small{Low-Complexity Asymptotically Optimal Solution}} %to Problem~\ref{prob:opt-asymp} (Problem~\ref{prob:opt-asymp-eq})}
%\small{\begin{algorithmic}[1]
%\STATE Initialize  $q_{\infty}^*=0$.
%     \FOR{$F_2^c=\max\{K_2^c,N-K_1^c-K_1^b\}:N-K_1^c$}
%     \FOR{$n_1^c=1:F_2^c+1$}
%     \STATE Choose $\mathcal F_1^{c}$ and $\mathcal F_2^{c}$ according to $\mathcal F_1^{c*}$ and $\mathcal F_2^{c*}$ in Property (b) (ii) of Theorem~\ref{Thm:opt-prop}.
%%      Let $\mathcal F_1^{c}=\left\{n_1^c,n_1^c+1,\cdots,n_1^c+K_1^{c}-1\right\}$ and $\mathcal F_2^{c}=\mathcal N\setminus  \left(\mathcal F_1^{c}\cup\mathcal F_1^{b}\right)$, where $\mathcal F_1^{b}=\left\{ n_1^c+K_1^c, n_1^c+K_1^c+1,\cdots, n_1^c+K_1^c+N-(K_1^c+F_2^{c})-1\right\}$.
%               \STATE Compute
%       $q_{\infty}\left(\mathcal F_1^c, \mathcal F_2^c,\mathbf T^*(\mathcal F_2^c)\right)\triangleq q_{\infty}$. If $q_{\infty}^*<q_{\infty}$, set $q_{\infty}^*=q_{\infty}$ and $(\mathcal F_1^{c^*},\mathcal F_2^{c^*},\mathbf T^*)=(\mathcal F_1^c,\mathcal F_2^c,\mathbf T^*(\mathcal F_2^c))$.
%       \ENDFOR
%       \ENDFOR
%\end{algorithmic}}\label{alg}
%\end{algorithm}

In some special cases, we can obtain extra  properties other than those in Theorem~\ref{Thm:opt-prop}.

\begin{Lem}[Structural Properties of Optimization in \eqref{eqn:prob-asymp-caching} in Special Cases](i) If $f_{1,K_1^c+K_1^b,\infty}\\>f_{2,K_2^c,\infty}(1)$, then $n_1^c$ in Theorem~\ref{Thm:opt-prop} satisfies  $n_1^c=1$;  (ii) If $f_{1,K_1^c,\infty}<f_{2,K_2^c,\infty}\left(\frac{K_2^c}{N-K_1^c}\right)$, then  $n_1^c$ in Theorem~\ref{Thm:opt-prop}  satisfies  $n_1^c\geq 2$. \label{Lem:opt-prop}
\end{Lem}
\begin{proof}
Please refer to Appendix H.
\end{proof}

Lemma~\ref{Lem:opt-prop} can be interpreted as follows.
Property (i) implies that  the most popular files are served by the 1st tier (cf. Case 1 in Fig.~\ref{fig:cachestru}), if $f_{1,K_1^c+K_1^b,\infty}>f_{2,K_2^c,\infty}(1)$. This condition holds when $\frac{P_1}{P_2}$  and $\frac{\lambda_1}{\lambda_2}$ are  above some thresholds, respectively. In this case, macro-BSs intend to  multicast the most popular files, as  they can offer relatively higher receive power, and hence higher successful transmission probability for the most popular files.  Note that when the condition in (i) holds,  by Theorem~\ref{Thm:opt-prop} and Lemma~\ref{Lem:opt-prop}, we can directly determine  $(\mathcal F_1^{c*}, \mathcal F_2^{c*})$. %the benefit  of multicast can be fully explored,  by allowing high power macro-BSs serve the most popular files.
Property (ii) implies that  the most popular file, i.e., file 1,  is not served by the 1st tier (cf. Cases 2-4  in Fig.~\ref{fig:cachestru}), if $f_{1,K_1^c,\infty}<f_{2,K_2^c,\infty}\left(\frac{K_2^c}{N-K_1^c}\right)$. This condition holds when $\frac{P_1}{P_2}$  and $\frac{\lambda_1}{\lambda_2}$ are below   some thresholds, respectively. In this case, pico-BSs intend to  multicast the most popular file, as  they can offer relatively higher receive power, and hence higher successful transmission probability for the most popular file. When the condition in (ii) holds, we can use Lemma~\ref{Lem:opt-prop}  together with Theorem~\ref{Thm:opt-prop} to reduce the set of  possible choices for $(\mathcal F_1^c,\mathcal F_2^c)$,  to further reduce the complexity for solving the optimization in \eqref{eqn:prob-asymp-caching}  without losing any optimality.
%%Property (ii) can be used to further reduce the complexity of Algorithm~\ref{alg} by reducing the number of inner loops.
%By restricting  to the choices for $(\mathcal F_1^c,\mathcal F_2^c)$ satisfying Property (b) (i) and (ii), we can greatly reduce the complexity for solving Problem~\ref{prob:opt-asymp-eq} without losing any optimality.
%Property (ii) can be used to further reduce the set of possible choices for $(\mathcal F_1^c,\mathcal F_2^c)$, so that we can greatly reduce the complexity for solving Problem~\ref{prob:opt-asymp-eq} without losing any optimality.
%the benefit  of multicast can be fully explored,  by allowing high power macro-BSs serve the most popular files.
%On the other hand, note that Theorem~\ref{Thm:opt-prop}  and Lemma~\ref{Lem:opt-prop}  can also provide design insights for low-complexity suboptimal algorithms.

\subsection{Near Optimal Solution in General Region}

First, we consider the near optimal solution for the continuous part (for given $\mathcal F_2^c$).
As illustrated  in Section~\ref{Sec:new-opt}-A, based on $\mathbf T^*(\mathcal F_2^c)$ obtained  using  Algorithm~\ref{alg:local} or  Lemma~\ref{Lem:solu-opt-infty} (when $\alpha_1=\alpha_2$), we can determine $\mathcal P^*(\mathcal F_2^c)$.
As illustrated in Section~\ref{subsec:asym-opt-prob}, any $\mathbf p^*\in \mathcal P^*(\mathcal F_2^c)$  is an optimal solution to the optimization in \eqref{eqn:prob-p-asymp-F_2^c-p}. In other words, for given $\mathcal F_2^c$, we have a set of asymptotically optimal solutions in  the high SNR and user density region.
Substituting $\mathbf p^*$ satisfying  \eqref{eqn:def-T-n*} into $q_2\left(\mathcal F_2^c,\mathbf p\right)$ in Theorem~\ref{Thm:generalKmulti}, we have $q_2\left(\mathcal F_2^c, \mathbf p^*,\mathbf T^*\right)$ given in \eqref{eq:linear_prog}.
%\small{\begin{align}
%q_2\left(\mathcal F_2^c,\mathbf p^*\right)=&\sum_{n\in \mathcal F_2^c}a_{n} \sum_{k^c=1}^{K_2^c} f_{2,k^c}(T_n^*) \sum_{i\in \mathcal I_n}\frac{p_i^*}{T_n^*}\sum_{\mathcal X\in g\left(\mathcal N_{i,-n},k^c-1\right) }\prod\limits_{m\in \mathcal X}\left(1-\eta\left(a_m\lambda_u,T_m^*\lambda_2\right)\right)\prod\limits_{m\in {\mathcal N_{i,-n}\setminus \mathcal X}}\eta\left(a_m\lambda_u,T_m^*\lambda_2\right)\nonumber\\
%=&\sum_{n\in \mathcal F_2^c}\frac{a_{n}}{T_n^*} \sum_{i\in \mathcal I_n}p_i^* \sum_{k^c=1}^{K_2^c} f_{2,k^c}(T_n^*) \sum_{\mathcal X\in g\left(\mathcal N_{i,-n},k^c-1\right) }\prod\limits_{m\in \mathcal X}\left(1-\eta\left(a_m\lambda_u,T_m^*\lambda_2\right)\right)\prod\limits_{m\in {\mathcal N_{i,-n}\setminus \mathcal X}}\eta\left(a_m\lambda_u,T_m^*\lambda_2\right)\nonumber\\
%=&\sum_{i\in \mathcal I} \left(\sum_{n\in \mathcal N_i}\frac{a_{n}}{T_n^*}  \sum_{k^c=1}^{K_2^c} f_{2,k^c}(T_n^*) \sum_{\mathcal X\in g\left(\mathcal N_{i,-n},k^c-1\right) }\prod\limits_{m\in \mathcal X}\left(1-\eta\left(a_m\lambda_u,T_m^*\lambda_2\right)\right)\prod\limits_{m\in {\mathcal N_{i,-n}\setminus \mathcal X}}\eta\left(a_m\lambda_u,T_m^*\lambda_2\right)\right)p_i^*.\nonumber\\
%&\triangleq q_2\left(\mathcal F_2^c, \mathbf p^*,\mathbf T^*\right)\nonumber
%\end{align}}
%\normalsize
\begin{figure*}[!t]
\small{\begin{align}
&q_2\left(\mathcal F_2^c,\mathbf p^*\right)\nonumber\\
%=&\sum_{n\in \mathcal F_2^c}a_{n} \sum_{k^c=1}^{K_2^c} f_{2,k^c}(T_n^*) \sum_{i\in \mathcal I_n}\frac{p_i^*}{T_n^*}\sum_{\mathcal X\in g\left(\mathcal N_{i,-n},k^c-1\right) }\prod\limits_{m\in \mathcal X}\left(1-\eta\left(a_m\lambda_u,T_m^*\lambda_2\right)\right)\prod\limits_{m\in {\mathcal N_{i,-n}\setminus \mathcal X}}\eta\left(a_m\lambda_u,T_m^*\lambda_2\right)\nonumber\\
&=\sum_{i\in \mathcal I} \left(\sum_{n\in \mathcal N_i}\frac{a_{n}}{T_n^*}  \sum_{k^c=1}^{K_2^c} f_{2,k^c}(T_n^*) \sum_{\mathcal X\in  \left\{\mathcal S \subseteq \mathcal N_{i,-n}: |\mathcal S|=k^c-1\right\}  }\prod\limits_{m\in \mathcal X}\left(1-\left(1+\frac{a_m\lambda_u}{3.5T_m^*\lambda_2}\right)^{-4.5}\right)\prod\limits_{m\in {\mathcal N_{i,-n}\setminus \mathcal X}}\left(1+\frac{a_m\lambda_u}{3.5T_m^*\lambda_2}\right)^{-4.5}\right)p_i^*.\nonumber\\
&\triangleq q_2\left(\mathcal F_2^c, \mathbf p^*,\mathbf T^*\right)\label{eq:linear_prog}
\end{align}}
\normalsize \hrulefill
\end{figure*}
%\begin{align}
%q_K\left(\mathbf{p}\right)=\sum_{n\in \mathcal N}\frac{a_{n}}{T_n^*} \sum_{i\in \mathcal I_n}p_i \left(\sum_{k=1}^KQ_{i,n,k}(\mathbf T_{i,-n}^*) f_k(T_n^*)\right)
%\end{align}
For given $\mathcal F_2^c$ (and $\mathbf T^*(\mathcal F_2^c)$), we would like to obtain the best asymptotically optimal solution which maximizes the successful transmission probability  $q_2\left(\mathcal F_2^c, \mathbf p^*,\mathbf T^*\right)$  in the general region among all the asymptotically optimal solutions in $ \mathcal P^*(\mathcal F_2^c)$.
\begin{Prob} [Optimization of $\mathbf p^*$ under Given $\mathcal F_2^c$ (and  $\mathbf T^*$)]
\begin{align}q^{\dagger}_{2}(\mathcal F_2^c)\triangleq\max_{\mathbf p^*} &\quad q_{2}(\mathcal F_2^c,\mathbf p^*,\mathbf T^*)\nonumber\\
s.t. & \quad  \eqref{eqn:cache-constr-indiv}, \eqref{eqn:cache-constr-sum},\eqref{eqn:def-T-n*}.\nonumber
\end{align}
The optimal solution is denoted as $\mathbf p^{\dagger}(\mathcal F_2^c)$.\label{Prob:asm-improvement}
\end{Prob}

Problem~\ref{Prob:asm-improvement} is a linear programming problem.
%The optimal solution to Problem~\ref{Prob:asm-improvement}  can be obtained using the simplex method.
To reduce the complexity for solving Problem~\ref{Prob:asm-improvement},  we first derive some caching probabilities which are zero based on the relationship between $\mathbf p$ and $\mathbf T$.
In particular, for all $i\in\mathcal I_n$ and $n\in\{n\in\mathcal F_2^c: T_n^*=0\}$, we have $p_i^{\dagger}({\mathcal F_c^c})=0$; for all $i\not\in\mathcal I_n$ and $n\in\{n\in\mathcal F_2^c: T_n^*=1\}$, we have $p_i^{\dagger}({\mathcal F_c^c})=0$. Thus, we have
\begin{align}
p_i^{\dagger}({\mathcal F_c^c})=0, \ i\in\mathcal I', \label{eqn:p_i-0}
\end{align}
where  $\mathcal I'\triangleq \cup_{n\in\{n\in\mathcal F_2^c: T_n^*=0\}}\mathcal I_n\cup\left(\mathcal I\setminus \cup_{n\in\{n\in\mathcal F_2^c: T_n^*=1\}}\mathcal I_n\right)$.
Then, we can compute  the remaining caching probabilities for the combinations in $\mathcal I\setminus \mathcal I'$ using the simplex method (refer to Step 6 of Algorithm~\ref{alg:sympK} for details). Therefore, using the above approach, for given $\mathcal F_2^c$, we can obtain the best asymptotically optimal solution $\mathbf p^{\dagger}(\mathcal F_2^c)$  in $\mathcal P^*(\mathcal F_2^c)$ to the optimization in \eqref{eqn:prob-p-F_2^c}.

Next, we consider the near optimal solution for the discrete part. Specifically, after obtaining $q^{\dagger}_{2}(\mathcal F_2^c)$ using the above approach for the continuous part, we consider the optimization of $ q_{1}(\mathcal F_1^c,\mathcal F_2^c)+q^{\dagger}_{2}(\mathcal F_2^c)$ over the set of $(\mathcal F_1^c, \mathcal F_2^c)$ satisfying Theorem~\ref{Thm:opt-prop} (and Lemma~\ref{Lem:opt-prop}). Let $(\mathcal F_1^{c\dagger}, \mathcal F_2^{c\dagger})$ and $q^{\dagger}$ denote the optimal solution and the optimal value.

Finally, combining the above discrete part and continuous part, we can obtain the near optimal solution $(\mathcal F_1^{c\dagger}, \mathcal F_2^{c\dagger}, \mathbf p^{\dagger}(\mathcal F_2^{c\dagger}))$ to  Problem~\ref{prob:opt} (Problem~\ref{prob:opt-eq}), as   summarized in Algorithm~\ref{alg:sympK}. We can
show that  in the general region, under a mild condition (i.e., $f_{2,k,\infty}(x)$ is convex),  the  near optimal solution $(\mathcal F_1^{c\dagger}, \mathcal F_2^{c\dagger}, \mathbf p^{\dagger}(\mathcal F_2^{c\dagger}))$ obtained by Algorithm~\ref{alg:sympK} achieves the successful transmission probability  $q^{\dagger}=q(\mathcal F_1^{c\dagger}, \mathcal F_2^{c\dagger}, \mathbf p^{\dagger}(\mathcal F_2^{c\dagger}))$   greater than or equal to that of any optimal solution to Problem~\ref{prob:opt-asymp-eq}, i.e., any  asymptotically optimal solution to Problem~\ref{prob:opt} (Problem~\ref{prob:opt-eq}).

\begin{Lem} We have  $q(\mathcal F_1^{c\dagger}, \mathcal F_2^{c\dagger}, \mathbf p^{\dagger}(\mathcal F_2^{c\dagger}))\geq q(\mathcal F_1^{c*}, \mathcal F_2^{c*}, \mathbf p^{*}(\mathcal F_2^{c*}))$, for all $\mathbf p^*(\mathcal F_2^{c*})\in\mathcal P^*(\mathcal F_2^{c*})$,
where $(\mathcal F_1^{c*}, \mathcal F_2^{c*}, \mathbf p^{*}(\mathcal F_2^{c*}))$  is an optimal solution to Problem~\ref{prob:opt-asymp-eq}.\end{Lem}

\begin{algorithm} \caption{Near Optimal Solution}
\small{\begin{algorithmic}[1]
\STATE Initialize  $q^{\dagger}=0$.
     \FOR{$F_2^c=\max\{K_2^c,N-K_1^c-K_1^b\}:N-K_1^c$}
     \FOR{$n_1^c=1:F_2^c+1$}
     \STATE Choose $\mathcal F_1^{c}$ and $\mathcal F_2^{c}$ according to  Theorem~\ref{Thm:opt-prop} (and Lemma~\ref{Lem:opt-prop}).
%      Let $\mathcal F_1^{c}=\left\{n_1^c,n_1^c+1,\cdots,n_1^c+K_1^{c}-1\right\}$ and $\mathcal F_2^{c}=\mathcal N\setminus  \left(\mathcal F_1^{c}\cup\mathcal F_1^{b}\right)$, where $\mathcal F_1^{b}=\left\{ n_1^c+K_1^c, n_1^c+K_1^c+1,\cdots, n_1^c+K_1^c+N-(K_1^c+F_2^{c})-1\right\}$.
\STATE Obtain the optimal solution  $\mathbf T^*(\mathcal F_2^c)$  to the optimization in \eqref{eqn:prob-p-asymp-F_2^c}  using  Algorithm~\ref{alg:local} or   Lemma~\ref{Lem:solu-opt-infty} (when $\alpha_1=\alpha_2$).
\STATE  Determine $\mathcal I'$ and choose $p_i^{\dagger}(\mathcal F_2^c)=0$ for all $i\in \mathcal I'$ according to \eqref{eqn:p_i-0}. Then, obtain $\{p_i^{\dagger}: i\in \mathcal I\setminus \mathcal I'\}$ and $q^{\dagger}_2(\mathcal F_2^c)$ by solving Problem~\ref{Prob:asm-improvement} (under the constraint in  \eqref{eqn:p_i-0}) using  the simplex method.
               \STATE Compute
       $q_1(\mathcal F_1^c,\mathcal F_2^c)+q_{2}^{\dagger}\left(\mathcal F_2^c\right)\triangleq q_{\infty}$. If $q_{\infty}^{\dagger}<q_{\infty}$, set $q_{\infty}^{\dagger}=q_{\infty}$ and $(\mathcal F_1^{c\dagger},\mathcal F_2^{c\dagger},\mathbf p^{\dagger}(\mathcal F_2^{c\dagger}))=(\mathcal F_1^c,\mathcal F_2^c,\mathbf p^{\dagger}(\mathcal F_2^c))$.
       \ENDFOR
       \ENDFOR
\end{algorithmic}}\label{alg:sympK}
\end{algorithm}

\section{Numerical Results}\label{Sec:simu}

\begin{figure}[t]
\begin{center}
 \subfigure[\small{%Successful transmission probability versus
Cache size  at $\gamma=1$ and $\lambda_u=5\times10^{-5}$.}]
 {\resizebox{6cm}{!}{\includegraphics{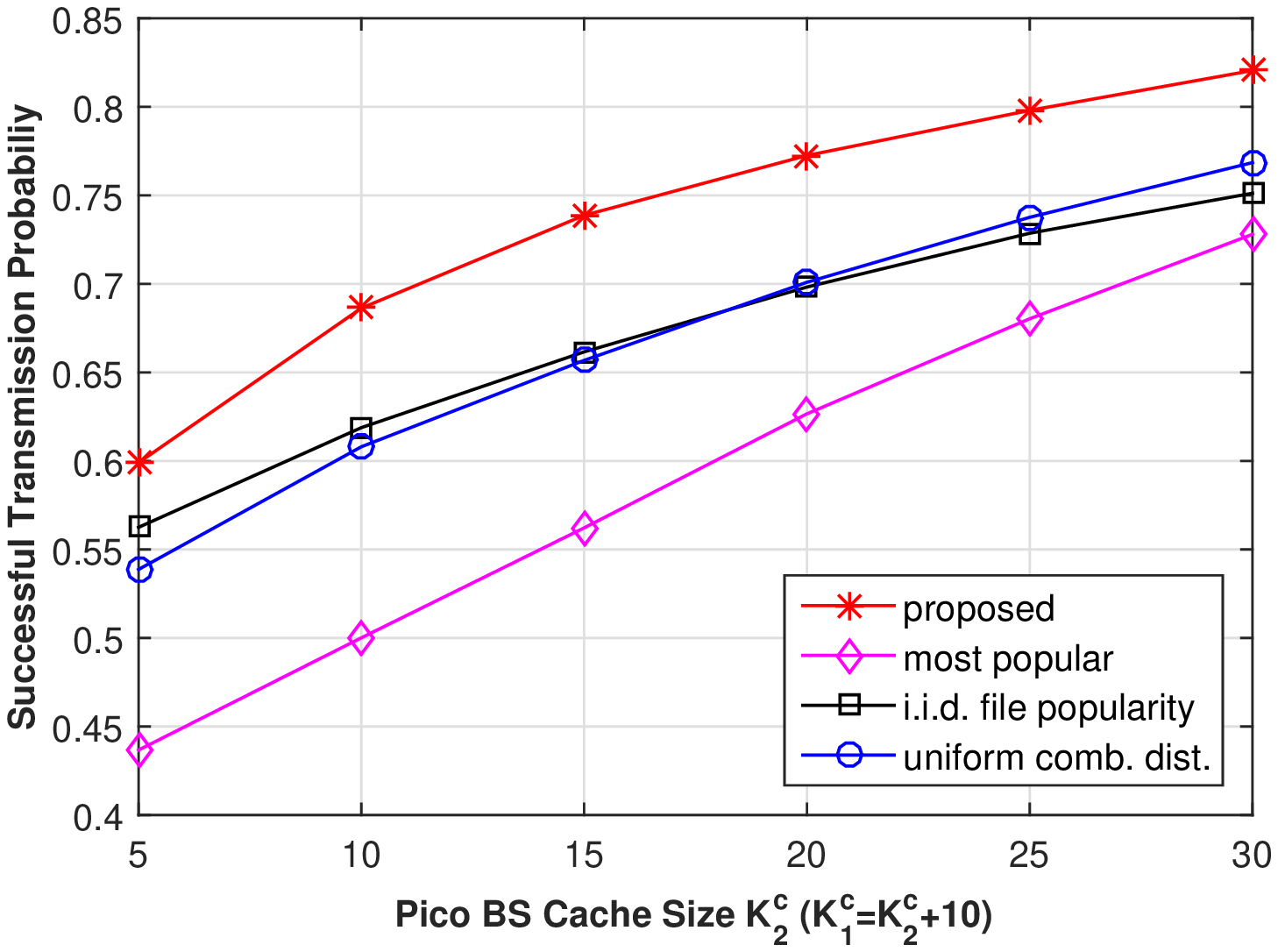}}}\quad\quad
 \subfigure[\small{%Successful transmission probability versus
Zipf exponent at $K_2^c=10$ and $\lambda_u=5\times10^{-5}$.}]
 {\resizebox{6cm}{!}{\includegraphics{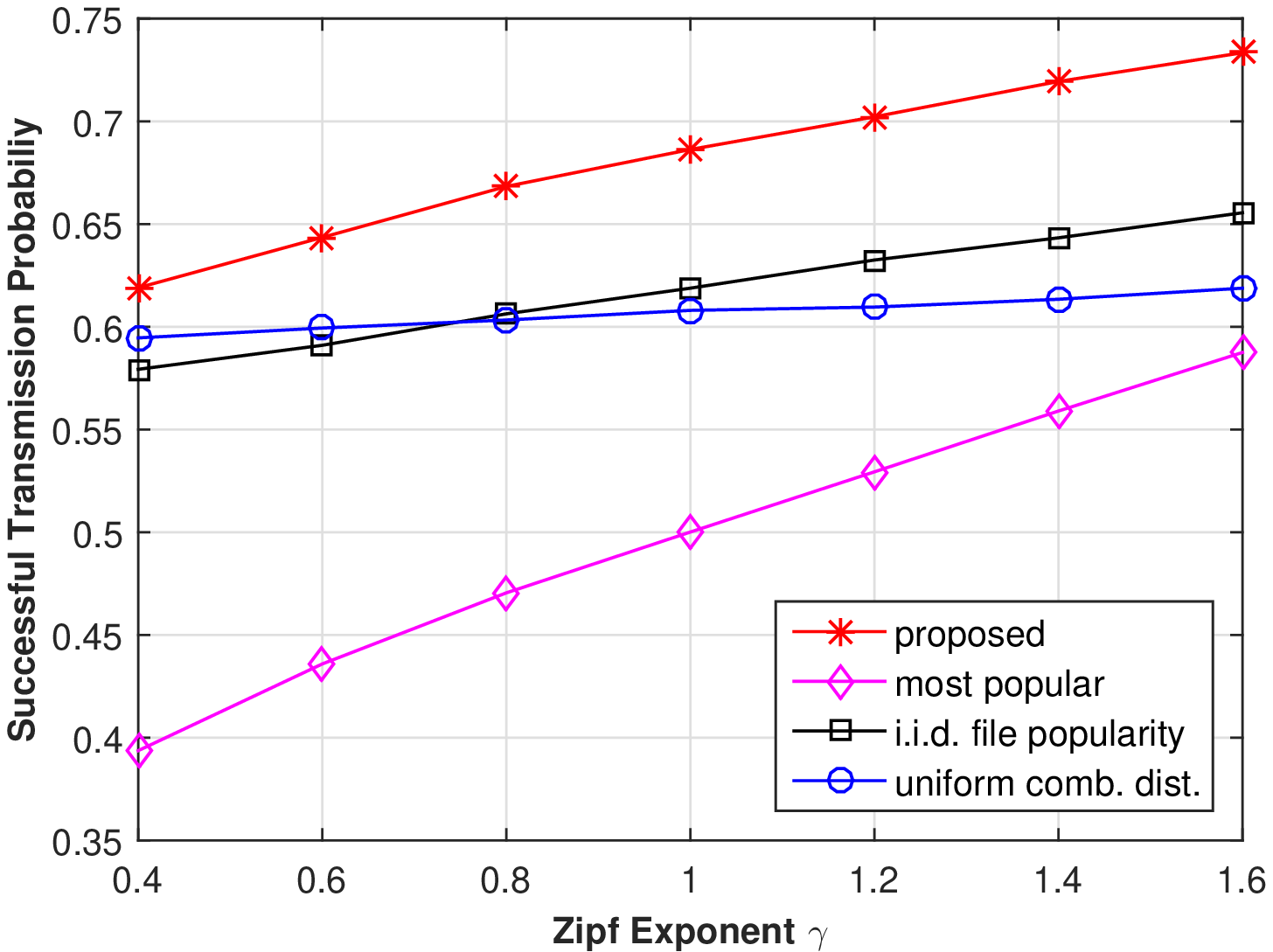}}}\quad\quad
% \subfigure[\small{%Successful transmission probability versus
%Pico-BS density  at $K_2^c=10$, $\gamma=1$, $\lambda_{1}=5\times10^{-7}$ and $\lambda_u=5\times10^{-5}$.}]
% {\resizebox{7cm}{!}{\includegraphics{fig/Pr_lambda_sb.eps}}}\quad\quad
 \subfigure[\small{%Successful transmission probability versus
User density  at $K_2^c=10$ and $\gamma=1$.}]
 {\resizebox{6cm}{!}{\includegraphics{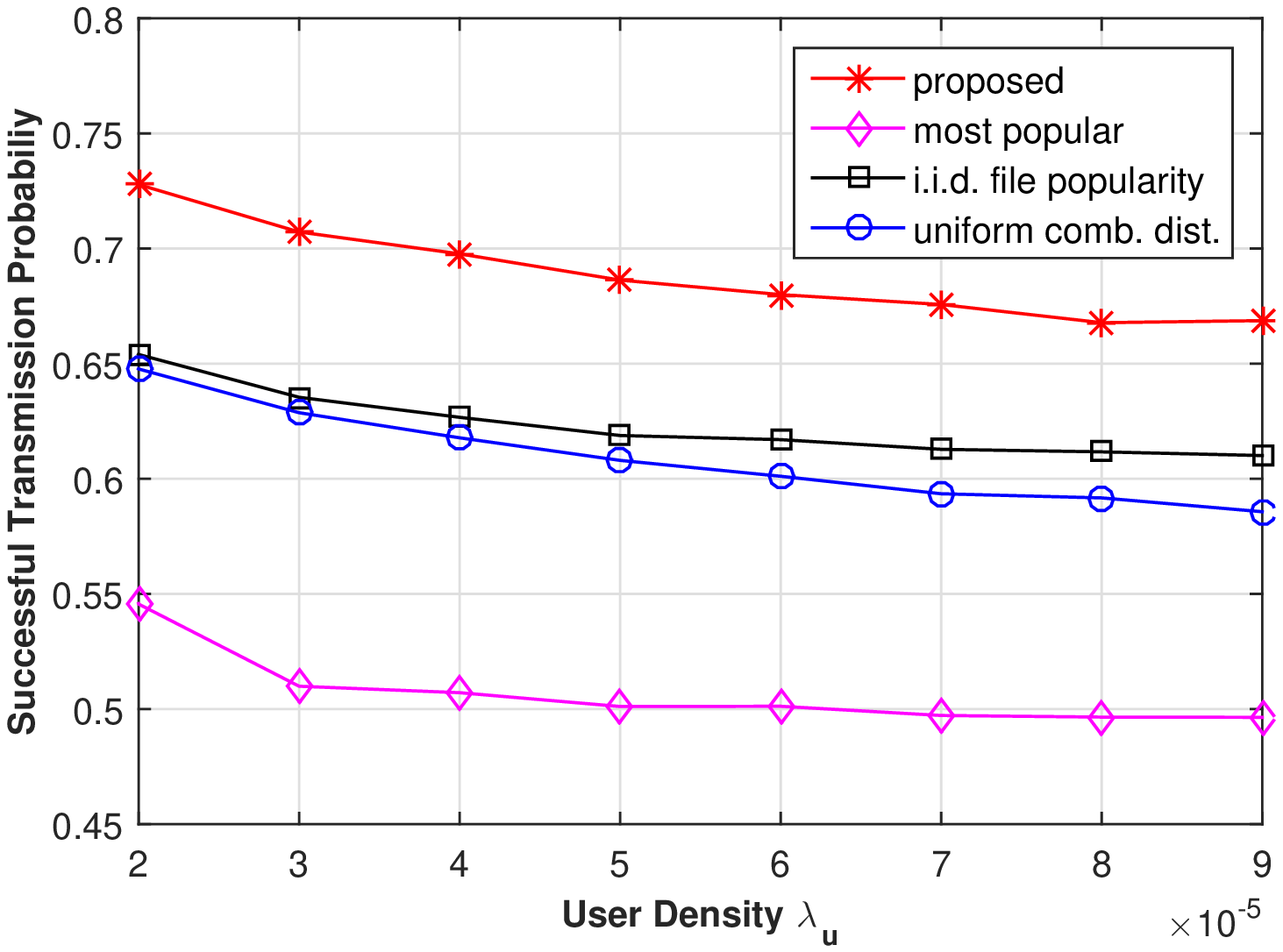}}}
 \end{center}
   \caption{\small{Successful transmission probability versus cache size $K_1^c$ and $K_2^c$, Zipf exponent $\gamma$  and user density $\lambda_u$. $K_1^c=K_2^c+10$, $K_1^b=15$, $\lambda_{1}=5\times10^{-7}$, $\lambda_{2}=3\times10^{-6}$, $\frac{P_1}{P_2}=16$dB, $\alpha_1=\alpha_2=4$, $W = 20\times 10^6$, $\tau =  2\times10^4$ and $N = 100$.}}
\label{fig:simulation-large}
\end{figure}

In this section, we compare the proposed near optimal  design given by Algorithm~\ref{alg:sympK} with three  schemes. %\footnote{When $N=1000$, the complexity of Local Opt. is not acceptable. Thus, in Fig.~\ref{fig:simulation-large}, we only consider the proposed Asymp. Opt. and the three baseline schemes.}
Baseline 1 (most popular) refers to the design in which each macro-BS selects the most  $K_1^c+K_1^b$ popular files to store and fetch, and each pico-BS selects the most $K_2^c$ popular  files to store \cite{EURASIP15Debbah,LiuYangICC16,Yang16}.
Baseline 2 (i.i.d. file popularity) refers to the design in which each macro-BS selects $K_1^c+K_1^b$ files to store and fetch,  and each pico-BS selects $K_2^c$ files to store,  in an i.i.d. manner with  file $n$ being selected with probability $a_n$  \cite{DBLP:journals/corr/BharathN15}.  Note that under this scheme, each (macro or pico) BS may cache multiple copies of one file, leading to storage waste, and each macro-BS may fetch multiple copies of one file, leading to backhaul waste.  Baseline 3 (uniform comb. dist.) refers to the design in which each macro-BS randomly selects a combination of $K_1^c+K_1^b$ different files  to store  and fetch, and each pico-BS randomly selects  a combination of $K_2^c$ different files to store,  according to the uniform distribution \cite{TamoorComLett16,DBLP:journals/corr/Tamoor-ul-Hassan15}. Under the three baseline schemes, each user requesting file $n$ is associated with the BS which stores file $n$ and offers the maximum long-term average receive power at this user. In addition, the three baseline schemes also adopt the same multicasting scheme as in our design.
In the simulation, we assume the popularity follows Zipf distribution, i.e., $a_n=\frac{n^{-\gamma}}{\sum_{n\in \mathcal N}n^{-\gamma}}$, where $\gamma$ is the Zipf exponent.
%Note that a small $\gamma$ means a heavy-tail popularity distribution.

Fig.~\ref{fig:simulation-large} illustrates   the successful transmission probability versus different parameters. From Fig.~\ref{fig:simulation-large}, we can observe that the proposed design outperforms all the three baseline schemes. In addition, the proposed design, Baseline 2 and Baseline 3 have much better performance than Baseline 1, as they provide  file diversity to improve the network performance, when the storage and backhaul resources are limited and the cache-enabled HetNet with backhaul constraints may not be able to satisfy all  file requests.
%The performance gap between the proposed design and Baseline 1 is relatively large at small cache size $K$, small Zipf exponent $\gamma$, large BS density $\lambda_b$ and small user density $\lambda_u$. This demonstrates the benefit of file diversity in these regions.
%Please note that the successful transmission probability in Fig.~\ref{fig:simulation-large} is small due to the low file availability when $K/N$ is small. This performance is achieved without backhaul cost or extra delay cost. Our goal here is to study how cache itself can affect the network performance.

Specifically,
Fig.~\ref{fig:simulation-large} (a)  illustrates   the successful transmission probability versus  the cache sizes $K_1^c$ and $K_2^c$.
We can see that the performance of all the schemes increases with $K_1^c$ and $K_2^c$. This is because as $K_1^c$ and $K_2^c$ increase, each BS can store more files, and the probability that a randomly requested file is cached at a nearby BS increases. Fig.~\ref{fig:simulation-large} (b)  illustrates   the  successful transmission probability versus  the Zipf exponent $\gamma$. We can observe that the performance of the proposed design, Baseline 1 and Baseline 2 increases with the Zipf exponent $\gamma$ faster than Baseline 3. This is because when $\gamma$ increases, the tail of popularity distribution becomes small, and hence, the average network file load decreases. The performance increase of Baseline 3 with $\gamma$ only comes from the decrease of the average network file load. While, under  the proposed design, Baseline 1 and Baseline 2, the probability that a randomly requested file is cached at a nearby BS increases with $\gamma$. Thus,
the performance increases of the proposed design, Baseline 1  and Baseline 2 with $\gamma$ are due to  the decrease of the  average network file load and the increase of the chance of a requested file being cached at a nearby BS.
Fig.~\ref{fig:simulation-large} (c)  illustrates   the successful transmission probability versus  the user density $\lambda_u$. We can see that the performance of all the schemes decreases with $\lambda_u$. This is because the probability of a cached file being requested by at least one user  increases, as $\lambda_u$ increases.

\section{Conclusion}
 In this paper, we considered  the analysis and optimization of caching and multicasting  in a large-scale  cache-enabled HetNet with backhaul constraints. We proposed a  hybrid caching design and  a corresponding multicasting design to  provide high spatial file diversity and ensure efficient content dissemination.
Utilizing tools from stochastic geometry, we  analyzed  the successful transmission probability in the general region and the asymptotic region. Then, we formulated a mixed discrete-continuous optimization problem to maximize the  successful transmission probability by optimizing the design parameters. By exploring  the structural properties,  we  obtained a near optimal solution with superior performance and manageable  complexity, based on a two-step optimization framework. The analysis and optimization results offered valuable design insights for practical cache-enabled HetNets.

%\newpage

\section*{Appendix A: Proof of Lemma~\ref{Lem:pmf-K-m}}

When typical user $u_0$ requests file $n\in\mathcal F_1^c \cup \mathcal F_1^b $, let random variable $Y_{m,n}\in\{0,1\}$ denote whether file $m\in \mathcal F_1^c  \cup \mathcal F_1^b \setminus \{n\}$ is requested by the users associated with serving macro-BS $\ell_{0}$. Specifically, when $u_0$ requests file $n\in\mathcal F_1^c $, we have $K_{1,n,0}^c=1+\sum_{m\in \mathcal F_1^c\setminus\{n\}} Y_{m,n}$ and $\overline{K}_{1,n,0}^{b}=\sum_{m\in \mathcal F_1^b} Y_{m,n}$. When $u_0$ requests file $n\in\mathcal F_1^b $, we have $\overline{K}_{1,n,0}^{c}=\sum_{m\in \mathcal F_1^c} Y_{m,n}$ and $K_{1,n,0}^b=1+\sum_{m\in \mathcal F_1^b \setminus \{n\}} Y_{m,n}$. Thus, we have
\small{\begin{align}
&\Pr \left[K_{1,n,0}^c=k^c\right]=\sum_{\mathcal X\in g(\mathcal F_{1,-n}^c,k^c-1) }\prod_{m\in \mathcal X}(1-\Pr[Y_{m,n}=0])\prod_{m\in \mathcal F_{1,-n}^c \setminus\mathcal X}\Pr[Y_{m,n}=0],\quad k^c=1 \cdots K_1^c,\nonumber\\
&\Pr \left[\overline{K}_{1,n,0}^b=k^b\right]=\sum_{\mathcal X\in g(\mathcal F_1^b,k^b) }\prod_{m\in \mathcal X}(1-\Pr[Y_{m,n}=0])\prod_{m\in \mathcal F_{1}^b \setminus\mathcal X}\Pr[Y_{m,n}=0],\quad k^b=0 \cdots F_1^b,\nonumber\\
&\Pr \Big[\overline{K}_{1,n,0}^c=k^c\Big]=\sum_{\mathcal X\in g(\mathcal F_1^c,k^c) }\prod_{m\in \mathcal X}(1-\Pr[Y_{m,n}=0])\prod_{m\in \mathcal F_{1}^c \setminus\mathcal X}\Pr[Y_{m,n}=0],\quad k^c=0 \cdots K_1^c,\nonumber\\
&\Pr \left[K_{1,n,0}^b=k^b\right]=\sum_{\mathcal X\in g(\mathcal F_{1,-n}^b,k^b-1) }\prod_{m\in \mathcal X}(1-\Pr[Y_{m,n}=0])\prod_{m\in \mathcal F_{1,-n}^b \setminus\mathcal X}\Pr[Y_{m,n}=0],\quad k^b=1 \cdots F_1^b.\nonumber
%\label{eqn:K-pmf-i}
\end{align}}
\normalsize

To prove \eqref{eqn:K-1-c}, \eqref{eqn:K-1-b-bar}, \eqref{eqn:K-1-c-bar} and \eqref{eqn:K-1-b}, it remains to  calculate $\Pr[Y_{m,n}=0]$.   The p.m.f.  of  $Y_{m,n}$  depends on the  p.d.f. of the size of the Voronoi cell of macro-BS $\ell_{0}$, i.e., the  p.d.f. of the size of the Voronoi cell to which a randomly chosen user belongs \cite{SGcellsize13}.
Thus, we can calculate the p.m.f. of  $Y_{m,n}$ using Lemma 3 of \cite{SGcellsize13} as follows
\small{\begin{align}
&\Pr[Y_{m,n}=0]
=\left(1+3.5^{-1}\frac{a_m\lambda_u}{\lambda_1}\right)^{-4.5}, \quad \ m\in \mathcal F_1^c \cup \mathcal F_1^b \setminus \{n\}.\label{eqn:y-pmf}
\end{align}}
\normalsize
Therefore, we complete the proof.

\section*{Appendix B: Proof of Lemma~\ref{Lem:pmf-K}}

When typical user $u_0$ requests file $n\in\mathcal F_2^c$, let random variable $Y_{m,n,i}\in\{0,1\}$ denote whether file $m\in \mathcal N_{i}\setminus \{n\}$ is requested by the users associated with serving pico-BS $\ell_{0}$ when pico-BS $\ell_{0}$ contains combination $i\in \mathcal I_n$. When $u_0$ requests file $n \in \mathcal F_2^c$ and serving pico-BS $\ell_{0}$ contains combination $i\in \mathcal I_n$, we have $K_{2,n,0}^c=1+\sum_{m\in \mathcal N_{i,-n}} Y_{m,n,i}$. Thus, we have
\small{\begin{align}
&\Pr \left[K_{2,n,0}^c=k^c|\text{pico-BS $\ell_0$ contains combination $i\in \mathcal I_n$}\right]\nonumber\\
=&\sum_{\mathcal X\in g\left(\mathcal N_{i,-n},k^c-1\right) }\prod_{m\in \mathcal X}(1-\Pr[Y_{m,n,i}=0])\prod_{m\in \mathcal N_{i,-n}\setminus\mathcal X}\Pr[Y_{m,n,i}=0],\quad k^c=1,\cdots, K_2^c. \label{eqn:K-pmf-i}
\end{align}}
\normalsize
The probability that pico-BS $\ell_{0}$ contains combination $i\in \mathcal I_n$ is $\frac{p_i}{T_n}$. Thus, by the law of total probability, we have
\small{\begin{align}
&\Pr \left[K_{2,n,0}^c=k^c\right]
=\sum_{i\in \mathcal I_n}\frac{p_i}{T_n}\Pr \left[K_{2,n,0}^c=k^c|\text{pico-BS $\ell_{0}$ contains combination $i\in \mathcal I_n$}\right],\quad k^c=1,\cdots, K_2^c.\nonumber%\label{eqn:proof-K-pmf}
\end{align}}
\normalsize
Thus, to prove \eqref{eqn:K-pmf}, it remains to  calculate $\Pr[Y_{m,n,i}=0]$.   The p.m.f.  of  $Y_{m,n,i}$ depends on the  p.d.f. of the size of the Voronoi cell of pico-BS $\ell_{0}$ w.r.t. file $m\in \mathcal N_{i,-n}$  when pico-BS $\ell_{0}$ contains combination $i\in \mathcal I_n$, which is unknown. We approximate this p.d.f. based on the known result of the  p.d.f. of the size of the Voronoi cell to which a randomly chosen user belongs \cite{SGcellsize13}.
Under this approximation, we can calculate the p.m.f. of  $Y_{m,n,i}$ using Lemma 3 of \cite{SGcellsize13} as follows
\small{\begin{align}
\Pr[Y_{m,n,i}=0]
=&\left(1+3.5^{-1}\frac{a_m\lambda_u}{T_m\lambda_2}\right)^{-4.5}, \quad \ m\in \mathcal N_{i,-n},\ i\in \mathcal I_n.\label{eqn:y-pmf}
\end{align}}
\normalsize

\section*{Appendix C: Proof of Theorem~\ref{Thm:generalKmulti}}

Based on~(\ref{eqn:succ-prob-def}), to prove Theorem~\ref{Thm:generalKmulti}, we calculate $q_{1}\left(\mathcal F_1^c,\mathcal F_2^c\right)$ and $q_{2}\left(\mathcal F_2^c, \mathbf p\right)$, respectively.

\subsection*{Calculation of $q_{1}\left(\mathcal F_1^c,\mathcal F_2^c\right)$}
When $u_0$ is a macro-user, as in the traditional connection-based HetNets, there are two types of interferers, namely, i) all the other macro-BSs besides its serving macro-BS, and ii) all the pico-BSs. Thus, we rewrite the SINR expression  in \eqref{eqn:SINR} as follows:
\small{\begin{align}
{\rm SINR}_{n,0}=&\frac{{D_{1,\ell_0,0}^{-\alpha_{1}}}\left|h_{1,\ell_0,0}\right|^{2}}{\sum_{\ell\in\Phi_{1}\backslash \{\ell_0\}}D_{1,\ell,0}^{-\alpha_{1}}\left|h_{1,\ell,0}\right|^{2}+\sum_{\ell\in\Phi_{2}}D_{2,\ell,0}^{-\alpha_{2}}\left|h_{2,\ell,0}\right|^{2}
\frac{P_{2}}{P_{1}}+\frac{N_{0}}{P_{1}}}=\frac{{D_{1,\ell_0,0}^{-\alpha_{1}}}\left|h_{1,\ell_0,0}\right|^{2}}{I_{1}+I_{2}\frac{P_2}{P_1}+\frac{N_0}{P_1}}, \ n\in\mathcal F_1^c\cup \mathcal F_1^b,
\label{eqn:SINR_v3}
\end{align}}
\normalsize
where
%$\Phi_{1}$ is the point process generated by macro-BSs and $\Phi_{2}$ is the point process generated by pico-BSs.
$I_{1}\triangleq\sum_{\ell\in\Phi_{1}\backslash \{\ell_{0}\}}D_{1,\ell,0}^{-\alpha_1}\left|h_{1,\ell,0}\right|^{2}$ and $I_{2}\triangleq \sum_{\ell\in\Phi_{2}}D_{2,\ell,0}^{-\alpha_2}\left|h_{2,\ell,0}\right|^{2}$.

Next, we calculate the conditional successful transmission probability of file $n\in \mathcal F_1^c \cup \mathcal F_1^b $ requested by $u_{0}$ conditioned on $D_{1,\ell_0,0}=d$ when the file load is $k$, i.e.,
%(let $k^c$ and $k^b$ denote file load of serving macro-BS from cache and backhaul respectively, $k=k^c+k^b$),
%$$q_{k,n,D_{1,\ell_0,0}}\left(d\right)\triangleq {\rm Pr}\left[\frac{W}{k}\log_{2}\left(1+{\rm SINR}_{n,0}\right)\ge \tau\big| D_{1,\ell_0,0}=d\right].$$
\small{\begin{align}
q_{k,n,D_{1,\ell_0,0}}\left(d\right)&\triangleq {\rm Pr}\left[\frac{W}{k}\log_{2}\left(1+{\rm SINR}_{n,0}\right)\ge \tau\Big| D_{1,\ell_0,0}=d\right]\notag\\
%&q_{k,n,D_{1,\ell_0,0}}\left(d\right)\notag\\
&\eqla{\rm E}_{I_{1},I_{2}}\left[{\rm Pr}\left[\left|h_{1,\ell_0,0}\right|^{2}\ge \left(2^{\frac{k\tau}{W}}-1\right)D_{1,\ell_0,0}^{\alpha_1}\left(I_{1}+I_{2}\frac{P_2}{P_1}+\frac{N_{0}}{P_1}\right)\Big|D_{1,\ell_0,0}=d\right]\right]\notag\\
&\eqlb{\rm E}_{I_{1},I_{2}}\left[\exp\left(-\left(2^{\frac{k\tau}{W}}-1\right)d^{\alpha_1}\left(I_{1}+I_{2}\frac{P_2}{P_1}+\frac{N_{0}}{P_1}\right)\right)\right]\notag
%\eqla&{\rm E}_{I_{1},\ldots,I_{N}}\left[\exp\left(-\left(2^{\frac{\tau}{W}}-1\right)d^{\alpha}\left(I_{n}+\sum_{k=1,k\neq n}^{N}I_{k}+\frac{N_{0}}{P}\right)\right)\right]\notag\\
%=&{\rm E}_{I_{n}}\left[\exp\left(-\left(2^{\frac{\tau}{W}}-1\right)d^{\alpha}I_{n}\right)\right]\prod_{k=1,k\neq n}^{N}{\rm E}_{I_{k}}\left[\exp\left(-\left(2^{\frac{\tau}{W}}-1\right)d^{\alpha}I_{k}\right)\right]\exp\left(-\left(2^{\frac{\tau}{W}}-1\right)d^{\alpha}\frac{N_{0}}{P}\right)\notag\\
%=&\mathcal{L}_{I_{1,n}}(s,d)|_{s=\left(2^{\frac{k\tau}{W}}-1\right)d^{\alpha_1}} \mathcal{L}_{I_{2}}(s,d)|_{s=\left(2^{\frac{k\tau}{W}}-1\right)d^{\alpha_1}}\exp\left(-\left(2^{\frac{k\tau}{W}}-1\right)d^{\alpha_1}\frac{N_{0}}{P_1}\right)\;.\label{eq:condi_CP_K_m}
\end{align}
\begin{align}
\hspace{11mm}&\eqlc\ \underbrace{{\rm E}_{I_{1}}\left[\exp\left(-\left(2^{\frac{k\tau}{W}}-1\right)d^{\alpha_1}I_{1}\right)\right]}_{\triangleq\mathcal{L}_{I_{1}}(s,d)|_{s=\left(2^{\frac{k\tau}{W}}-1\right)d^{\alpha_1}} }
%\underbrace{{\rm E}_{I_{2}}\left[\exp\left(-\left(2^{\frac{k\tau}{W}}-1\right)d^{\alpha_2}I_{2,-n}\right)\right]}_{\triangleq\mathcal{L}_{I_{2,-n}}(s,d)|_{s=\left(2^{\frac{k\tau}{W}}-1\right)d^{\alpha_2}} }\notag\\
\underbrace{{\rm E}_{I_{2}}\left[\exp\left(-\left(2^{\frac{k\tau}{W}}-1\right)d^{\alpha_1}I_{2}\frac{P_2}{P_1}\right)\right]}_{\triangleq\mathcal{L}_{I_{2}}(s,d)|_{s=\left(2^{\frac{k\tau}{W}}-1\right)d^{\alpha_1}} }\notag\\
&\hspace{8mm}\times\exp\left(-\left(2^{\frac{k\tau}{W}}-1\right)d^{\alpha_1}\frac{N_{0}}{P_1}\right),\label{eq:condi_CP_K_m}
\end{align}}
\normalsize
where $(a)$ is obtained based on (\ref{eqn:SINR_v3}), (b) is obtained by noting that $|h_{1,\ell_0,0}|^{2}\dis \exp(1)$, and (c) is due to the independence of the Rayleigh fading channels and the independence of the PPPs. To calculate $q_{k,n,D_{1,\ell_0,0}}\left(d\right)$ according to \eqref{eq:condi_CP_K_m}, we first calculate $\mathcal{L}_{I_{1}}(s,d)$ and $\mathcal{L}_{I_{2}}(s,d)$, respectively. The expression of $\mathcal{L}_{I_{1}}(s,d)$ is calculated as follows:
\small{\begin{align}\label{eq:LT_K1_n_m}
\mathcal{L}_{I_{1}}(s,d)&={\rm E}\left[\exp\left(-s\sum_{\ell\in\Phi_{1}\backslash \{\ell_{0}\}}D_{1,\ell,0}^{-\alpha_1}\left|h_{1,\ell,0}\right|^{2} \right)\right]
={\rm E}\left[\prod_{\ell\in\Phi_{1}\backslash \{\ell_{0}\}}\exp\left(-s D_{1,\ell,0}^{-\alpha_1}\left|h_{1,\ell,0}\right|^{2} \right)\right]\notag\\
&\eqld\exp\left(-2\pi \lambda_{1}\int_{d}^{\infty}\left(1-\frac{1}{1+sr^{-\alpha_1}}\right)r{\rm d}r\right)\notag\\
&\eqle\exp\left(-\frac{2\pi}{\alpha_1}\lambda_{1}s^{\frac{2}{\alpha_1}}B^{'}\left(\frac{2}{\alpha_1},1-\frac{2}{\alpha_1},\frac{1}{1+sd^{-\alpha_1}}\right)\right),
\end{align}}
\normalsize
where $(d)$ is obtained by utilizing the probability generating functional of PPP \cite[Page 235]{FTNhaenggi09}, and $(e)$ is obtained by first replacing $s^{-\frac{1}{\alpha_1}}r$ with $t$, and then replacing $\frac{1}{1+t^{-\alpha_1}}$ with $w$. Similarly, the expression of $\mathcal{L}_{I_{2}}(s,d)$ is calculated as follows:
\small{\begin{align}\label{eq:LT_K1_m1}
\mathcal{L}_{I_{2}}(s,d)=&{\rm E}\left[\exp\left(-s\sum_{\ell\in\Phi_{2}}D_{2,\ell,0}^{-\alpha_2}\left|h_{2,\ell,0}\right|^{2}\frac{P_2}{P_1}\right)\right]={\rm E}\left[\prod_{\ell\in\Phi_{2}}\exp\left(-s D_{2,\ell,0}^{-\alpha_2}\left|h_{2,\ell,0}\right|^{2}\frac{P_2}{P_1}\right)\right]\notag\\
=&\exp\left(-2\pi \lambda_{2}\int_{0}^{\infty}\left(1-\frac{1}{1+\frac{P_2}{P_1}sr^{-\alpha_2}}\right)r{\rm d}r\right)\notag\\
=&\exp\left(-\frac{2\pi}{\alpha_2}\lambda_{2}\left(\frac{P_2}{P_1}s\right)^{\frac{2}{\alpha_2}}B\left(\frac{2}{\alpha_2},1-\frac{2}{\alpha_2}\right)\right)\;.
\end{align}}
\normalsize
%To calculate $q_{k,n,D_{1,\ell_0,0}}\left(d\right)$ according to \eqref{eq:condi_CP_K_m}, we first calculate $\mathcal{L}_{I_{1,n}}(s,d)$, $\mathcal{L}_{I_{2}}(s,d)$, respectively. Similar to (\ref{eq:LT_K1_n}) and (\ref{eq:LT_K1_m}), we have:
%\begin{align}\label{eq:LT_K1_n_m}
%\mathcal{L}_{I_{1,n}}(s,d)=\exp\left(-\frac{2\pi}{\alpha_1}\lambda_{1}s^{\frac{2}{\alpha_1}}B^{'}\left(\frac{2}{\alpha_1},1-\frac{2}{\alpha_1},\frac{1}{1+sd^{-\alpha_1}}\right)\right),
%\end{align}
%\begin{align}\label{eq:LT_K1_k_m}
%\mathcal{L}_{I_{2}}(s,d)=\exp\left(-\frac{2\pi}{\alpha_2}\lambda_{2}\left(\frac{P_2}{P_1}s\right)^{\frac{2}{\alpha_2}}B\left(\frac{2}{\alpha_2},1-\frac{2}{\alpha_2}\right)\right)\;.
%\end{align}
Substituting (\ref{eq:LT_K1_n_m}) and (\ref{eq:LT_K1_m1}) into (\ref{eq:condi_CP_K_m}), we obtain  $q_{k,n,D_{1,\ell_0,0}}\left(d\right)$ as follows:
\small{\begin{align}\label{eq:condi_CP_K_v2_m}
q_{k,n,D_{1,\ell_0,0}}\left(d\right)
=&\exp\left(-\frac{2\pi}{\alpha_1}\lambda_{1}d^{2}\left(2^{\frac{k\tau}{W}}-1\right)^{\frac{2}{\alpha_1}}B^{'}\left(\frac{2}{\alpha_1},1-\frac{2}{\alpha_1},2^{-\frac{k\tau}{W}}\right)\right)\exp\left(-\left(2^{\frac{k\tau}{W}}-1\right)d^{\alpha_1}\frac{N_{0}}{P_1}\right)\notag\\
&\times \exp\left(-\frac{2\pi}{\alpha_2}\lambda_{2}d^{\frac{2\alpha_1}{\alpha_2}}\left(\frac{P_2}{P_1}\left(2^{\frac{k\tau}{W}}-1\right)\right)^{\frac{2}{\alpha_2}}B\left(\frac{2}{\alpha_2},1-\frac{2}{\alpha_2}\right)\right).
\end{align}}
\normalsize
Now, we calculate $q_{1}\left(\mathcal F_1^c,\mathcal F_2^c\right)$ by first removing the condition of $q_{k,n,D_{1,\ell_0,0}}\left(d\right)$ on $D_{1,\ell_0,0}=d$.  Note that we have the p.d.f. of $D_{1,\ell_0,0}$ as $f_{D_{1,\ell_0,0}}(d)=2\pi \lambda_{1}d\exp\left(-\pi \lambda_{1}d^{2}\right)$. Thus, we have:
\small{\begin{align}
&\int_{0}^{\infty}q_{k,n,D_{1,\ell_0,0}}\left(d\right)f_{D_{1,\ell_0,0}}(d){\rm d}d\nonumber\\
=&2\pi\lambda_{1}\int_{0}^{\infty}d\exp\left(-\frac{2\pi}{\alpha_1}\lambda_{1}d^{2}\left(2^{\frac{k\tau}{W}}-1\right)^{\frac{2}{\alpha_1}}B^{'}\left(\frac{2}{\alpha_1},1-\frac{2}{\alpha_1},2^{-\frac{k\tau}{W}}\right)\right)\exp\left(-\left(2^{\frac{k\tau}{W}}-1\right)d^{\alpha_1}\frac{N_{0}}{P_1}\right)\notag\\
&\times \exp\left(-\frac{2\pi}{\alpha_2}\lambda_{2}d^{\frac{2\alpha_1}{\alpha_2}}\left(\frac{P_2}{P_1}\left(2^{\frac{k\tau}{W}}-1\right)\right)^{\frac{2}{\alpha_2}}B\left(\frac{2}{\alpha_2},1-\frac{2}{\alpha_2}\right)\right) \exp\left(-\pi \lambda_{1}d^{2}\right){\rm d}d.\label{eq:CP_K_n_m}
\end{align}}
\normalsize
%Finally, let $k^c$ and $k^b$ denote file load of serving macro-BS from cache and backhaul respectively, $k=k^c+k^b$, and we have:
Therefore, by~(\ref{eqn:succ-prob-def-1}) and by letting $k=k^c+k^b$ in~(\ref{eq:CP_K_n_m}), we have
\small{\begin{align}\label{eqn:suc_transmit_m}
q_{1}\left(\mathcal F_1^c,\mathcal F_2^c\right)=&\sum_{n\in \mathcal F_1^c}a_{n}\sum_{k^c=1}^{K_1^c} \sum_{k^b=0}^{F_1^b}\Pr [K_{1,n,0}^c=k^c]\Pr[\overline{K}_{1,n,0}^{b}=k^b]\int_{0}^{\infty}q_{k,n,D_{1,\ell_0,0}}\left(d\right)f_{D_{1,\ell_0,0}}(d){\rm d}d\nonumber\\
&+\sum_{n\in \mathcal F_1^b}a_{n}\sum_{k^c=0}^{K_1^c} \sum_{k^b=1}^{F_1^b}\Pr [\overline{K}_{1,n,0}^{c}=k^c]\Pr[K_{1,n,0}^{b}=k^b]\frac{\min\{K_1^b,k^b\}}{k^b}\int_{0}^{\infty}q_{k,n,D_{1,\ell_0,0}}\left(d\right)f_{D_{1,\ell_0,0}}(d){\rm d}d.
\end{align}}
\normalsize

\subsection*{Calculation of $q_{2}\left(\mathcal F_2^c,\mathbf p\right)$}
%When $u_0$ is a macro-user, same as in the traditional connection-based HetNets, there are two types of interferers, namely, i) all the other macro-BSs besides its serving macro-BS, and ii) all the pico-BSs. When $u_0$ is a pico-user, different from  the traditional connection-based HetNets, there are three types of interferers, namely,  i) all the other pico-BSs storing the combinations containing the desired file of $u_0$  besides its serving pico-BS,   ii) all the pico-BSs without  the desired file of $u_0$, and iii) all the macro-BSs. By carefully handling these interferers and  using tools from stochastic geometry, we can obtain tractable expressions for the probabilities of the corresponding channel capacity greater than or equal to $\tau$.
% In addition, using tools from stochastic geometry, we can obtain the closed-form expressions for
%the exact p.m.f.s of $K_{1,n,0}^c,\overline{K}_{1,n,0}^b,\overline{K}_{1,n,0}^c,K_{1,n,0}^b$  and approximated p.m.f. of $K_{2,n,0}^c$.  Therefore, we can obtain $q(\mathcal F_1^c,\mathcal F_2^c, \mathbf p)$ in \eqref{eqn:succ-prob-def}, as summarized   in the following theorem.

%First, as illustrated in Section~\ref{Sec:perf}, when typical user $u_0$ requests file $n \in \mathcal F_2^c$, there are three types of interferers, i.e., the interfering pico-BSs storing the file requested  by $u_{0}$, the interfering pico-BSs storing the other files and all the macro-BSs.   Thus, we rewrite the SINR expression ${\rm SINR}_{n,0}$ in \eqref{eqn:SINR} as follows:

When $u_0$ is a pico-user, different from  the traditional connection-based HetNets, there are three types of interferers, namely,  i) all the other pico-BSs storing the combinations containing the desired file of $u_0$  besides its serving pico-BS,   ii) all the pico-BSs without  the desired file of $u_0$, and iii) all the macro-BSs. Thus, we rewrite the SINR expression in \eqref{eqn:SINR} as follows:
\small{\begin{align}
{\rm SINR}_{n,0}=&\frac{{D_{2,\ell_0,0}^{-\alpha_{2}}}\left|h_{2,\ell_0,0}\right|^{2}}{\sum_{\ell\in\Phi_{2,n}\backslash \{\ell_0\}}D_{2,\ell,0}^{-\alpha_{2}}\left|h_{2,\ell,0}\right|^{2}+\sum_{\ell\in\Phi_{2,-n}}D_{2,\ell,0}^{-\alpha_{2}}\left|h_{2,\ell,0}\right|^{2}+\sum_{\ell\in\Phi_{1}}D_{1,\ell,0}^{-\alpha_{1}}\left|h_{1,\ell,0}\right|^{2}
\frac{P_{1}}{P_{2}}+\frac{N_{0}}{P_{2}}}\nonumber\\
=&\frac{{D_{2,\ell_0,0}^{-\alpha_{2}}}\left|h_{2,\ell_0,0}\right|^{2}}{I_{2,n}+I_{2,-n}+I_{1}\frac{P_1}{P_2}+\frac{N_0}{P_2}}, \ n\in\mathcal N,
\label{eqn:SINR_v2}
\end{align}}
\normalsize
where $\Phi_{2,n}$ is the point process generated by pico-BSs containing file combination $i\in\mathcal{I}_{n}$, $\Phi_{2,-n}$ is the point process generated by pico-BSs containing file combination $i\not\in\mathcal{I}_{n}$, $I_{2,n}\triangleq\sum_{\ell\in\Phi_{2,n}\backslash \{\ell_{0}\}}D_{2,\ell,0}^{-\alpha_2}\left|h_{2,\ell,0}\right|^{2}$, $I_{2,-n}\triangleq \sum_{\ell\in\Phi_{2,-n}}D_{2,\ell,0}^{-\alpha_2}\left|h_{2,\ell,0}\right|^{2}$ and $I_{1}\triangleq \sum_{\ell\in\Phi_{1}}D_{1,\ell,0}^{-\alpha_1}\left|h_{1,\ell,0}\right|^{2}$. Due to the random caching policy and independent thinning \cite[Page 230]{FTNhaenggi09}, we obtain that $\Phi_{2,n}$ is a homogeneous PPP with density $\lambda_{2}T_n$ and $\Phi_{2,-n}$ is a homogeneous PPP with density $\lambda_{2}\left(1-T_n\right)$.

Next, we calculate the conditional successful transmission probability of file $n\in \mathcal F_2^c$ requested by $u_{0}$ conditioned on $D_{2,\ell_0,0}=d$ when the file load is $k$, denoted as   $$q_{k,n,D_{2,\ell_0,0}}\left({\bf p},d\right)\triangleq {\rm Pr}\left[\frac{W}{k}\log_{2}\left(1+{\rm SINR}_{n,0}\right)\ge \tau \Big|D_{2,\ell_0,0}=d\right].$$ Similar to~(\ref{eq:condi_CP_K_m}) and based on (\ref{eqn:SINR_v2}), we have:
\small{\begin{align}
&q_{k,n,D_{2,\ell_0,0}}\left({\bf p},d\right)\notag\\
&={\rm E}_{I_{2,n},I_{2,-n},I_{1}}\left[{\rm Pr}\left[\left|h_{2,\ell_0,0}\right|^{2}\ge \left(2^{\frac{k\tau}{W}}-1\right)D_{2,\ell_0,0}^{\alpha_2}\left(I_{2,n}+I_{2,-n}+I_{1}\frac{P_1}{P_2}+\frac{N_{0}}{P_2}\right)\Big|D_{2,\ell_0,0}=d\right]\right]\notag
%&\eqlb{\rm E}_{I_{2,n},I_{2,-n},I_{1}}\left[\exp\left(-\left(2^{\frac{k\tau}{W}}-1\right)d^{\alpha_2}\left(I_{2,n}+I_{2,-n}+I_{1}\frac{P_1}{P_2}+\frac{N_{0}}{P_2}\right)\right)\right]\notag
%\eqla&{\rm E}_{I_{1},\ldots,I_{N}}\left[\exp\left(-\left(2^{\frac{\tau}{W}}-1\right)d^{\alpha}\left(I_{n}+\sum_{k=1,k\neq n}^{N}I_{k}+\frac{N_{0}}{P}\right)\right)\right]\notag\\
%=&{\rm E}_{I_{n}}\left[\exp\left(-\left(2^{\frac{\tau}{W}}-1\right)d^{\alpha}I_{n}\right)\right]\prod_{k=1,k\neq n}^{N}{\rm E}_{I_{k}}\left[\exp\left(-\left(2^{\frac{\tau}{W}}-1\right)d^{\alpha}I_{k}\right)\right]\exp\left(-\left(2^{\frac{\tau}{W}}-1\right)d^{\alpha}\frac{N_{0}}{P}\right)\notag\\
%=&\mathcal{L}_{I_{2,n}}(s,d)|_{s=\left(2^{\frac{k\tau}{W}}-1\right)d^{\alpha_2}} \mathcal{L}_{I_{2,-n}}(s,d)|_{s=\left(2^{\frac{k\tau}{W}}-1\right)d^{\alpha_2}}\mathcal{L}_{I_{1}}(s,d)|_{s=\left(2^{\frac{k\tau}{W}}-1\right)d^{\alpha_2}}\exp\left(-\left(2^{\frac{k\tau}{W}}-1\right)d^{\alpha_2}\frac{N_{0}}{P_2}\right)\;.
\end{align}
\begin{align}
&\hspace{-30mm}= \underbrace{{\rm E}_{I_{2,n}}\left[\exp\left(-\left(2^{\frac{k\tau}{W}}-1\right)d^{\alpha_2}I_{2,n}\right)\right]}_{\triangleq\mathcal{L}_{I_{2,n}}(s,d)|_{s=\left(2^{\frac{k\tau}{W}}-1\right)d^{\alpha_2}} }\underbrace{{\rm E}_{I_{2,-n}}\left[\exp\left(-\left(2^{\frac{k\tau}{W}}-1\right)d^{\alpha_2}I_{2,-n}\right)\right]}_{\triangleq\mathcal{L}_{I_{2,-n}}(s,d)|_{s=\left(2^{\frac{k\tau}{W}}-1\right)d^{\alpha_2}} }\notag\\
&\hspace{-25mm}\times\underbrace{{\rm E}_{I_{1}}\left[\exp\left(-\left(2^{\frac{k\tau}{W}}-1\right)d^{\alpha_2}I_{1}\frac{P_1}{P_2}\right)\right]}_{\triangleq\mathcal{L}_{I_{1}}(s,d)|_{s=\left(2^{\frac{k\tau}{W}}-1\right)d^{\alpha_2}} }
 \exp\left(-\left(2^{\frac{k\tau}{W}}-1\right)d^{\alpha_2}\frac{N_{0}}{P_2}\right).\label{eq:condi_CP_K}
\end{align}}
\normalsize
To calculate $q_{k,n,D_{2,\ell_0,0}}\left({\bf p},d\right)$ according to \eqref{eq:condi_CP_K}, we first calculate $\mathcal{L}_{I_{2,n}}(s,d)$, $\mathcal{L}_{I_{2,-n}}(s,d)$ and $\mathcal{L}_{I_{1}}(s,d)$ , respectively. Similar to~(\ref{eq:LT_K1_n_m}) and~(\ref{eq:LT_K1_m1}), we have:
\small{\begin{align}\label{eq:LT_K1_n}
\mathcal{L}_{I_{2,n}}(s,d)%&{\rm E}\left[\exp\left(-s\sum_{\ell\in\Phi_{2,n}\backslash \ell_{0}}D_{2,\ell,0}^{-\alpha_2}\left|h_{2,\ell,0}\right|^{2} \right)\right]
%={\rm E}\left[\prod_{\ell\in\Phi_{2,n}\backslash \ell_{0}}\exp\left(-s D_{2,\ell,0}^{-\alpha_2}\left|h_{2,\ell,0}\right|^{2} \right)\right]\notag\\
%\eqld&\exp\left(-2\pi T_{n}\lambda_{2}\int_{d}^{\infty}\left(1-\frac{1}{1+sr^{-\alpha_2}}\right)r{\rm d}r\right)\notag\\
=\exp\left(-\frac{2\pi}{\alpha_2}T_{n}\lambda_{2}s^{\frac{2}{\alpha_2}}B^{'}\left(\frac{2}{\alpha_2},1-\frac{2}{\alpha_2},\frac{1}{1+sd^{-\alpha_2}}\right)\right),
\end{align}
\begin{align}\label{eq:LT_K1_p1}
\mathcal{L}_{I_{2,-n}}(s,d)%&{\rm E}\left[\exp\left(-s\sum_{\ell\in\Phi_{2,-n}}D_{2,\ell,0}^{-\alpha_2}\left|h_{2,\ell,0}\right|^{2}\right)\right]={\rm E}\left[\prod_{\ell\in\Phi_{2,-n}}\exp\left(-s D_{2,\ell,0}^{-\alpha_2}\left|h_{2,\ell,0}\right|^{2}\right)\right]\notag\\
%=&\exp\left(-2\pi (1-T_n)\lambda_{2}\int_{0}^{\infty}\left(1-\frac{1}{1+sr^{-\alpha_2}}\right)r{\rm d}r\right)\notag\\
=\exp\left(-\frac{2\pi}{\alpha_2}(1-T_n)\lambda_{2}s^{\frac{2}{\alpha_2}}B\left(\frac{2}{\alpha_2},1-\frac{2}{\alpha_2}\right)\right)\;,
\end{align}
\begin{align}\label{eq:LT_K1_p2}
\mathcal{L}_{I_{1}}(s,d)%&{\rm E}\left[\exp\left(-s\sum_{\ell\in\Phi_{1}}D_{1,\ell,0}^{-\alpha_1}\left|h_{1,\ell,0}\right|^{2}\frac{P_1}{P_2}\right)\right]={\rm E}\left[\prod_{\ell\in\Phi_{1}}\exp\left(-s D_{1,\ell,0}^{-\alpha_1}\left|h_{1,\ell,0}\right|^{2}\frac{P_1}{P_2}\right)\right]\notag\\
%=&\exp\left(-2\pi \lambda_{1}\int_{0}^{\infty}\left(1-\frac{1}{1+\frac{P_1}{P_2}sr^{-\alpha_1}}\right)r{\rm d}r\right)\notag\\
=\exp\left(-\frac{2\pi}{\alpha_1}\lambda_{1}\left(\frac{P_1}{P_2}s\right)^{\frac{2}{\alpha_1}}B\left(\frac{2}{\alpha_1},1-\frac{2}{\alpha_1}\right)\right)\;.
\end{align}}
\normalsize
Substituting (\ref{eq:LT_K1_n}), (\ref{eq:LT_K1_p1}) and (\ref{eq:LT_K1_p2}) into (\ref{eq:condi_CP_K}), we obtain  $q_{k,n,D_{2,\ell_0,0}}\left({\bf p},d\right)$ as follows:
\small{\begin{align}\label{eq:condi_CP_K_v2}
q_{k,n,D_{2,\ell_0,0}}\left({\bf p},d\right)
%=&{\rm E}_{I_{i},I_{{\rm other}}}\left[{\rm Pr}\left(\left|h_{0,0}\right|^{2}\ge \left(2^{\frac{K\tau}{W}}-1\right)d^{\alpha}\left(I_{i}+I_{{\rm other}}+\frac{N_{0}}{P}\right)\Big|u_{0}\in\mathcal{U}_{n},D_{0,0}=d\right)\right]\notag\\
%\eqla&{\rm E}_{I_{1},\ldots,I_{N}}\left[\exp\left(-\left(2^{\frac{\tau}{W}}-1\right)d^{\alpha}\left(I_{n}+\sum_{k=1,k\neq n}^{N}I_{k}+\frac{N_{0}}{P}\right)\right)\right]\notag\\
%=&{\rm E}_{I_{n}}\left[\exp\left(-\left(2^{\frac{\tau}{W}}-1\right)d^{\alpha}I_{n}\right)\right]\prod_{k=1,k\neq n}^{N}{\rm E}_{I_{k}}\left[\exp\left(-\left(2^{\frac{\tau}{W}}-1\right)d^{\alpha}I_{k}\right)\right]\exp\left(-\left(2^{\frac{\tau}{W}}-1\right)d^{\alpha}\frac{N_{0}}{P}\right)\notag\\
=&\exp\left(-\frac{2\pi}{\alpha_2}T_n\lambda_{2}d^{2}\left(2^{\frac{k\tau}{W}}-1\right)^{\frac{2}{\alpha_2}}B^{'}\left(\frac{2}{\alpha_2},1-\frac{2}{\alpha_2},2^{-\frac{k\tau}{W}}\right)\right) \exp\left(-\left(2^{\frac{k\tau}{W}}-1\right)d^{\alpha_2}\frac{N_{0}}{P_2}\right)\notag\\ &\times\exp\left(-\frac{2\pi}{\alpha_2}\left(1-T_n\right)\lambda_{2}d^{2}\left(2^{\frac{k\tau}{W}}-1\right)^{\frac{2}{\alpha_2}}B\left(\frac{2}{\alpha_2},1-\frac{2}{\alpha_2}\right)\right)\notag\\
&\times \exp\left(-\frac{2\pi}{\alpha_1}\lambda_{1}d^{\frac{2\alpha_2}{\alpha_1}}\left(\frac{P_1}{P_2}\left(2^{\frac{k\tau}{W}}-1\right)\right)^{\frac{2}{\alpha_1}}B\left(\frac{2}{\alpha_1},1-\frac{2}{\alpha_1}\right)\right).
\end{align}}
\normalsize
Now, we calculate $q_{2}\left(\mathcal F_2^c,{\bf p}\right)$ by first removing the condition of $q_{k,n,D_{2,\ell_0,0}}\left({\bf p},d\right)$ on $D_{2,\ell_0,0}=d$.  Note that we have the p.d.f. of $D_{2,\ell_0,0}$ as $f_{D_{2,\ell_0,0}}(d)=2\pi T_n\lambda_{2}d\exp\left(-\pi T_n\lambda_{2}d^{2}\right)$, as pico-BSs storing file $n$ form a homogeneous PPP with density $T_n\lambda_{2}$. Thus, we have:
\small{\begin{align}
&\int_{0}^{\infty}q_{k,n,D_{2,\ell_0,0}}\left({\bf p},d\right)f_{D_{2,\ell_0,0}}(d){\rm d}d\nonumber\\
=&2\pi T_n\lambda_{2}\int_{0}^{\infty}d\exp\left(-\frac{2\pi}{\alpha_2}T_n\lambda_{2}d^{2}\left(2^{\frac{k\tau}{W}}-1\right)^{\frac{2}{\alpha_2}}B^{'}\left(\frac{2}{\alpha_2},1-\frac{2}{\alpha_2},2^{-\frac{k\tau}{W}}\right)\right) \exp\left(-\left(2^{\frac{k\tau}{W}}-1\right)d^{\alpha_2}\frac{N_{0}}{P_2}\right)\nonumber\\ &\times\exp\left(-\frac{2\pi}{\alpha_2}\left(1-T_n\right)\lambda_{2}d^{2}\left(2^{\frac{k\tau}{W}}-1\right)^{\frac{2}{\alpha_2}}B\left(\frac{2}{\alpha_2},1-\frac{2}{\alpha_2}\right)\right)\exp\left(-\pi T_n\lambda_{2}d^{2}\right)\notag\\
&\times \exp\left(-\frac{2\pi}{\alpha_1}\lambda_{1}d^{\frac{2\alpha_2}{\alpha_1}}\left(\frac{P_1}{P_2}\left(2^{\frac{k\tau}{W}}-1\right)\right)^{\frac{2}{\alpha_1}}B\left(\frac{2}{\alpha_1},1-\frac{2}{\alpha_1}\right)\right)
 {\rm d}d.\label{eq:CP_K_n}
\end{align}}
\normalsize
Therefore, by~(\ref{eqn:succ-prob-def-2}) and by letting $k=k^c$ in~(\ref{eq:CP_K_n}), we have
\small{\begin{align}\label{eqn:suc_transmit_p}
q_{2}\left(\mathcal F_2^c,{\bf p}\right)=\sum_{n\in \mathcal F_2^c}a_{n}\sum_{k=1}^{K_2^c} \Pr [K_{2,n,0}^c=k^c]\int_{0}^{\infty}q_{k,n,D_{2,\ell_0,0}}\left({\bf p},d\right)f_{D_{2,\ell_0,0}}(d){\rm d}d.
\end{align}}
\normalsize

\section*{Appendix D: Proof of Lemma~\ref{Lem:asym-perf} and Lemma~\ref{Lem:asym-perf-v2}}

\subsection*{Proof of Lemma~\ref{Lem:asym-perf}}
When $\frac{P}{N}\to \infty$, $\exp\left(-\left(2^{\frac{k\tau}{W}}-1\right)d^{\alpha}\frac{N_{0}}{P_1}\right)\to1$ and $\exp\left(-\left(2^{\frac{k\tau}{W}}-1\right)d^{\alpha}\frac{N_{0}}{P_2}\right)\to1$. When $\lambda_u\to \infty$, discrete random variables $K_{1,n,0}^c,\overline{K}_{1,n,0}^c\to K_1^c$, $\overline{K}_{1,n,0}^b,K_{1,n,0}^b\to F_1^b$ and $K_{2,n,0}^c\to K_2^c$ in distribution. Thus, when $P_1=\beta P$, $P_2=P$,
 $\frac{P}{N_0}\to \infty$, and $\lambda_u\to \infty$, we can show $f_{1,K_1^c+\min\{K_1^b,F_1^b\}}\to f_{1,K_1^c+\min\{K_1^b,F_1^b\},\infty}$ and $f_{2,K_2^c}(x)\to f_{2,K_2^c,\infty}(x)$.
%\small{\begin{align}
%f_{1,K_1^c+\min\{K_1^b,F_1^b\}}=& 2\pi\lambda_{1}\int_{0}^{\infty} d\exp\left(-\pi\lambda_{1}d^{2}\right)\exp\left(-\frac{2\pi\lambda_{2}}{\alpha}d^2\left(\frac{1}{\beta}\left(2^{\frac{(K_1^c+\min\{K_1^b,F_1^b\})\tau}{W}}-1\right)\right)^{\frac{2}{\alpha}}B\left(\frac{2}{\alpha},1-\frac{2}{\alpha}\right)\right)\nonumber\\
%&  \times\exp\left(-\frac{2\pi\lambda_{1}}{\alpha}d^2\left(2^{\frac{(K_1^c+\min\{K_1^b,F_1^b\})\tau}{W}}-1\right)^{\frac{2}{\alpha}}B^{'}\left(\frac{2}{\alpha},1-\frac{2}{\alpha},2^{-\frac{(K_1^c+\min\{K_1^b,F_1^b\})\tau}{W}}\right)\right){\rm d}d\nonumber\\
%=&f_{1,K_1^c+\min\{K_1^b,F_1^b\},\infty},\label{eqn:f-1-k-appendix}\\
%f_{2,K_2^c}(x)= &2\pi\lambda_{2}x\int_{0}^{\infty} d\exp\left(-\pi\lambda_{2}xd^{2}\right)\exp\left(-\frac{2\pi\lambda_{1}}{\alpha}d^2\left(\beta\left(2^{\frac{K_2^c\tau}{W}}-1\right)\right)^{\frac{2}{\alpha}}B\left(\frac{2}{\alpha},1-\frac{2}{\alpha}\right)\right)\nonumber\\
%& \times
%\exp\left(-\frac{2\pi\lambda_{2}}{\alpha}d^2\left(2^{\frac{K_2^c\tau}{W}}-1\right)^{\frac{2}{\alpha}}\left(xB^{'}\left(\frac{2}{\alpha},1-\frac{2}{\alpha},2^{-\frac{K_2^c\tau}{W}}\right)+\left(1-x\right)B\left(\frac{2}{\alpha},1-\frac{2}{\alpha}\right)\right)\right)
%{\rm d}d\nonumber\\
%=&f_{2,K_2^c,\infty}(x).\label{eqn:f-2-k-appendix}
%\end{align}}
%\normalsize{
Thus, we can prove Lemma~\ref{Lem:asym-perf}.

\subsection*{Proof of Lemma~\ref{Lem:asym-perf-v2}}
When $P_1=\beta P$, $P_2=P$, $\frac{P}{N_0}\to \infty$, $\lambda_u\to \infty$, and $\alpha_1=\alpha_2=\alpha$ we have:
\small{\begin{align}
f_{1,K_1^c+\min\{K_1^b,F_1^b\},\infty}=& 2\pi\lambda_{1}\int_{0}^{\infty} d\exp\left(-\pi\lambda_1\omega_{K_1^c+\min\{K_1^b,F_1^b\}}d^2\right){\rm d}d=\frac{1}{\omega_{K_1^c+\min\{K_1^b,F_1^b\}}},\label{eqn:f-1-k-appendix_cf}\\
f_{2,K_2^c,\infty}(x)=&2\pi\lambda_{2}\int_{0}^{\infty} d\exp\left(-\pi\lambda_2d^2\left(\theta_{2,K_2^c}+x\theta_{1,K_2^c}\right)\right){\rm d}d=\frac{x}{\theta_{2,K_2^c}+\theta_{1,K_2^c}x},\label{eqn:f-2-k-appendix_cf}
\end{align}}
\normalsize
where $\omega_k$, $\theta_{1,k}$ and $\theta_{2,k}$ are given by \eqref{eqn:c_1_k}, \eqref{eqn:c_2_1_k} and \eqref{eqn:c_2_2_k}. Noting that $\int_{0}^{\infty}d\exp\left(-cd^{2}\right){\rm d}d=\frac{1}{2c}$ ($c$ is a constant), we can solve integrals in \eqref{eqn:f-1-k-infty} and \eqref{eqn:f-2-k-infty}. Thus, by Lemma \ref{Lem:asym-perf}, we can prove Lemma~\ref{Lem:asym-perf-v2}.

\section*{Appendix E: Proof of Lemma~\ref{Lem:mono-general-asym}}

To prove Lemma~\ref{Lem:mono-general-asym}, we first have the following lemma.
\begin{Lem} [Monotonicity of $f_{2,k,\infty}(x)$]
$f_{2,k,\infty}(x)$ is an increasing function of $x$.
\label{Lem:monotonicity-f-2-k}
\end{Lem}
\begin{proof}
By replacing $\exp\left(-\pi x\lambda_2d^2\right)$ with $y$ in~(\ref{eqn:f-2-k-infty}), we have:
\small{\begin{align}\label{eqn:prop-increase}
f_{2,k,\infty}(x)
%\triangleq &2\pi x\lambda_{2}\int_{0}^{\infty}d\exp\left(-\pi T_n\lambda_{2}d^{2}\right)\exp\left(-\frac{2\pi}{\alpha_2}T_n\lambda_{2}d^{2}\left(2^{\frac{k\tau}{W}}-1\right)^{\frac{2}{\alpha_2}}B^{'}\left(\frac{2}{\alpha_2},1-\frac{2}{\alpha_2},2^{-\frac{k\tau}{W}}\right)\right) \notag\\ &\times\exp\left(-\frac{2\pi}{\alpha_2}\left(1-x\right)\lambda_{2}d^{2}\left(2^{\frac{k\tau}{W}}-1\right)^{\frac{2}{\alpha_2}}B\left(\frac{2}{\alpha_2},1-\frac{2}{\alpha_2}\right)\right)\exp\left(-\left(2^{\frac{k\tau}{W}}-1\right)d^{\alpha_2}\frac{N_{0}}{P_2}\right) \notag\\
%&\times \exp\left(-\frac{2\pi}{\alpha_1}\lambda_{1}d^{\frac{2\alpha_2}{\alpha_1}}\left(\frac{P_1}{P_2}\left(2^{\frac{k\tau}{W}}-1\right)\right)^{\frac{2}{\alpha_1}}B\left(\frac{2}{\alpha_1},1-\frac{2}{\alpha_1}\right)\right)
%{\rm d}d.\notag\\
%%&\hspace{-13mm}\overset{y\triangleq \exp\left(-\pi T_n\lambda_2d^2\right)}{=}
%&\hspace{-4mm}\eqla
=&\int_{0}^{1}y^{-\frac{2}{\alpha_2}\left(2^{\frac{k\tau}{W}}-1\right)^{\frac{2}{\alpha_2}}\left(B\left(\frac{2}{\alpha_2},1-\frac{2}{\alpha_2}\right)-B^{'}\left(\frac{2}{\alpha_2},1-\frac{2}{\alpha_2},2^{-\frac{k\tau}{W}}\right)\right) }y^{\frac{2}{\alpha_2x}\left(2^{\frac{k\tau}{W}}-1\right)^{\frac{2}{\alpha_2}}B\left(\frac{2}{\alpha_2},1-\frac{2}{\alpha_2}\right) }\notag\\
&\times y^{\left(\frac{-\ln y}{\pi\lambda_2}\right)^{\frac{\alpha_2}{\alpha_1}-1}\frac{2}{\alpha_1}\left(\frac{1}{x}\right)^\frac{\alpha_2}{\alpha_1}\frac{\lambda_1}{\lambda_2}\left(\frac{P_1}{P_2}\right)^{\frac{2}{\alpha_1}}\left(2^{\frac{k\tau}{W}}-1\right)^{\frac{2}{\alpha_1}}B\left(\frac{2}{\alpha_1},1-\frac{2}{\alpha_1}\right)}
{\rm d}y.
\end{align}}
\normalsize
%\begin{align}
%y\triangleq \exp\left(-\pi T_n\lambda_2d^2\right)
%\end{align}
%\begin{align}
%f_{2,K_2^c}(T_n)\triangleq &\int_{0}^{1}y^{-\frac{2}{\alpha_2}\left(2^{\frac{k\tau}{W}}-1\right)^{\frac{2}{\alpha_2}}\left(B\left(\frac{2}{\alpha_2},1-\frac{2}{\alpha_2}\right)-B^{'}\left(\frac{2}{\alpha_2},1-\frac{2}{\alpha_2},2^{-\frac{k\tau}{W}}\right)\right) }\notag\\ &\times y^{\frac{2}{\alpha_2T_n}\left(2^{\frac{k\tau}{W}}-1\right)^{\frac{2}{\alpha_2}}B\left(\frac{2}{\alpha_2},1-\frac{2}{\alpha_2}\right) }\notag\\
%&\times y^{\left(\frac{-lny}{\pi\lambda_2}\right)^{\frac{\alpha_2}{\alpha_1}-1}\frac{2}{\alpha_1}\left(\frac{1}{T_n}\right)^\frac{\alpha_2}{\alpha_1}\frac{\lambda_1}{\lambda_2}\left(\frac{P_1}{P_2}\right)^{\frac{2}{\alpha_1}}\left(2^{\frac{k\tau}{W}}-1\right)^{\frac{2}{\alpha_1}}B\left(\frac{2}{\alpha_1},1-\frac{2}{\alpha_1}\right)}\notag\\
%&\times y^{\left(\frac{-lny}{\pi\lambda_2}\right)^{\frac{\alpha_2}{2}-1}\left(\frac{1}{T_n}\right)^\frac{\alpha_2}{2}\frac{1}{\pi\lambda_2}\left(2^{\frac{k\tau}{W}}-1\right)\frac{N_0}{P_2}}{\rm d}y.
%\end{align}
When $y \in (0,1)$ and $a\in \left(0,\infty\right)$, $y^a$ is a decreasing function of $a$. Because $B\left(\frac{2}{\alpha_2}, 1-\frac{2}{\alpha_2}\right)$, $B\left(\frac{2}{\alpha_1},1-\frac{2}{\alpha_1}\right)$ and $2^{\frac{k\tau}{W}}-1>0$, and $\frac{1}{x}$ and $\left(\frac{1}{x}\right)^{\frac{\alpha_2}{\alpha_1}}$ are decreasing functions of $x$. The integrand is an increasing function of $x$ for all $y\in(0,1)$. Therefore, we can show that $f_{2,k,\infty}(x)$ is an increasing function of $x$.
\end{proof}

%\subsection*{Proof of Lemma~\ref{Lem:mono-general-asym}}
Now, we prove Lemma~\ref{Lem:mono-general-asym}.
Let $(\mathcal F_1^{c*}, \mathcal F_2^{c*}, \mathbf T^*)$ denote an optimal solution to Probelm~\ref{prob:opt-asymp-eq}. Consider $n_1, n_2 \in \mathcal F_2^{c*}$ satisfying $a_{n_1}> a_{n_2}$. Suppose $T_{n_1}^*<T_{n_2}^*$. Based on Lemma~\ref{Lem:monotonicity-f-2-k}, we have  $f_{2,K_2^{c},\infty}(T_{n_1}^*)<f_{2,K_2^{c},\infty}(T_{n_2}^*)$. Now, we construct a feasible solution $(\mathcal F_1^{c'}, \mathcal F_2^{c'}, \mathbf T^{'})$ to Problem~\ref{prob:opt-asymp-eq} by choosing $\mathcal F_1^{c'}=\mathcal F_1^{c*}$, $\mathcal F_2^{c'}=\mathcal F_2^{c*}$, $T_{n_1}'=T_{n_2}^{*}, T_{n_2}'=T_{n_1}^{*}$, and $T_n^{'}=T_n^*$ for all $n \in \mathcal F_2^c \setminus \{n_1,n_2\}$. Thus, by Lemma~\ref{Lem:asym-perf} and the optimality of $(\mathcal F_1^{c*}, \mathcal F_2^{c*}, \mathbf T^*)$, we have:
\small{\begin{align}\label{eqn:contradiction-thm2-a}
q_{\infty}\left(\mathcal F_1^{c'},\mathcal F_2^{c'},\mathbf T^{'}\right)-q_{\infty}\left(\mathcal F_1^{c*},\mathcal F_2^{c*},\mathbf T^{*}\right)=\left(a_{n_1}-a_{n_2}\right)\left(f_{2,K_2^c,\infty}(T_{n_2}^*)-f_{2,K_2^c,\infty}(T_{n_1}^*)\right)\leq 0.
\end{align}}
\normalsize
Since $a_{n_1}> a_{n_2}$, by~(\ref{eqn:contradiction-thm2-a}), we have $f_{2,K_2^c,\infty}(T_{n_2}^*)-f_{2,K_2^c,\infty}(T_{n_1}^*)\leq0$, which contradicts the assumption. Therefore, by contradiction, we can prove Lemma~\ref{Lem:mono-general-asym}.

\section*{Appendix F: Proof of Lemma~\ref{Lem:solu-opt-infty}}
For given $\mathcal F_2^c$, when $\alpha_1=\alpha_2 = \alpha$,
 $\frac{P}{N_0}\to \infty$,  and $\lambda_u  \to \infty $, the Lagrangian of the optimization in \eqref{eqn:prob-p-asymp-F_2^c} is given by
\small{$$L(\mathbf T, \boldsymbol \lambda, \boldsymbol \eta, \nu)=\sum_{n\in \mathcal F_2^c}\frac{a_{n}T_n}{\theta_{2,K_2^c}+\theta_{1,K_2^c}T_n}+\sum_{n\in \mathcal F_2^c}\lambda_nT_n+\sum_{n\in \mathcal F_2^c}\eta_n(1-T_n)+\nu\left(K_2^c-\sum_{n\in \mathcal F_2^c} T_n\right),$$}
\normalsize
where $\lambda_n$ and  $\eta_n\geq 0$  are the Lagrange multipliers associated with \eqref{eqn:cache-constr-indiv-t},  $\nu$ is the Lagrange multiplier associated with \eqref{eqn:cache-constr-sum-t}, $\boldsymbol\lambda \triangleq (\lambda_n)_{n\in \mathcal F_2^c}$, and $\boldsymbol\eta\triangleq (\eta_n)_{n\in \mathcal F_2^c}$. Thus, we have
\small{$$\frac{\partial L}{\partial  T_n}(\mathbf T,\boldsymbol \lambda, \boldsymbol \eta, \nu)=\frac{a_n\theta_{2,K_2^c}}{(\theta_{2,K_2^c}+\theta_{1,K_2^c}T_n)^2}+\lambda_n-\eta_n-\nu.$$}
\normalsize
Since strong duality holds, primal optimal $\mathbf T^*$ and dual optimal  $\boldsymbol\lambda^*$, $\boldsymbol\eta^*$, $\nu^*$ satisfy KKT conditions, i.e.,  (i) primal constraints: \eqref{eqn:cache-constr-indiv-t}, \eqref{eqn:cache-constr-sum-t},  (ii) dual constraints $\lambda_n\geq 0$ and $\eta_n\geq 0$ for all $n\in \mathcal F_2^c$, (iii) complementary slackness $\lambda_n T_n=0$ and $\eta_n(1-T_n)=0$ for all $n\in \mathcal F_2^c$, and (iv) $\frac{a_n\theta_{2,K_2^c}}{(\theta_{2,K_2^c}+\theta_{1,K_2^c}T_n)^2}+\lambda_n-\eta_n-\nu=0$ for all $n\in \mathcal F_2^c$.
By (ii), (iii), and (iv),  when $T_n=0$, we have $\lambda_n\geq 0$, $\eta_n=0$, and $\nu\geq \frac{a_n}{\theta_{2,K_2^c}}$; when $0<T_n<1$, we have  $\lambda_n= 0$, $\eta_n=0$, \small{$T_n=\frac{1}{\theta_{1,K_2^c}}\sqrt{\frac{a_n\theta_{2,K_2^c}}{\nu}}-\frac{\theta_{2,K_2^c}}{\theta_{1,K_2^c}}$}\normalsize, and $\frac{a_n\theta_{2,K_2^c}}{(\theta_{2,K_2^c}+\theta_{1,K_2^c})^2}<\nu< \frac{a_n}{\theta_{2,K_2^c}}$; when $T_n=1$, we have $\lambda_n= 0$, $\eta_n\geq 0$, and $\nu\leq\frac{a_n\theta_{2,K_2^c}}{(\theta_{2,K_2^c}+\theta_{1,K_2^c})^2}$. Therefore, we have \small{$T_n^*=\min\left\{\left[\frac{1}{\theta_{1,K_2^c}}\sqrt{\frac{a_n\theta_{2,K_2^c}}{\nu^*}}-\frac{\theta_{2,K_2^c}}{\theta_{1,K_2^c}}\right]^+,1\right\}$}\normalsize. Combining \eqref{eqn:cache-constr-sum-t}, we can prove Lemma~\ref{Lem:solu-opt-infty}.

\section*{Appendix G: Proof of Theorem~\ref{Thm:opt-prop}}

\subsection*{Proof of Property $(i)$ of Theorem~\ref{Thm:opt-prop}}
By constraints~(\ref{eqn:cache-constr}) and~(\ref{eqn:backhaul}), we have $K_2^c\leq F_2^{c*}$, $ F_2^{c*}=N-K_1^c-F_1^{b*}$ and $0\leq F_1^{b*}$. To prove property (i) of Theorem~\ref{Thm:opt-prop}, it remains to prove $F_1^{b*} \leq K_1^b$.
%thus  $\max\{K_2^c, N-K_1^c-K_1^b\}\leq F_2^{c*}\leq  N-K_1^c $. We prove property (b)(i) of Theorem~\ref{Thm:opt-prop}. we prove $F_1^{b*} \leq K_1^b$ for any optimal solution $(\mathcal F_1^{c*}, \mathcal F_2^{c*}, \mathbf T^{*})$ to Probelm~\ref{prob:opt-asymp-eq}. First,
Suppose there exists an optimal solution $(\mathcal F_1^{c*}, \mathcal F_2^{c*}, \mathbf T^*)$ to Probelm~\ref{prob:opt-asymp-eq} satisfying $F_1^{b*} > K_1^b$, Then we have:
\small{\begin{align}\label{eq:optimal-value-appendix-E2}
q_{\infty}^{*}=f_{1,K_1^c+K_1^b,\infty}\left(\sum_{n \in \mathcal F_1^{c*}}a_n+K_1^b\sum_{n\in \mathcal F_1^{b*} }\frac{a_n}{F_1^{b*}}\right)+\sum_{n\in \mathcal F_2^{c*}}a_nf_{2,K_2^c,\infty}(T_n^*).
\end{align}}\normalsize
Now, we construct a feasible solution $(\mathcal F_1^{c'}, \mathcal F_2^{c'}, \mathbf T^{'})$ to Problem~\ref{prob:opt-asymp-eq}, where $\mathcal F_1^{b'}$ consists of the most $K_1^b$ popular files of $\mathcal F_1^{b*}$, $\mathcal F_1^{c'}=\mathcal F_1^{c*}$, $\mathcal F_2^{c'}=\mathcal F_2^{c*} \cup (\mathcal F_1^{b*} \setminus \mathcal F_1^{b'})$, $T_n'=T_n^*$ for all $n \in \mathcal F_2^{c*}$ and $T_n'=0$ for all $ n \in \mathcal F_1^{b*} \setminus \mathcal F_1^{b'}$. By Lemma~\ref{Lem:asym-perf}, we have:
\small{\begin{align}\label{eq:contradiction-apendix-E2}
q_{\infty}(\mathcal F_1^{c'}, \mathcal F_2^{c'}, \mathbf T^{'})-q_{\infty}^{*}=f_{1,K_1^c+K_1^b,\infty}\left(\frac{1}{K_1^b}\sum_{n\in \mathcal F_1^{b'}}a_n-\frac{1}{F_1^b}\sum_{n\in \mathcal F_1^{b*}}a_n\right)K_1^b>0.
\end{align}}\normalsize
Thus, $(\mathcal F_1^{c*}, \mathcal F_2^{c*}, \mathbf T^{*})$ is not an optimal solution, which contradicts the assumption. Therefore, by contradiction, we can prove  $F_1^{b*}\leq K_1^b$ for any optimal solution $(\mathcal F_1^{c*}, \mathcal F_2^{c*}, \mathbf T^{*})$ to Probelm~\ref{prob:opt-asymp-eq}.
Since $F_2^{c*}\geq K_2^c$, $ F_2^{c*}=N-K_1^c-F_1^{b*}$ and $0\leq F_1^{b*}\leq K_1^b$, we have  $\max\{K_2^c, N-K_1^c-K_1^b\}\leq F_2^{c*}\leq  N-K_1^c $. Therefore, We can prove property (i) of Theorem~\ref{Thm:opt-prop}.

\subsection*{Proof of Property $(ii)$ of Theorem~\ref{Thm:opt-prop}}
First, we prove that there exists an optimal solution $(\mathcal F_1^{c*},\mathcal F_2^{c*},\mathbf T)$ to Problem~\ref{prob:opt-asymp-eq}, such that files in $\mathcal F_1^{c*} \cup \mathcal F_1^{b*}$ are consecutive. By Lemma~\ref{Lem:asym-perf} and $F_1^{b*}\leq K_1^b$ shown in the proof of Property$(i)$ , we have:
\small{\begin{align}\label{eqn:optimal-value-appendix-b2}
q_{\infty}^{*}=f_{1,K_1^c+F_1^{b*},\infty}\sum_{n\in\mathcal F_1^{c*}\cup \mathcal F_1^{b*}}a_n+\sum_{n\in\mathcal F_2^{c*}}a_nf_{2,K_2^c,\infty}(T_n^{*}).
\end{align}}\normalsize
Let $n_1 (n_2)$ denote the most (least) popular file in $\mathcal F_1^{c*} \cup \mathcal F_1^{b*}$. Suppose for any optimal solution $(\mathcal F_1^{c*}, \mathcal F_2^{c*}, \mathbf T^{*})$ to Probelm~\ref{prob:opt-asymp-eq}, files in $\mathcal F_1^{c*} \cup \mathcal F_1^{b*}$ are not consecutive, i.e., there exists $n_3\in\mathcal F_2^{c*}$ satisfying $n_1<n_3<n_2$. Now, we can construct a feasible solution $(\mathcal F_1^{c'},\mathcal F_2^{c'},\mathbf T')$ to Probelm~\ref{prob:opt-asymp-eq} where files in $\mathcal F_1^{c'} \cup \mathcal F_1^{b'}$ are consecutive as follows.

%of Problem~\ref{prob:opt-asymp}, in which files in $\mathcal F_1^{c'} \cup \mathcal F_1^{b'}$ are consecutive as follows:

%Assume that files in $\mathcal F_1^{c*} \cup \mathcal F_1^{b*}$ are not consecutive in any optimal solution of Problem~\ref{prob:opt-asymp}, i.e.,
%there exists $n_3\in\mathcal F_2^{c*}$ and $n_1<n_3<n_2$. We construct a feasible solution $(\mathcal F_1^{c'},\mathcal F_2^{c'},\mathbf T')$ as follows of Problem~\ref{prob:opt-asymp}, in which files in $\mathcal F_1^{c'} \cup \mathcal F_1^{b'}$ are consecutive as follows:

$\bullet$ If $f_{1,K_1^c+F_1^{b*},\infty}< f_{2,K_2^c,\infty}(T_{n_3}^{*})$, choose $\mathcal F_1^{c'}\cup \mathcal F_1^{b'}=\mathcal F_1^{c*}\cup \mathcal F_1^{b*}\cup\{n_3\} \setminus \{n_1\}$, $\mathcal F_2^{c'}=\mathcal F_2^{c*}\cup \{n_1\}\setminus \{n_3\}$, $T_{n_1}'=T_{n_3}^{*}$ and $T_n'=T_n^*$ for all $n \in \mathcal F_2^{c*}\setminus \{n_3\}$. By Lemma~\ref{Lem:asym-perf}, we have:
\small{\begin{align}\label{eqn:controdiction-thm2-b2-leq}
q_{\infty}\left(\mathcal F_1^{c'},\mathcal F_2^{c'},\mathbf T^{'}\right)-q_{\infty}^{*}=\left(a_{n_1}-a_{n_3}\right)\left(f_{2,K_2^c,\infty}(T_{n_3}^{*})-f_{1,K_1^c+F_1^{b*},\infty}\right)> 0,
\end{align}}\normalsize
where the inequality is due to $n_1<n_3$ and $f_{1,K_1^c+F_1^{b*},\infty}< f_{2,K_2^c,\infty}(T_{n_3}^{*})$.
%if $q_{\infty}\left(\mathcal F_1^{c'},\mathcal F_2^{c'},\mathbf T^{'}\right)-q_{\infty}^{*}>0$, $(\mathcal F_1^{c*}, \mathcal F_2^{c*}, \mathbf T^{*})$ is not an optimal solution which contradicts the assumption.

$\bullet$ If $f_{1,K_1^c+F_1^{b*},\infty}>f_{2,K_2^c,\infty}(T_{n_3}^{*})$, choose $\mathcal F_1^{c'}\cup \mathcal F_1^{b'}=\mathcal F_1^{c*}\cup \mathcal F_1^{b*}\cup\{n_3\} \setminus \{n_2\}$, $\mathcal F_2^{c'}=\mathcal F_2^{c*}\cup \{n_2\}\setminus \{n_3\}$, $T_{n_2}'=T_{n_3}^{*}$ and $T_n'=T_n^*$ for all $n \in \mathcal F_2^{c*}\setminus \{n_3\}$. By Lemma~\ref{Lem:asym-perf}, we have:
\small{\begin{align}\label{eqn:controdiction-thm2-b2-geq}
q_{\infty}\left(\mathcal F_1^{c'},\mathcal F_2^{c'},\mathbf T^{'}\right)-q_{\infty}^{*}=\left(a_{n_2}-a_{n_3}\right)\left(f_{2,K_2^c,\infty}(T_{n_3}^{*})-f_{1,K_1^c+F_1^{b*},\infty}\right)> 0,
\end{align}}\normalsize
where the inequality is due to $n_3<n_2$ and $f_{2,K_2^c,\infty}(T_{n_3}^{*})< f_{1,K_1^c+F_1^{b*},\infty}$.

$\bullet$ If $f_{1,K_1^c+F_1^{b*},\infty}= f_{2,K_2^c,\infty}(T_{n_3}^{*})$, choose $\mathcal F_1^{c'}\cup \mathcal F_1^{b'}=\mathcal F_1^{c*}\cup \mathcal F_1^{b*}\cup\{n_3\} \setminus \{n_1\}$, $\mathcal F_2^{c'}=\mathcal F_2^{c*}\cup \{n_1\}\setminus \{n_3\}$, $T_{n_1}'=T_{n_3}^{*}$ and $T_n'=T_n^*$ for all $n \in \mathcal F_2^{c*}\setminus \{n_3\}$. By Lemma~\ref{Lem:asym-perf}, we have:
\small{\begin{align}\label{eqn:controdiction-thm2-b2-eq}
q_{\infty}\left(\mathcal F_1^{c'},\mathcal F_2^{c'},\mathbf T^{'}\right)-q_{\infty}^{*}=\left(a_{n_1}-a_{n_3}\right)\left(f_{2,K_2^c,\infty}(T_{n_3}^{*})-f_{1,K_1^c+F_1^{b*},\infty}\right)= 0,
\end{align}}\normalsize

%In~(\ref{eqn:controdiction-thm2-b2-leq}) or~(\ref{eqn:controdiction-thm2-b2-geq}), if $q_{\infty}\left(\mathcal F_1^{c'},\mathcal F_2^{c'},\mathbf T^{'}\right)-q_{\infty}^{*}>0$, $(\mathcal F_1^{c*}, \mathcal F_2^{c*}, \mathbf T^{*})$ is not an optimal solution, which contradicts the assumption. In~(\ref{eqn:controdiction-thm2-b2-leq}) or~(\ref{eqn:controdiction-thm2-b2-geq}), if $q_{\infty}\left(\mathcal F_1^{c'},\mathcal F_2^{c'},\mathbf T^{'}\right)-q_{\infty}^{*}=0$, we can always construct an optimal solution $\left(\mathcal F_1^{c'},\mathcal F_2^{c'},\mathbf T^{'}\right)$, where files in $\mathcal F_1^{c'} \cup \mathcal F_1^{b'}$ are consecutive.
By (\ref{eqn:controdiction-thm2-b2-leq}), (\ref{eqn:controdiction-thm2-b2-geq}) and (\ref{eqn:controdiction-thm2-b2-eq}), we know that if $f_{1,K_1^c+F_1^{b*},\infty}< f_{2,K_2^c,\infty}(T_{n_3}^{*})$ or $f_{1,K_1^c+F_1^{b*},\infty}>f_{2,K_2^c,\infty}(T_{n_3}^{*})$, $(\mathcal F_1^{c*}, \mathcal F_2^{c*}, \mathbf T^{*})$ is not an optimal solution, which contradicts the assumption; and if $f_{1,K_1^c+F_1^{b*},\infty}= f_{2,K_2^c,\infty}(T_{n_3}^{*})$, we can always construct an optimal solution $\left(\mathcal F_1^{c'},\mathcal F_2^{c'},\mathbf T^{'}\right)$, satisfying that files in $\mathcal F_1^{c'} \cup \mathcal F_1^{b'}$ are consecutive.
Thus,  we can prove that there exists an optimal solution $(\mathcal F_1^{c*},\mathcal F_2^{c*},\mathbf T)$ to Problem~\ref{prob:opt-asymp-eq}, such that files in $\mathcal F_1^{c*} \cup \mathcal F_1^{b*}$ are consecutive.

In addition, by~(\ref{eqn:optimal-value-appendix-b2}), we know that whether file $n\in \mathcal F_1^{c*}\cup \mathcal F_1^{b*}$ belongs to $\mathcal F_1^{c*}$ or $\mathcal F_1^{b*}$ makes no difference in the optimal successful transmission probability. Therefore, we can prove the property (ii) of Theorem~\ref{Thm:opt-prop}.

\section*{Appendix H: Proof of Lemma~\ref{Lem:opt-prop}}

\subsection*{Proof of Property $(i)$ of Lemma~\ref{Lem:opt-prop}}

We prove that if $f_{1,K_1^c+K_1^b,\infty}>f_{2,K_2^c,\infty}(1)$, the most popular file $n=1$ belongs to $\mathcal F_1^{c*} \cup \mathcal F_1^{b*}$ for any optimal solution $\left(\mathcal F_1^{c*}, \mathcal F_2^{c*}, \mathbf T^*\right)$ to Problem~\ref{prob:opt-asymp-eq}. Suppose that there exists an optimal solution $\left(\mathcal F_1^{c*}, \mathcal F_2^{c*}, \mathbf T^*\right)$ to Problem~\ref{prob:opt-asymp-eq}, such that the most popular file $n=1$ belongs to $\mathcal F_2^{c*}$. Let $n_2$ denote a file in $\mathcal F_1^{c*} \cup\mathcal F_1^{b*}$. Now, we can construct a feasible solution $(\mathcal F_1^{c'},\mathcal F_2^{c'},\mathbf T')$ to Probelm~\ref{prob:opt-asymp-eq}, where $\mathcal F_1^{c'}\cup \mathcal F_1^{b'}=\mathcal F_1^{c*}\cup \mathcal F_1^{b*}\cup \{1\}\setminus \{n_2\}$, $\mathcal F_2^{c'}=\mathcal F_2^{c*}\cup \{n_2\}\setminus \{1\}$, $T_{n_2}'=T_{1}^*$ and $T_n'=T_n^*$ for all $n\in \mathcal F_2^{c*}\setminus \{1\}$. By Lemma~\ref{Lem:asym-perf}, we have:
\small{\begin{align}\label{eq:contradiction_appendix_E_2}
q_{\infty}\left(\mathcal F_1^{c'},\mathcal F_2^{c'},\mathbf  T^{'}\right)-q_{\infty}^{*}=\left(a_{1}-a_{n_2}\right)\left(f_{1,K_1^c+\min\{K_1^b,F_1^{b*}\},\infty}-f_{2,K_2^c,\infty}(T_{1}^*)\right).
\end{align}}\normalsize
Since $a_{1}>a_{n_2}$ and $f_{1,K_1^c+\min\{K_1^b,F_1^{b*}\},\infty}\geq f_{1,K_1^c+K_1^b,\infty}>f_{2,K_2^c,\infty}(1)\geq f_{2,K_2^c,\infty}(T_{1}^{*})$, we have $q_{\infty}\left(\mathcal F_1^{c'},\mathcal F_2^{c'},\mathbf  T^{'}\right)-q_{\infty}^{*}>0$. Thus, $(\mathcal F_1^{c*}, \mathcal F_2^{c*}, \mathbf T^*)$ is not an optimal solution to Problem~\ref{prob:opt-asymp-eq}, which contradicts the assumption. By contradiction,  we prove that if $f_{1,K_1^c+K_1^b,\infty}>f_{2,K_2^c,\infty}(1)$, the most popular file $n=1$ belongs to $\mathcal F_1^{c*} \cup \mathcal F_1^{b*}$ for any optimal solution $\left(\mathcal F_1^{c*}, \mathcal F_2^{c*}, \mathbf T^*\right)$ to Problem~\ref{prob:opt-asymp-eq}, and hence $n_1^c$ in Theorem~\ref{Thm:opt-prop} (ii) satisfies $n_1^c=1$.

\subsection*{Proof of Property $(ii)$ of Lemma~\ref{Lem:opt-prop}}

We prove that if $f_{1,K_1^c,\infty}<f_{2,K_2^c,\infty}(\frac{K_2^c}{N-K_1^c})$, the most popular file $n=1$ belongs to $\mathcal F_2^{c*}$ for any optimal solution $\left(\mathcal F_1^{c*}, \mathcal F_2^{c*}, \mathbf T^*\right)$ to Problem~\ref{prob:opt-asymp-eq}. Suppose that there exists an optimal solution $\left(\mathcal F_1^{c*}, \mathcal F_2^{c*}, \mathbf T^*\right)$ to Problem~\ref{prob:opt-asymp-eq}, such that file $n=1$ belongs to $\mathcal F_1^{c*}\cup \mathcal F_1^{b*}$. Let $n_2$ denote the most popular file in $\mathcal F_2^{c*}$. Based on Lemma~\ref{Lem:mono-general-asym}, we have $T_{n_2}^*\geq T_n^*$ for any $n\in \mathcal F_2^{c*}\setminus \{n_2\}$, and hence $T_{n_2}^*\geq \frac{K_2^c}{N-K_1^c}$. Now, we can construct a feasible solution $(\mathcal F_1^{c'},\mathcal F_2^{c'},\mathbf T')$ to Probelm~\ref{prob:opt-asymp-eq}, where $\mathcal F_1^{c'}\cup \mathcal F_1^{b'}=\mathcal F_1^{c*}\cup \mathcal F_1^{b*}\cup \{n_2\}\setminus \{1\}$, $\mathcal F_2^{c'}=\mathcal F_2^{c*}\cup \{1\}\setminus \{n_2\}$, $T_1'=T_{n_2}^*$ and $T_n'=T_n^*$ for all $n\in \mathcal F_2^{c*}\setminus \{n_2\}$. By Lemma~\ref{Lem:asym-perf}, we have:
\small{\begin{align}\label{eq:contradiction_appendix_E_3}
q_{\infty}\left(\mathcal F_1^{c'},\mathcal F_2^{c'},\mathbf  T^{'}\right)-q_{\infty}^{*}=\left(a_{1}-a_{n_2}\right)\left(f_{2,K_2^c,\infty}(T_{n_2}^*)-f_{1,K_1^c+\min\{K_1^b+\mathcal F_1^{b*}\},\infty}\right).
\end{align}}\normalsize
Since $a_1> a_{n_2}$ and  $f_{1,K_1^c+\min\{K_1^b+\mathcal F_1^{b*}\},\infty}\leq f_{1,K_1^c,\infty}<f_{2,K_2^c,\infty}(\frac{K_2^c}{N-K_1^c})\leq f_{2,K_2^c,\infty}(T_{n_2}^*)$, we have  $q_{\infty}\left(\mathcal F_1^{c'},\mathcal F_2^{c'},\mathbf  T^{'}\right)-q_{\infty}^{*}>0$. Thus, $\left(\mathcal F_1^{c*}, \mathcal F_2^{c*}, \mathbf T^*\right)$ is not an optimal solution to Problem~\ref{prob:opt-asymp-eq}, which contradicts the assumption. Therefore, we prove that if $f_{1,K_1^c,\infty}<f_{2,K_2^c,\infty}(\frac{K_2^c}{N-K_1^c})$, the most popular file $n=1$ belongs to $\mathcal F_2^{c*}$ for any optimal solution $\left(\mathcal F_1^{c*}, \mathcal F_2^{c*}, \mathbf T^*\right)$ to Problem~\ref{prob:opt-asymp-eq}, and hence $n_1^c$ in Theorem~\ref{Thm:opt-prop} (ii) satisfies $n_1^c\geq 2$.

%\bibliographystyle{IEEEtran}
%\bibliography{IEEEabrv,GLOBECOM15}

\begin{thebibliography}{10}
\providecommand{\url}[1]{#1}
\csname url@samestyle\endcsname
\providecommand{\newblock}{\relax}
\providecommand{\bibinfo}[2]{#2}
\providecommand{\BIBentrySTDinterwordspacing}{\spaceskip=0pt\relax}
\providecommand{\BIBentryALTinterwordstretchfactor}{4}
\providecommand{\BIBentryALTinterwordspacing}{\spaceskip=\fontdimen2\font plus
\BIBentryALTinterwordstretchfactor\fontdimen3\font minus
  \fontdimen4\font\relax}
\providecommand{\BIBforeignlanguage}[2]{{%
\expandafter\ifx\csname l@#1\endcsname\relax
\typeout{** WARNING: IEEEtran.bst: No hyphenation pattern has been}%
\typeout{** loaded for the language `#1'. Using the pattern for}%
\typeout{** the default language instead.}%
\else
\language=\csname l@#1\endcsname
\fi
#2}}
\providecommand{\BIBdecl}{\relax}
\BIBdecl

\bibitem{HoadleyWC12}
J.~Hoadley and P.~Maveddat, ``Enabling small cell deployment with hetnet,''
  \emph{IEEE Wireless Communications}, vol.~19, no.~2, pp. 4--5, April 2012.

\bibitem{AndrewsComMag13}
J.~G. Andrews, ``Seven ways that hetnets are a cellular paradigm shift,''
  \emph{IEEE Communications Magazine}, vol.~51, no.~3, pp. 136--144, March
  2013.

\bibitem{CMag14Chen}
X.~Wang, M.~Chen, T.~Taleb, A.~Ksentini, and V.~Leung, ``Cache in the air:
  exploiting content caching and delivery techniques for 5g systems,''
  \emph{Communications Magazine, IEEE}, vol.~52, no.~2, pp. 131--139, February
  2014.

\bibitem{Sarkissian4GLTE11}
H.~Sarkissian, ``The business case for caching in 4g lte networks,''
  \emph{Wireless 20|20}, vol. 20120, 2012.

\bibitem{cachingmimoLiu15}
A.~Liu and V.~Lau, ``Exploiting base station caching in {MIMO} cellular
  networks: Opportunistic cooperation for video streaming,'' \emph{{IEEE}
  Trans. Signal Process.}, vol.~63, no.~1, pp. 57--69, Jan 2015.

\bibitem{Shanmugam13}
K.~Shanmugam, N.~Golrezaei, A.~Dimakis, A.~Molisch, and G.~Caire,
  ``Femtocaching: Wireless content delivery through distributed caching
  helpers,'' \emph{Information Theory, IEEE Transactions on}, vol.~59, no.~12,
  pp. 8402--8413, Dec 2013.

\bibitem{LiTWC15}
J.~Li, Y.~Chen, Z.~Lin, W.~Chen, B.~Vucetic, and L.~Hanzo, ``Distributed
  caching for data dissemination in the downlink of heterogeneous networks,''
  \emph{IEEE Transactions on Communications}, vol.~63, no.~10, pp. 3553--3568,
  Oct 2015.

\bibitem{BioglioGC15}
V.~Bioglio, F.~Gabry, and I.~Land, ``Optimizing {MDS} codes for caching at the
  edge,'' in \emph{2015 IEEE Global Communications Conference (GLOBECOM)}, Dec
  2015, pp. 1--6.

\bibitem{EURASIP15Debbah}
E.~Bastu$\rm\breve{g}$, M.~Bennis, M.~Kountouris, and M.~Debbah,
  ``Cache-enabled small cell networks: Modeling and tradeoffs,'' \emph{EURASIP
  Journal on Wireless Communications and Networking}, 2015.

\bibitem{LiuYangICC16}
\BIBentryALTinterwordspacing
D.~Liu and C.~Yang, ``Cache-enabled heterogeneous cellular networks: Comparison
  and tradeoffs,'' in \emph{IEEE Int. Conf. on Commun. (ICC)}, Kuala Lumpur,
  Malaysia, June 2016. [Online]. Available:
  \url{http://arxiv.org/abs/1602.08255}
\BIBentrySTDinterwordspacing

\bibitem{Yang16}
C.~Yang, Y.~Yao, Z.~Chen, and B.~Xia, ``Analysis on cache-enabled wireless
  heterogeneous networks,'' \emph{Wireless Communications, IEEE Transactions
  on}, vol.~15, no.~1, pp. 131--145, Jan 2016.

\bibitem{QuekTWC16}
\BIBentryALTinterwordspacing
Z.~Chen, J.~Lee, T.~Q.~S. Quek, and M.~Kountouris, ``Cooperative caching and
  transmission design in cluster-centric small cell networks,'' \emph{CoRR},
  vol. abs/1601.00321, 2016. [Online]. Available:
  \url{http://arxiv.org/abs/1601.00321}
\BIBentrySTDinterwordspacing

\bibitem{TamoorComLett16}
S.~T. ul~Hassan, M.~Bennis, P.~H.~J. Nardelli, and M.~Latva-aho, ``Caching in
  wireless small cell networks: A storage-bandwidth tradeoff,'' \emph{IEEE
  Communications Letters}, vol.~PP, no.~99, pp. 1--1, 2016.

\bibitem{DBLP:journals/corr/Tamoor-ul-Hassan15}
\BIBentryALTinterwordspacing
S.~Tamoor{-}ul{-}Hassan, M.~Bennis, P.~H.~J. Nardelli, and M.~Latva{-}aho,
  ``Modeling and analysis of content caching in wireless small cell networks,''
  \emph{CoRR}, vol. abs/1507.00182, 2015. [Online]. Available:
  \url{http://arxiv.org/abs/1507.00182}
\BIBentrySTDinterwordspacing

\bibitem{eMBMS}
D.~Lecompte and F.~Gabin, ``Evolved multimedia broadcast/multicast service
  ({eMBMS}) in {LTE}-advanced: overview and {Rel}-11 enhancements,''
  \emph{{IEEE} Commun. Mag.}, vol.~50, no.~11, pp. 68--74, 2012.

\bibitem{WCNC14Tassiulas}
K.~Poularakis, G.~Iosifidis, V.~Sourlas, and L.~Tassiulas, ``Multicast-aware
  caching for small cell networks,'' in \emph{IEEE WCNC}, April 2014, pp.
  2300--2305.

\bibitem{ZhouTWC15}
\BIBentryALTinterwordspacing
B.~Zhou, Y.~Cui, and M.~Tao, ``Stochastic content-centric multicast scheduling
  for cache-enabled heterogeneous cellular networks,'' \emph{submitted to
  Wireless Communications, IEEE Transactions on}, vol. abs/1509.06611, 2015.
  [Online]. Available: \url{http://arxiv.org/abs/1509.06611}
\BIBentrySTDinterwordspacing

\bibitem{ICC15Giovanidis}
B.~Blaszczyszyn and A.~Giovanidis, ``Optimal geographic caching in cellular
  networks,'' in \emph{IEEE Int. Conf. on Commun. (ICC)}, London, United
  Kingdom, June 2015, pp. 1--6.

\bibitem{DBLP:journals/corr/BharathN15}
\BIBentryALTinterwordspacing
B.~{N. Bharath} and K.~{G. Nagananda}, ``Caching with unknown popularity
  profiles in small cell networks,'' \emph{CoRR}, vol. abs/1504.03632, 2015.
  [Online]. Available: \url{http://arxiv.org/abs/1504.03632}
\BIBentrySTDinterwordspacing

\bibitem{Altman13}
\BIBentryALTinterwordspacing
E.~Altman, K.~Avrachenkov, and J.~Goseling, ``Coding for caches in the plane,''
  \emph{CoRR}, vol. abs/1309.0604, 2013. [Online]. Available:
  \url{http://arxiv.org/abs/1309.0604}
\BIBentrySTDinterwordspacing

\bibitem{arXivSGCaching15}
\BIBentryALTinterwordspacing
Y.~Cui, D.~Jiang, and Y.~Wu, ``Analysis and optimization of caching and
  multicasting in large-scale cache-enabled wireless networks,'' \emph{CoRR},
  vol. abs/1512.06176, 2015. [Online]. Available:
  \url{http://arxiv.org/abs/1512.06176}
\BIBentrySTDinterwordspacing

\bibitem{Andrews11}
J.~Andrews, F.~Baccelli, and R.~Ganti, ``A tractable approach to coverage and
  rate in cellular networks,'' \emph{Communications, IEEE Transactions on},
  vol.~59, no.~11, pp. 3122--3134, November 2011.

\bibitem{WCOM13Andrews}
S.~Singh, H.~S. Dhillon, and J.~G. Andrews, ``Offloading in heterogeneous
  networks: modeling, analysis, and design insights,'' \emph{{IEEE} Trans.
  Wireless Commun.}, vol.~12, no.~5, pp. 2484--2497, March 2013.

\bibitem{AndrewsTWCOffloading14}
S.~Singh and J.~Andrews, ``Joint resource partitioning and offloading in
  heterogeneous cellular networks,'' \emph{Wireless Communications, IEEE
  Transactions on}, vol.~13, no.~2, pp. 888--901, Feb 2014.

\bibitem{SGcellsize13}
S.~M. Yu and S.-L. Kim, ``Downlink capacity and base station density in
  cellular networks,'' in \emph{Modeling Optimization in Mobile, Ad Hoc
  Wireless Networks (WiOpt), 2013 11th International Symposium on}, May 2013,
  pp. 119--124.

\bibitem{Bertsekasbooknonlinear:99}
D.~P. Bertsekas, \emph{Nonlinear Programming}, 2nd~ed.\hskip 1em plus 0.5em
  minus 0.4em\relax Belmont, MA: Athena Scientific, 1999.

\bibitem{FTNhaenggi09}
M.~Haenggi and R.~K. Ganti, ``Interference in large wireless networks,''
  \emph{Foundations and Trends in Networking}, vol.~3, no.~2, pp. 127--248,
  2009.

\end{thebibliography}

% Generated by IEEEtran.bst, version: 1.13 (2008/09/30)

\end{document}